\documentclass[12pt]{article}
\usepackage[english]{babel}
\usepackage[letterpaper,top=2cm,bottom=2cm,left=3cm,right=3cm,marginparwidth=1.75cm]{geometry}
\usepackage{authblk}
\usepackage{amsmath}
\usepackage{graphicx}
\usepackage{amsfonts}
\usepackage{amssymb}
\usepackage{amsthm}
\usepackage{cancel}
\usepackage{booktabs}
\usepackage{multirow}
\usepackage{makecell}
\usepackage{bm}
\usepackage[colorlinks=true, allcolors=blue]{hyperref}
\usepackage{setspace}
\usepackage[most]{tcolorbox}
\usepackage{orcidlink}
\usepackage{chngcntr}
\usepackage{enumitem}
\usepackage{natbib}
\usepackage{listings}
\usepackage{tikz}
\usepackage{pgfplots}
\usepackage{minted}
\usepackage[toc,page]{appendix}
\usepackage[ruled]{algorithm2e}
\usepackage{appendix} 
\usepackage{svg}
\usepackage{subcaption}
\newtheorem{proposition}{Proposition}
\newtheorem{lemma}{Lemma}

\usepackage[most]{tcolorbox}  
\usepackage{xcolor} 
\definecolor{boxbg}{RGB}{245,245,245}   
\definecolor{boxborder}{RGB}{180,180,180} 
\tcbset{
    sectionintro/.style={
        colback=boxbg,        
        colframe=boxborder,   
        fonttitle=\bfseries,  
        fontupper=\scriptsize,
        coltitle=black,       
        boxrule=0.8pt,        
        arc=4mm,              
        left=2mm,
        right=2mm,
        top=1mm,
        bottom=1mm,
        width=\textwidth,
        boxsep=3pt,
    }
}

\lstdefinestyle{rstyle}{
    language=R,
    backgroundcolor=\color{gray!10},
    basicstyle=\small\ttfamily,
    breaklines=true,
    numbers=none,
    frame=none,
    showstringspaces=false,
    columns=fullflexible,
    keepspaces=true,
}

\graphicspath{ {./img/} }

\newcommand\ci{\perp\!\!\!\perp} 
\newcommand{\determinant}[1]{\left\lvert #1 \right\rvert}

\def\keywordname{{\bfseries \emph Keywords}}%
\def\keywords#1{\par\addvspace\medskipamount{\rightskip=0pt plus1cm
\def\and{\ifhmode\unskip\nobreak\fi\ $\cdot$
}\noindent\keywordname\enspace\ignorespaces#1\par}}
\def\correspondencename{{\bfseries \emph Corresponding author:}}%
\def\correspondence#1{\par\addvspace\medskipamount{\rightskip=0pt plus1cm
\def\and{\ifhmode\unskip\nobreak\fi\ $\cdot$
}\noindent\correspondencename\enspace\ignorespaces#1\par}}

\title{What is your Prior Worth? Effective Sample Size and Sample Size Planning for Gaussian Graphical Models}
\author[1]{Giuseppe Arena}
\author[1]{Lourens Waldorp}
\author[1]{Maarten Marsman}
\affil[1]{Department of Psychology, University of Amsterdam}

\begin{document}
\maketitle
\vspace{-2em}
\begin{center}
\textbf{This manuscript has not yet been peer-reviewed.} 
\end{center}
\correspondence{Giuseppe Arena\orcidlink{0000-0001-5204-3326}, Department of Psychology, University of Amsterdam -- Nieuwe Achtergracht 129-B, PO Box 15906, 1001 NK Amsterdam, The Netherlands -- E-mail:  \texttt{g.arena@uva.nl}. }
\keywords{Effective Sample Size \and Gaussian Graphical Models \and Sample Size Planning \and Prior Elicitation}

\begin{abstract}
In Bayesian analysis, the prior effective sample size (ESS) expresses the information carried by a prior distribution in units of observations, quantifying how much independent information the prospective data must provide to outweigh an informative prior elicited from a previous study. For network models such as Gaussian graphical models (GGMs), the prior ESS is not straightforward to compute. The Wishart and G-Wishart priors induce dependence among the entries of the precision matrix, and their informativeness has never been expressed in an interpretable, observation-equivalent unit. As a result, researchers eliciting an informative prior for a GGM have had no principled basis for sample size planning. In this paper, we close this gap by formalizing a pre-data ESS for GGMs under the Wishart and G-Wishart priors. We adapt five ESS estimators to the GGM setting and compute each through two aggregation schemes: a global ESS measure based on a determinant ratio, and a parameterwise version based on a Cholesky decomposition. Building on these measures, we introduce two complementary planning strategies: the Data-to-Prior Information Ratio (DPIR), which determines the sample size at which the data dominate the prior, and a GGM extension of Bayes Factor Design Analysis (BFDA), which determines the sample size required for conclusive edge-based evidence. Simulation studies show that the two procedures target complementary design goals and that the ESS estimators differ systematically in their sensitivity to network structure and geometry. We conclude by outlining extensions to other graphical models, including time-dependent variants, as well as to matrix-variate mixture priors.
\end{abstract}


\section{Introduction}

Gaussian graphical models (GGMs) provide a probabilistic framework to represent the conditional independence structure among a set of variables through the sparsity of the precision matrix $\bm{\Theta}$ \citep{Lauritzen1996,Dempster1972}. By encoding pairwise partial correlations $\rho_{ij}$ as edge weights in an undirected graph $G$, GGMs have become a standard tool for network inference in the behavioral, social, psychological, and biomedical sciences~\citep{Kramer2009,Borsboom2013,Costantini2015,Isvoranu2022}. The direct relationship between zero partial correlations and conditional independence makes GGMs suitable for settings where the dependence structure itself is the scientific target, and where study design decisions, including sample size, directly determine whether that structure can be reliably recovered.

Bayesian inference for GGMs requires the specification of a prior distribution over $\bm{\Theta}$ and, when the graph $G$ is unknown, over the space of graphs itself. The Wishart distribution $W(\nu,\bm{\Psi})$ and its graph-constrained generalization, the \textit{G}-Wishart $W_{G}(\nu,\bm{\Psi})$, are the conjugate priors for the precision matrix under a Gaussian likelihood \citep{Giudici1995,Roverato2002,Atay2005,Lenkoski2013}, and constitute the standard prior choice in the Bayesian GGM literature \citep[e.g.,][]{Mohammadi2015,Mohammadi2019,Williams2020a,Williams2020b}. Both distributions are governed by two hyperparameters: (i) the degrees of freedom $\nu$, which controls the concentration of the prior around $\bm{\Psi}$ and (ii) the scale matrix $\bm{\Psi}$, which encodes the prior beliefs about the precision structure. Their specification has direct consequences in posterior inference: a large $\nu$ produces an informative prior that may dominate the likelihood at moderate sample sizes, whereas a small $\nu$ yields a diffuse prior that offers little regularization and can produce unstable edge estimates. Despite this sensitivity, the informativeness of the Wishart or \textit{G}-Wishart prior has never been expressed in an interpretable, observation-equivalent unit. This is not merely a theoretical gap: it has direct practical implications. Without knowing how many observations the prior is worth, a researcher eliciting an informative prior for a GGM has no principled basis for determining how large a sample is needed for the data to dominate the prior, or for designing a study at all.

A natural framework for characterizing the informativeness of a prior distribution is the effective sample size (ESS), which expresses the information content of the prior in units of observations. The concept originates with \citet{Clarke1996}, who represented an informative prior as the posterior arising from updating a reference prior with a hypothetical dataset, and was substantially developed by \citet{Clarke2006} and \citet[MTM;][]{Morita2008,Morita2010}, whose Fisher-information curvature approach became the dominant computational method. Several alternative ESS estimators have since been proposed. The variance-ratio (VR) and precision-ratio (PR) estimators generalize the conjugate ESS intuition to non-conjugate scenarios by relating prior to expected uncertainty of a sample-based estimator \citep{Neuenschwander2020}. The Pennello-Thompson (PT) estimator, formalized by \citet{Neuenschwander2020} 
as a simplified variant of MTM, approximates the expected posterior information at the prior mode\footnote{Following \citet{Neuenschwander2020}, who derive PT as a deterministic approximation to the simulation-based ESS procedure of \citet{Pennello2007}.}. The expected local information ratio (ELIR) averages the ratio of prior curvature to Fisher information under the prior predictive distribution \citep{Neuenschwander2020}. Applications have been developed primarily in univariate conjugate settings, including clinical trial design \citep{Morita2008, Morita2010}, and meta-analytic priors \citep{Neuenschwander2020}. Despite this, no ESS estimator has been derived or evaluated for the Wishart or \textit{G}-Wishart distributions. The multivariate, constrained structure of these priors makes the direct application of existing methods non-trivial, and the absence of a GGM specific ESS has left a critical gap in Bayesian study design for GGMs.

Closing this gap requires an ESS method that is operable before data collection, since sample size planning is inherently a pre-data analysis. Pre-data ESS methods characterize prior informativeness on the basis of $\pi(\bm{\Theta})$ and the model $f(\mathbf{X} ; \bm{\Theta})$ alone, expressing the informative prior as equivalent to a hypothetical number of observations~\citep{Clarke1996,Clarke2006,Morita2008,Neuenschwander2020}. A parallel strand of the ESS literature has developed post-data methods which instead condition on the observed data to quantify how much the prior actually contributed to the posterior analysis \citep{Reimherr2021, Wiesenfarth2019, Jones2022}. These methods are sensitive to prior-data conflict and can return negative ESS values when the prior and the likelihood are discordant. While post-data methods are valuable for posterior diagnostics, they cannot serve as inputs to sample size planning: the number of observations to collect must be determined before any data are collected. In this paper, we operate exclusively within the pre-data paradigm.

The absence of a GGM-specific ESS is particularly problematic for Bayes Factor Design Analysis \citep[BFDA;][]{Schnbrodt2017,Stefan2019}, a principled approach to sample size planning that determines the minimum $n^{\star}$ required to achieve a target level of evidence for edge inclusion or exclusion. BFDA relies on the prior predictive distribution $\int f(\mathbf{X} ; \bm{\Theta})\, \pi(\bm{\Theta})\, d\bm{\Theta}$, and when an informative prior is elicited, the design is only meaningful if the prior informativeness, that is the ESS, is known. Without ESS, the researcher cannot assess whether the informative prior overwhelms the data, nor determine how many observations are needed for the evidence to reflect the data rather than the prior. To our knowledge, no solution to this problem exists in the GGM literature.

This paper fills this gap. We make three contributions to the Bayesian GGM literature. First, we formalize pre-data ESS for GGMs under the Wishart and \textit{G}-Wishart priors, deriving a global ESS via the determinant ratio approach \citep{Vats2019} and a parameterwise ESS via Cholesky decomposition. The predictive consistency of the candidate ESS estimators --- VR, PR, MTM, PT, and ELIR --- in the sense of \citet{Neuenschwander2020} is evaluated through analytical and simulation-based results in Appendix~\ref{appendix:pcc_proofs}. Second, we introduce the Data-to-Prior Information Ratio (DPIR) procedure, which compares the average information provided by the data relative to that of the prior and uses the ESS to determine the minimum sample size $n^{\star}$ such that the information of the data exceeds that of the elicited prior. DPIR is complemented by BFDA \citep{Schnbrodt2017,Stefan2019}, together providing two principled complementary routes to sample size planning for Bayesian GGMs. The prior predictive distribution $\int f(\mathbf{X} ; \bm{\Theta})\, \pi(\bm{\Theta})\, d\bm{\Theta}$ is the natural object unifying this pipeline, as pre-data ESS, DPIR, and BFDA are all defined with respect to it. Third, we provide an open-source \texttt{R} package, \texttt{designbgm}~\citep{designbgm}, implementing all estimators and the full DPIR--BFDA pipeline.

The remainder of the paper is organized as follows. In Section \ref{sec:ggm_and_priors}, we introduce the GGM likelihood and the Wishart and \textit{G}-Wishart prior distributions. In Section \ref{sec:ess}, we introduce the global and parameterwise prior ESS for GGMs and adapt five pre-data ESS estimators (VR, PR, MTM, PT, and ELIR) to the precision matrix. In Section \ref{sec:sample_size_planning}, we present two complementary sample size planning strategies, the DPIR and BFDA. In Section \ref{sec:simulation_studies}, we report two simulation studies, the first examining the behavior of the prior ESS estimators and the second validating the planning strategies and assessing their sensitivity to prior misspecification, across a range of graph structures and prior specifications. Finally, in Section \ref{sec:discussion}, we conclude with a discussion of limitations and directions for future work.

\section{The Gaussian Graphical Model}
\label{sec:ggm_and_priors}

\subsection{Model and Likelihood}

In a Gaussian Graphical Model (GGM), graph theory and multivariate statistics are combined to represent the conditional dependence structure among a set of variables. Let $G = \left(V,E\right)$ be an undirected graph, where the $V = \left\lbrace 1,\ldots, p \right\rbrace$ is a set of $p$ nodes corresponding to the variables of interest and $E$ is a set of edges encoding conditional dependence between pairs of variables. The absence of an edge $(i,j) \notin E$ indicates that variables $i$ and $j$ are conditional independent given all remaining variables, written $i \ci j \ | \ V_{\setminus {} (i,j)}$.

A GGM assumes that the data are a realization of a $p$-variate Gaussian distribution. Let $\mathbf{X} = (\mathbf{x}_1,\ldots,\mathbf{x}_n)$ denotes a sample of size $n$, with
\[
\mathbf{x}_i\overset{\text{iid}}{\sim}\mathcal{N}_p\left(\bm{0},\bm{\Sigma}\right)
\]
where variables are mean-centered, with covariance matrix $\bm{\Sigma}$ and precision matrix $\bm{\Theta} = \bm{\Sigma}^{-1}$. The graph structure is encoded directly in $\bm{\Theta}$: an entry $\bm{\Theta}_{ij} = 0$ if and only if $i \ci j \ | \ V_{\setminus {} (i,j)}$. Therefore, learning the sparsity of $\bm{\Theta}$ is equivalent to identifying the edge set of $G$. 

Since variables are mean-centered, $S = \frac{1}{n}\sum_{k=1}^{n}\mathbf{x}_k\mathbf{x}_k^{\top}$ coincides with the sample covariance matrix. The likelihood of the GGM is

\[ f(\mathbf{X};\bm{\Theta}) =(2\pi)^{-np/2} |\bm{\Theta}|^{n/2}\exp{\left\lbrace-\frac{n}{2}\text{tr}\left(S\bm{\Theta}\right)\right\rbrace}\]

where $\bm{\Theta}$ is the parameter of interest: its sparsity structure identifies the graph, and its non-zero entries quantify the strength of conditional dependence between variables. The score function, Hessian and Fisher information of the GGM are derived in Appendix \ref{appendix:ggm_derivatives}.
The partial correlation between variables $i$ and $j$ given all remaining variables is $\rho_{ij} = -\theta_{ij}/\sqrt{\theta_{ii}\theta_{jj}}$, which lies in $(-1,1)$ and provides a standardized measure of conditional dependence.

Bayesian inference on $\bm{\Theta}$ combines the likelihood with a prior distribution $\pi(\bm{\Theta})$, resulting in the posterior 
\[
\pi(\bm{\Theta} \mid \mathbf{X}) \propto f(\mathbf{X};\bm{\Theta})\pi(\bm{\Theta})
\]
The choice of the prior reflects the assumptions about the graph structure, in particular, whether the graph is assumed to be complete or sparse. In the next Section,  we introduce the prior distributions used in this paper.

\subsection{Prior Distributions: Wishart and \textit{G}-Wishart}
The choice of prior distribution for $\bm{\Theta}$ reflects assumptions about the graph structure and controls the amount of prior information entering the analysis. We consider two prior distributions. 

The Wishart distribution, $\bm{\Theta}\sim W(\nu,\bm{\Psi})$, is the conjugate prior for the precision matrix in a GGM with a complete graph. The scale matrix $\bm{\Psi}$ encodes prior beliefs about the structure and magnitude of $\bm{\Theta}$, while the degrees of freedom $\nu$ determine the prior concentration. Larger $\nu$ results in a tighter prior around $\bm{\Psi}$, whereas smaller $\nu$ produces a more diffuse prior.

The \textit{G}-Wishart distribution, $\bm{\Theta}\sim W_G(\nu,\bm{\Psi}_G)$, extends the Wishart to sparse graphs by restricting the support to precision matrices consistent with the zero structure of $G$ \citep{Roverato2000}. The graph $G$ determines which entries of $\bm{\Theta}$ are constrained to zero, $\bm{\Psi}_G$ encodes prior beliefs about the remaining free entries, and $\nu$ plays the same role as in the Wishart. Throughout this paper, we parametrize the \textit{G}-Wishart in terms of $\nu$ to align with the Wishart notation, following the reparametrization $\nu = \delta + p - 1$ of \citet{Roverato2000}, where $\delta>0$ is the degrees of freedom parameter in the original formulation, so that $\nu>p-1$. Under this convention, the Wishart and \textit{G}-Wishart share the same interpretation of $\nu$, and results derived under one prior carry over throughout.

When prior information is available from a study of size $n_0$, the hyperparameters $(\nu, \bm{\Psi}^{*})$ are elicited by setting $\nu = n_0$ and $\bm{\Psi}^{*} = \bm{\Psi} / \nu$, where $\bm{\Psi}$ encodes the researcher's prior belief about the precision matrix, set equal to the estimate from the prior study. For the Wishart $\mathbb{E}[\bm{\Theta}] = \nu \bm{\Psi}^{*} = \bm{\Psi}$, while for the \textit{G}-Wishart, elicited analogously with $\bm{\Psi}^{*}_G = \bm{\Psi}_G / \nu$, $\mathbb{E}[\bm{\Theta}] \approx \bm{\Psi}_G$, where the approximation arises from the intractability of the \textit{G}-Wishart normalizing constant for non-decomposable graphs \citep{Roverato2002}, and vanishes as $\nu \to \infty$ and as $G$ approaches the complete graph. This intractability is specific to non-decomposable graphs: when $G$ is
decomposable, the normalizing constant is available in closed form and the prior factorizes over the cliques and separators of the graph \citep{Dawid1993,Roverato2000}, which simplifies the effective sample sizes introduced in Section~\ref{sec:ess}. We do not exploit this structure and treat general graphs via Monte Carlo, so decomposable graphs are covered as a special case rather than handled separately.

The framework developed in this paper applies to prior with a closed-form density over $\bm{\Theta}$. Mixture-type priors, such as the matrix-F distribution \citep{Mulder2018}, which can be represented as a scale mixture of Wishart distributions, require a different treatment, as the ESS and planning procedures do not extend directly to this setting. We leave the generalization to mixture priors for future work.

In each case, the prior parameters $(\nu,\bm{\Psi}^{*})$ fully characterize the information encoded in $\bm{\Theta}$ before any data are observed, a quantity that we formalize in Section \ref{sec:ess} through the prior Effective Sample Size. Further details on each distribution, including conjugacy results, are provided in Appendix \ref{appendix:priors}.

\section{The Prior Effective Sample Size}
\label{sec:ess}

The posterior distribution $\pi(\bm{\Theta}\mid\mathbf{X})$ combines two sources of information: the observed data through the likelihood $f(\mathbf{X};\bm{\Theta})$ and prior beliefs on the precision matrix through $\pi(\bm{\Theta})$. A natural question is how much each source contributes and whether these contributions can be expressed in terms of an equivalent number of independent observations. We formalize the concept of prior Effective Sample Size (prior ESS) for the multiparametric GGM setting. Under the Wishart prior, the posterior parameters are $\left(\nu + n, \left({\bm{\Psi}^{*}}^{-1} + nS\right)^{-1}\right)$, suggesting that the prior contributes $\nu$ equivalent observations while the data contribute $n$. In an elicited Wishart prior, $\nu$ has a natural interpretation as the size of the prior study on which the elicitation is based. However, the number of observations the prior is actually worth is not necessarily equal to $\nu$: in a multivariate precision matrix, parameters are not independent, and the information the prior carries about each entry of $\bm{\Theta}$ depends on the full covariance structure encoded in $\bm{\Psi}^{*}$ and the graph topology $G$. A principled ESS must account for these dependencies. Following \citet{Neuenschwander2020}, this motivates an additive decomposition of the posterior ESS,
\[
ESS_{\bm{\Theta}\mid\mathbf{X}} = ESS_{\mathbf{X}\mid\bm{\Theta}} + ESS_{\bm{\Theta}}  
\]
where $ESS_{\mathbf{X}\mid\bm{\Theta}}$ is the data ESS and $ESS_{\bm{\Theta}}$ is the prior ESS. Under the prior predictive distribution, the data contribute $n$ equivalent observations, so that $ESS_{\bm{\Theta}\mid\mathbf{X}} = n + ESS_{\bm{\Theta}} $; we establish this for the candidate estimators in Appendix~\ref{appendix:pcc_proofs}. We focus on the prior ESS, $ESS_{\bm{\Theta}}$, as the basis for characterizing prior informativeness and guiding sample size planning in Bayesian GGMs.

\subsection{Global and Parameterwise Prior ESS}
\label{subsec:global_and_parameterwise_prior_ess}

The precision matrix $\bm{\Theta}$ contains $p(p+1)/2$ unique parameters, and any ESS method applied to $\bm{\Theta}$ produces a ratio of two symmetric positive definite matrices: a numerator matrix $A$ and a denominator matrix $B$, both of dimension $d \times d$, with $d = p(p+1)/2$, after half-vectorization of $\bm{\Theta}$. The reduction from $p^2 \times p^2$ is described in Appendix \ref{appendix:matrix_operators}. The five ESS methods presented in Section \ref{subsec:ess_methods} each instantiate this framework with method-specific choices of $A$ and $B$. The prior ESS can be summarized from this ratio at two levels: (i) globally across $\bm{\Theta}$, or (ii) separately for each parameter $\theta_{ij}$.

\paragraph*{Global prior ESS.} The global prior ESS is defined via the determinant ratio \citep{Vats2019},
\[
ESS(\bm{\Theta}) = \left(\frac{\determinant{A}}{\determinant{B}}\right)^{1/d}.
\]
The determinant captures the full covariance structure of $\bm{\Theta}$, including cross-parameter dependencies encoded in the off-diagonal entries of $A$ and $B$. The \textit{d}-th root normalizes by the number of unique parameters, placing the global ESS on the scale of a single parameter ESS.

\paragraph*{Parameterwise prior ESS.} The global ESS decomposes into parameter-specific contributions through the two Cholesky factorizations: $A = L^{A}\left(L^{A}\right)^{\top}$ and $B = L^{B}\left(L^{B}\right)^{\top}$. Since $\text{det}(A) = \prod_{k=1}^d \left(L^{A}_{kk}\right)^2$ and $\text{det}(B) = \prod_{k=1}^d \left(L^{B}_{kk}\right)^2$, the determinant ratio becomes
\[
ESS(\bm{\Theta}) = \left(\frac{\prod_{k=1}^d \left(L^{A}_{kk}\right)^2}{\prod_{k=1}^d \left(L^{B}_{kk}\right)^2}\right)^{1/d}
\]
This motivates a parameterwise prior ESS for each unique entry $\theta_{ij}$, where $k = 1,\ldots, d$ indexes the unique entries of $\bm{\Theta}$,
\[
ESS(\theta_{ij}) = \frac{\left(L^{A}_{kk}\right)^2}{\left(L^{B}_{kk}\right)^2}
\]
such that the global ESS is the geometric mean of the parameterwise values. The parameterwise prior ESS
shows how the overall prior informativeness is distributed across the individual parameters, identifying which edges of the graph contribute the most. The Cholesky factorization is order-dependent, so that the parameterwise prior ESS values $\text{ESS}(\theta_{ij})$ depend on the indexing of the unique entries of $\bm{\Theta}$, while the global $\text{ESS}(\bm{\Theta})$ remains invariant under any permutation. We adopt the column-major lower triangular ordering of $\bm{\Theta}$, indexing the unique entries $\theta_{ij}$, $i \geq j$ as the canonical parameterization, consistent with the half-vectorization of $\bm{\Theta}$ used throughout.

Alternative aggregations such as the average of diagonal ratios $\frac{1}{d}\sum_{k}A_{kk}/B_{kk}$ and the trace ratio $\text{tr}(A)/\text{tr}(B)$ operate only on diagonal entries, disregarding the off-diagonal structure that encodes covariance and cross-information between parameters. For these reasons, we adopt the determinant ratio as the global prior ESS throughout this paper.

\subsection{ESS Methods for the Precision Matrix}
\label{subsec:ess_methods}

We adapt five pre-data ESS methods to the precision matrix $\bm{\Theta}$ under the elicitation convention $\bm{\Psi}^{*} = \bm{\Psi}/\nu$, so that $\mathbb{E}[\bm{\Theta}] = \bm{\Psi}$ exactly under the Wishart prior and approximately under the \textit{G}-Wishart prior. Under this convention, $\bm{\Psi}$ is directly interpretable as the prior belief about the precision matrix. Each method instantiates the general $(A,B)$ framework introduced in Section \ref{subsec:global_and_parameterwise_prior_ess} with method-specific choices of $A$ and $B$, expressed in terms of the prior variance matrix $\mathbb{V}(\bm{\Theta})$ and the expected (inverse) Fisher information for a single observation $\mathcal{I}(\mathbf{x}_1;\bm{\Theta})$, both of dimension $p^2 \times p^2$ under full vectorization. Prior-specific expressions for these quantities, and their reduction to dimension $\frac{p(p+1)}{2} \times \frac{p(p+1)}{2}$ via half-vectorization, are provided in Appendix \ref{appendix:matrix_operators}, together with the definitions and key properties of the commutation matrix $\mathbf{K}$, elimination matrix $\mathbf{E}$, duplication matrix $\mathbf{D}$ and corresponding Moore-Penrose pseudoinverse $\mathbf{D}^{+}$ used throughout. Analytical derivations of the closed-form results are provided in Appendix \ref{appendix:prior_ess_formulas}.

\paragraph{Variance Ratio (VR).} The Variance Ratio \citep[VR;][]{Neuenschwander2020} defines the prior ESS as the ratio of the expected single-observation information to the prior variance,
\[
ESS_{\text{VR}} = \left(\frac{\determinant{\mathbb{E}_{\bm{\Theta}}\left[\mathcal{I}^{-1}(\mathbf{x}_1;\bm{\Theta})\right]}}{\determinant{\mathbb{V}\left(\bm{\Theta}\right)}}\right)^{1/d} = \nu \cdot \left( \frac{\determinant{\mathbf{A} + \mathbf{B}/\nu}}{\determinant{\mathbf{A}}} \right)^{1/d}
\]
where 
\begin{align*}
    \mathbf{A} &= \mathbf{E}(\mathbf{I} + \mathbf{K})(\bm{\Psi} \otimes \bm{\Psi})\mathbf{E}^{\top} \quad{}\text{(with } \mathbf{I}\text{ the identity matrix of appropriate dimensions)}\\
    \mathbf{B} &= 2\mathbf{D}^{+}\left[\text{vec}(\bm{\Psi})\text{vec}(\bm{\Psi})^{\top} + \mathbf{K} (\bm{\Psi} \otimes\bm{\Psi})\right]{\mathbf{D}^{+}}^{\top}.
\end{align*}
As $\nu \rightarrow \infty$, $ESS_{\text{VR}}\sim\nu$, growing linearly with the prior sample size. For finite $\nu$ and $p$, $ESS_{\text{VR}}>\nu$ and converges to $\nu$ from above. Under the \textit{G}-Wishart prior, $\mathbb{E}_{\bm{\Theta}}\left[\mathcal{I}^{-1}(\mathbf{x}_1;\bm{\Theta})\right]$ and $\mathbb{V}\left(\bm{\Theta}\right)$ do not admit closed-form expressions and are estimated via Monte Carlo simulation from $W_{G}\left(\nu,\bm{\Psi}/\nu\right)$.

\paragraph{The Precision Ratio (PR).} The Precision Ratio \citep[PR;][]{Neuenschwander2020} inverts the VR relationship, comparing prior precision directly to expected information
\[
ESS_{\text{PR}} = \left(\frac{\determinant{\mathbb{V}^{-1}\left(\bm{\Theta}\right)}}{\determinant{\mathbb{E}_{\bm{\Theta}}\left[\mathcal{I}(\mathbf{x}_1;\bm{\Theta})\right]}}\right)^{1/d} = \frac{(\nu - p)(\nu - p - 1)(\nu - p - 3)}{\nu (\nu - p - 2)}\cdot \left( \frac{\lvert\mathbf{B}\rvert}{ \lvert   \mathbf{B} + \mathbf{C}/(\nu - p - 2) \rvert} \right)^{1/d} 
\]
where
\begin{align*}
\mathbf{B} &= \mathbf{D}^{\top}(\bm{\Psi}^{-1}  \otimes  \bm{\Psi}^{-1})\mathbf{D} \\
\mathbf{C} &= \mathbf{D}^{\top}\left(\ \text{vec}(\bm{\Psi}^{-1})\text{vec}(\bm{\Psi}^{-1})^{\top} +  \ \mathbf{K} (\bm{\Psi}^{-1} \otimes  \bm{\Psi}^{-1})\right)\mathbf{D}.
\end{align*}
As $\nu \rightarrow \infty$, by continuity of the determinant $\left( \frac{\lvert\mathbf{B}\rvert}{ \lvert   \mathbf{B} + \mathbf{C}/(\nu - p - 2) \rvert} \right)^{1/d} \rightarrow 1$ and $ESS_{\text{PR}}$ is approximated by the factor $\frac{(\nu - p)(\nu - p - 1)(\nu - p - 3)}{\nu (\nu - p - 2)}$, which grows linearly with $\nu$ and converges to $\nu$ as $\nu \rightarrow \infty$ with $p$ fixed. For finite $\nu$ and $p$, $ESS_{\text{PR}} < \nu$, with the gap increasing as $p$ grows relative to $\nu$. Under the \textit{G}-Wishart prior, $\mathbb{V}^{-1}\left(\bm{\Theta}\right)$ and $\mathbb{E}_{\bm{\Theta}}\left[\mathcal{I}(\mathbf{x}_1;\bm{\Theta})\right]$ do not admit closed-form expressions and are estimated via Monte Carlo simulation from $W_{G}\left(\nu, \bm{\Psi}/\nu\right)$, as for VR.

\paragraph{The Morita-Thall-Müller (MTM).} The MTM method \citep{Morita2008,Morita2010} identifies the prior ESS as the sample size $m$ that minimizes the discrepancy between the prior and a locally-matched $\varepsilon\text{-information}$ reference prior at the prior mean $\overline{\bm{\Theta}} = \bm{\Psi}$,
\[
ESS_{\text{MTM}} = \arg\min_{m} \Delta(m,\overline{\bm{\Theta}},\pi,\pi_{\varepsilon})
\]
where $\nu_{\text{min}}=p$ is the minimum admissible degrees of freedom for the reference prior and $\varepsilon>0$ is arbitrarily small. Under the Wishart prior, the discrepancy reduces analytically and setting $\varepsilon = 1$ yields
\[
ESS_{\text{MTM}} = \nu - 1.
\]
Under the \textit{G}-Wishart, $\nu_{\text{min}}=p$ by the same argument and $ESS_{\text{MTM}} = \nu - 1$. 

\paragraph{The Pennello-Thompson (PT).} The PT method \citep{Neuenschwander2020}, computes the ESS by taking the ratio of prior information to single-observation Fisher information at the prior mode $\tilde{\bm{\Theta}}$
\[
ESS_{\text{PT}} = \frac{I(\tilde{\bm{\Theta}})}{\mathcal{I}(\mathbf{x}_1;\tilde{\bm{\Theta}})}.
\]
Under the Wishart prior, the prior and single-observation Fisher information share the same matrix $\mathbf{D}^{\top}(\tilde{\bm{\Theta}}^{-1}\otimes\tilde{\bm{\Theta}}^{-1})\mathbf{D}$, which cancels in the ratio, yielding 
\[ 
ESS_{\text{PT}} = \nu - p -1.
\]
Under the \textit{G}-Wishart, $ESS_{\text{PT}} = \nu - p -1$.
    
\paragraph{The Expected Local Information Ratio (ELIR).} The ELIR method \citep{Neuenschwander2020} averages this ratio over the prior distribution, accounting for the uncertainty in $\bm{\Theta}$,
\[
ESS_{\text{ELIR}} = \mathbb{E}_{\bm{\Theta}}\left[\frac{I(\bm{\Theta})}{\mathcal{I}(\mathbf{x}_1;\bm{\Theta})}\right].
\]
Under the Wishart prior, the matrix $\mathbf{D}^{\top}(\tilde{\bm{\Theta}}^{-1}\otimes\tilde{\bm{\Theta}}^{-1})\mathbf{D}$ cancels in the ratio as in PT, and the expectation reduces to
\[
ESS_{\text{ELIR}} = \nu - p -1.
\]
Under the \textit{G}-Wishart, $ESS_{\text{ELIR}} = \nu - p - 1$.

In GGMs, the observed and expected Fisher information are equal, since the observed information does not depend on the data. This property simplifies the computation of all five methods above. The five methods differ in their sensitivity to the precision matrix and its structure. As shown above, the prior ESS under MTM, PT, and ELIR depends only on $\nu$ and $p$, and is invariant to $\bm{\Psi}^{*}$ and the graph topology. The VR and PR methods, in contrast, depend on the full structure of $\bm{\Psi}^{*}$ through $\mathbf{A}$, $\mathbf{B}$, and $\mathbf{C}$, making the ESS sensitive to parameter dependencies and graph structure. Under the elicitation convention $\bm{\Psi}^{*} = \bm{\Psi}/\nu$, both VR and PR grow linearly with $\nu$ as $\nu \rightarrow \infty$. The opposite behaviors of VR and PR, $ESS_{\text{VR}} > \nu$ and $ESS_{\text{PR}} < \nu$, reflect a Jensen gap. For positive definite random matrices, the convexity of matrix inversion gives $\mathbb{E}_{\bm{\Theta}}\left[\mathcal{I}(\mathbf{x}_1;\bm{\Theta})\right]^{-1} \le \mathbb{E}_{\bm{\Theta}}\left[\mathcal{I}^{-1}(\mathbf{x}_1;\bm{\Theta})\right]
$. In both methods, the determinant is nonlinear in $\bm{\Theta}$, and its opposite curvature under the VR and PR constructions produces this gap for finite $\nu$ and $p$. The Jensen gap $J = 1-ESS_{\text{PR}}/ESS_{\text{VR}}$ formalizes this difference and is derived in Appendix \ref{appendix:jensen_gap}. The magnitude of this gap as a function of $\nu$ and $p$, and its empirical consequences, are examined in the simulation study in Section \ref{subsec:simulation_study_prior_ess}.

Taken together, all five methods satisfy a predictive consistency criterion that extends \citet{Neuenschwander2020} to the multiparametric GGM setting, holding exactly for MTM, PT, and ELIR under both the Wishart and \textit{G}-Wishart priors, and asymptotically for VR and PR as $\nu \rightarrow \infty$ with $p$ fixed, consistent with their convergence to $\nu$ established above (see Appendix~\ref{appendix:pcc_proofs}).

\section{Sample Size Planning}
\label{sec:sample_size_planning}

The information encoded in the prior distribution directly affects sample size planning in GGMs. If an informative prior is derived from a previous study, the new data must provide sufficient information to ensure that the posterior is influenced more by the observed evidence than by prior beliefs. We address this issue through a two complementary planning strategies. The first, the Data-to-Prior Information Ratio (DPIR), a novel information-based method introduced here, determines the sample size needed for the data to dominate the prior in terms of Fisher information. The second, Bayes Factor Design Analysis~\citep[BFDA;][]{Schnbrodt2017}, extended here to the GGM setting, determines the sample size needed to achieve sufficient evidence for or against conditional independence on each edge. The two strategies target complementary aspects of the planning problem: DPIR ensures that the data carry enough information to dominate the prior, while BFDA ensures that the data provide enough evidence for edge-specific testing. Furthermore, their recommended sample sizes are reported jointly with the prior ESS of Section~\ref{subsec:ess_methods} as a diagnostic of prior informativeness. Recommendations on how to use the two sample sizes together are given at the end of Section~\ref{subsec:bfda}.

The two complementary planning strategies are illustrated on a toy example throughout this section. The example consists of a single precision matrix $\bm{\Psi}$ of dimension $p = 10$ with prior degrees of freedom $\nu = 100$, corresponding to a prior study of size $100$. The matrix was generated using the procedure described in Algorithm~\ref{algorithm:random_precision} and respects the zero pattern of a random graph $G$ with 30 present edges out of 45 possible (density $= 0.67$). The partial correlations range from $-0.325$ to $0.483$, reflecting a heterogeneous association structure across the present edges. We report the matrix $\bm{\Psi}$ in Appendix~\ref{appendix:illustrative_precision_matrix} for reproducibility, and use the elicited scale $\bm{\Psi}^{*} = \bm{\Psi}/\nu$ throughout this section. The planning strategies illustrated here are validated across prior specifications and graph structures in the simulation study of Section~\ref{sec:simulation_studies}.

\subsection{Data-to-Prior Information Ratio}
\label{subsec:dpir}

The DPIR measures the relative information provided by a sample of size $n$ with respect to that contained in the prior. For $\mathbf{x}_i \overset{\text{iid}}{\sim} \mathcal{N}_p(\mathbf{0}, \bm{\Theta}^{-1})$, $i = 1, \ldots, n$, with prior $\bm{\Theta} \sim \pi(\bm{\Theta})$, the DPIR is defined as
\[
\Lambda_{n} (\bm{\Theta},\mathbf{X}) = \left(\frac{\lvert\mathcal{I}(\mathbf{X};\hat{\bm{\Theta}}\left(\mathbf{X}\right))\rvert}{\lvert I(\bm{\Theta})\rvert}\right)^{1/d}
\]
where $d = p(p+1)/2$, $\mathcal{I}(\mathbf{X};\hat{\bm{\Theta}}\left(\mathbf{X}\right))$ is the expected Fisher information of the data evaluated at the MLE $\hat{\bm{\Theta}}(\mathbf{X}) = S^{-1}$ (see Appendix \ref{appendix:ggm_derivatives}), and $I\left(\bm{\Theta}\right)$ is the observed Fisher information of the prior at $\bm{\Theta}$ (see Appendix \ref{appendix:priors}). Values $\Lambda_n>1$ indicate that the data carry more information than the prior.

Since $\Lambda_n$ depends both on $\bm{\Theta}$ and $\mathbf{X}$, we consider its behavior under the prior predictive distribution. The probability that the data information exceeds the prior information by a factor of at least $\xi\ge 1$ is
\[
    \Pr(\Lambda_{n}  > \xi) = \mathbb{E}_{\bm{\Theta}}\left[\mathbb{E}_{\mathbf{X}|\bm{\Theta}}\left[\mathbf{1}\left(\Lambda_{n} (\bm{\Theta},\mathbf{X})> \xi\right)\right]\right]
\]
The default choice $\xi=1$ corresponds to requiring the data to carry at least as much information as the prior. Values $\xi>1$ impose a stricter requirement for dominance.

The optimal sample size is the minimal $n$ such that this probability reaches a target level $\tau\in(0,1]$,
\[
n^{\star} = \min {\left\lbrace n \in \mathbb{N} : \Pr(\Lambda_{n} > \xi) \ge \tau\right\rbrace}
\]
with default $\tau=0.99$, ensuring that the data dominate the prior with high probability under the prior predictive distribution. Larger values of $\tau$ require dominance to hold with greater probability under the prior predictive distribution, while smaller values relax this requirement.
Algorithm \ref{algorithm:dpir} describes the computation of $\Pr(\Lambda_{n}  > \xi)$ for a fixed $n$, returning a Monte Carlo estimate obtained by averaging the indicator $\mathbf{1}(\Lambda_n>\xi)$ over draws from the prior predictive distribution. The optimal sample size $n^{\star}$ is found by bisection on $n$, exploiting the monotonicity of $\Pr(\Lambda_n > \xi)$ in $n$ under the prior predictive distribution.

\begin{figure}[t]
\centering
\begin{subfigure}{0.49\textwidth}
    \def\svgwidth{\linewidth}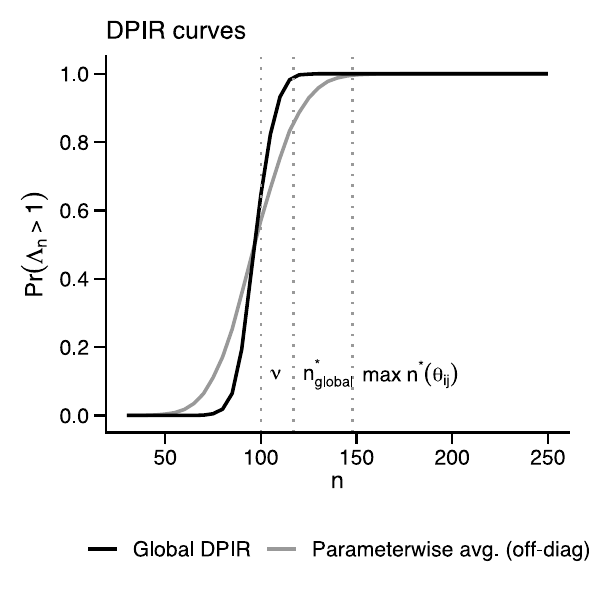
    \caption{}
    \label{subfig:dpir_curves}
\end{subfigure}
\hfill
\begin{subfigure}{0.49\textwidth}
    \def\svgwidth{\linewidth}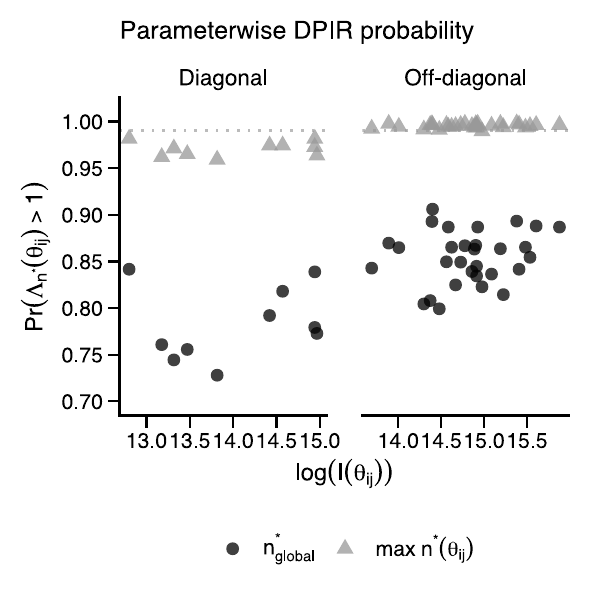
    \caption{}
    \label{subfig:dpir_diagnostic}
\end{subfigure}
\caption{DPIR sample size planning for the toy example with $p = 10$ and $\nu = 100$. \emph{Panel (a)}: $\Pr(\Lambda_n > 1)$ as a function of $n$ for the global (black) and the average parameterwise DPIR across off-diagonal parameters (gray). Dotted vertical lines mark $\nu$, $n^{\star}_{\text{global}}$, and $\max_{(i,j) \notin \text{diag}} n^{\star}(\theta_{ij})$. \emph{Panel (b)}: $\Pr(\Lambda_{n^{\star}}(\theta_{ij}) > 1)$ at $n^{\star}_{\text{global}}$ (circles) and $\max_{(i,j) \notin \text{diag}} n^{\star}(\theta_{ij})$ (triangles) against prior Fisher information $\log I(\theta_{ij})$, separately for diagonal and off-diagonal parameters.}
\label{fig:dpir}
\end{figure}

The DPIR can also be evaluated at the parameter level. Following the Cholesky decomposition, the determinant ratio decomposes into parameter-specific contributions, resulting in a parameterwise DPIR for each precision entry $\theta_{ij}$,
\[
\Lambda_n(\theta_{ij})=\frac{\left(L_{kk}^{\mathcal{I}(\mathbf{X};\hat{\bm{\Theta}}(\mathbf{X}))}\right)^{2}}{\left(L_{kk}^{I(\bm{\Theta})}\right)^{2}}
\]
where $k$ indexes the unique entries of $\bm{\Theta}$ following the column-major lower triangular ordering of Section \ref{subsec:global_and_parameterwise_prior_ess}, and $L_{kk}^{A}$ denotes the \textit{k}-th diagonal entry of the Cholesky factor of $A$. Values $\Lambda_n(\theta_{ij})>\xi$ indicate that the data exceed the prior information by at least a factor $\xi$ for that parameter. The parameterwise optimal sample size is then
\[
n^{\star}(\theta_{ij}) = \min{\left\lbrace n\in\mathbb{N} \ : \ \Pr(\Lambda_n(\theta_{ij}) > \xi) \ge \tau \right\rbrace}.
\]
The parameterwise DPIR evaluated at the planned sample size is a useful post-planning diagnostic. For each precision parameter $\theta_{ij}$, a value of $\Pr(\Lambda_{n^{\star}}(\theta_{ij}) > \xi)$ below the target $\tau$ indicates that the data do not dominate the prior for that precision parameter at the planned $n^{\star}$. The diagnostic is computed following Algorithm~\ref{algorithm:dpir}, averaging $\mathbf{1}(\Lambda_{n^{\star}}(\theta_{ij}) > \xi)$ over Monte Carlo draws from the prior predictive distribution at the planned $n^{\star}$. While the parameterwise DPIR is defined for every entry of $\bm{\Theta}$, sample size planning focuses on off-diagonal parameters, which encode the partial correlations targeted by the edge-specific BFDA in Section~\ref{subsec:bfda}.

To illustrate the DPIR analysis, we apply it to the toy example introduced above, with $p = 10$, $\nu = 100$, and $\bm{\Theta} \sim W_G(\nu, \bm{\Psi}^{*})$ with $G$ encoding 30 present edges out of 45 possible, and a target probability $\tau = 0.99$. Figure \ref{fig:dpir} shows the results of the analysis. Figure \ref{subfig:dpir_curves} presents the probability $\Pr(\Lambda_n > 1)$ as a function of $n$ for the global DPIR (black) and the average parameterwise DPIR across off-diagonal parameters (gray). Both curves exhibit a sigmoidal shape, rising steeply from zero near $n = \nu = 100$ (gray dotted vertical line) and converging to one as $n$ grows. The global optimal sample size $n^{\star}_{\text{global}} = 117$ is above $\nu$, confirming that somewhat more observations than the prior study size are needed for the data to dominate the prior. The gray curve, the average $\Pr(\Lambda_n(\theta_{ij}) > 1)$ over off-diagonal parameters, rises more slowly than the black one, showing that edges on average require more data than the aggregate global criterion. The parameterwise optimum $\max_{(i,j) \notin \text{diag}} n^{\star}(\theta_{ij}) = 148$ exceeds the global one, $n^{\star}_{\text{global}}$. This gap reflects that the global criterion aggregates across parameters, allowing the weakly informative parameters to compensate for those with higher prior information. Figure \ref{subfig:dpir_diagnostic} shows $\Pr(\Lambda_{n^{\star}}(\theta_{ij}) > 1)$ at $n^{\star}_{\text{global}} = 117$ (circles) and $\max_{(i,j) \notin \text{diag}}  n^{\star}(\theta_{ij}) = 148$ (triangles), plotted against the prior Fisher information $\log I(\theta_{ij})$, separately for diagonal and off-diagonal parameters. At $n^{\star}_{\text{global}} = 117$, all diagonal parameters are below the target, while off-diagonal parameters are closer to but not yet at the target. At $\max_{(i,j) \notin \text{diag}} n^{\star}(\theta_{ij}) = 148$, all off-diagonal parameters reach the target probability of $0.99$, while the diagonal entries lie between $0.96$ and $0.99$.

Thus, the DPIR provides the first component of the planning strategy, determining the minimum sample size at which the information contributed by the data exceeds that encoded in the prior. The second component, the BFDA, addresses the complementary question of how many observations are needed to achieve sufficient evidence for edge-specific testing, as described in the following subsection.

\subsection{Bayes Factor Design Analysis} 
\label{subsec:bfda}
The second strategy determines the sample size needed to achieve sufficient evidence for or against conditional independence, calibrated to a representative planning edge and guaranteed to extend to all stronger edges by monotonicity. We apply BFDA \citep{Schnbrodt2017} to the edge-specific hypotheses of conditional independence in GGMs. Since $\rho_{ij} = 0$ if and only if $\theta_{ij}=0$, this is equivalently expressed as a hypothesis on the partial correlation,
\begin{equation}
    \begin{cases}
    \mathcal{H}_0 &: \rho_{ij} = 0 \\
    \mathcal{H}_1 &: \rho_{ij} \in(-1,1)
    \end{cases} 
    \label{eq:hypothesis_conditional_independence}
\end{equation}
We focus exclusively on edge-specific hypotheses. Inference on the whole graph structure is not considered in this paper. Testing is performed on all present edges in $G$, with the test in \eqref{eq:hypothesis_conditional_independence} assessing whether the data support their exclusion.

\paragraph{Bayes factor for complete graphs.} When the prior graph $G$ is complete, the researcher specifies the same hyperparameters $(\nu,\bm{\Psi}^{*})$ used to determine the ESS in Section~\ref{sec:ess}. Under $\mathcal{H}_1$, $\bm{\Theta}$ is drawn from the full Wishart $W(\nu,\bm{\Psi}^{*})$; under $\mathcal{H}_0$, from the \textit{G}-Wishart $W_{G_{-(i,j)}}(\nu,\bm{\Psi}^{*})$, where  $G_{-(i,j)}$ denotes the graph with edge $(i,j)$ removed. This construction ensures compatibility between the two priors in the sense of \citet{Dickey1971}, which is required for the Savage-Dickey representation of the Bayes factor to hold.

Following \citet[Lemma~3]{Giudici1995}, the Bayes factor on $\theta_{ij}$ conditional on the remaining variables $b=\left\lbrace 1,\ldots,p \right\rbrace \setminus\left\lbrace i,j \right\rbrace$ takes the form of a Savage-Dickey density ratio:
\begin{equation}
BF_{01} = \frac{\Pr(\mathbf{X} \mid \mathcal{H}_0)}{\Pr(\mathbf{X} \mid \mathcal{H}_1)} = \frac{\Pr(\theta_{ij}=0 \mid \mathbf{X},\mathcal{H}_1)}{\Pr(\theta_{ij}=0 \mid \mathcal{H}_1)}.
\label{eq:bf01}
\end{equation}
This has a closed form:
\begin{equation}
BF_{01} = C(n,\nu)\cdot\frac{(1-\hat{\rho}_{ij}^{2})^{(\nu+n)/2}}{(1-\rho_{ij}^{2})^{\nu/2}}\cdot\sqrt{\frac{T_{ii}T_{jj}}{(T_{ii}+S_{ii|b})(T_{jj}+S_{jj|b})}}
\label{eq:conditional_bf_giudici_closed_form}
\end{equation}
where $\rho_{ij}$ is the prior partial correlation between variables $i$ and $j$ given $b$, $\hat{\rho}_{ij}$ is the posterior partial correlation, $\mathbf{T}=(\bm{\Psi}^{*}_{aa})^{-1}$ is the prior partial variance matrix restricted to $\left\lbrace i,j \right\rbrace$, $S_{ij|b}$ is the sample partial deviance of variables $i$ and $j$ conditional on $b=\left\lbrace 1,\ldots,p \right\rbrace \setminus \left\lbrace i , j\right\rbrace$, and $C(n,\nu)$ is a ratio of Gamma functions depending only on $n$ and $\nu$\footnote{$C(n,\nu) = \exp\!\left[
\log\Gamma\!\left(\frac{\nu}{2}\right) +
\log\Gamma\!\left(\frac{\nu-1}{2}\right) +
2\log\Gamma\!\left(\frac{\nu+n+1}{2}\right) -
\log\Gamma\!\left(\frac{\nu+n}{2}\right) -
\log\Gamma\!\left(\frac{\nu+n-1}{2}\right) -
2\log\Gamma\!\left(\frac{\nu+1}{2}\right)
\right]$}. Conditioning on the remaining variables reduces the effective dimensionality of the testing problem, making the approach suitable when $p > n$.
Since $\rho_{ij} = T_{ij}/\sqrt{T_{ii}T_{jj}}$ is scale invariant, the BFDA output depends on $\bm{\Psi}^{*}$ only through the prior partial correlations. Consequently, the power curves and required sample size $n^{\star}$ are determined entirely by $\rho_{ij}$ and $\nu$, making ESS and $n^{\star}$ directly comparable as functions of the same parameter elicitation $(\nu,\bm{\Psi}^{*})$.

\paragraph{Bayes factor for sparse graphs.} When the prior structure $G$ has at least one edge constrained to zero, the prior is $\bm{\Theta} \sim W_G(\nu,\bm{\Psi}^{*})$ under $\mathcal{H}_1$, and $\bm{\Theta} \sim W_{G_{-(i,j)}}(\nu,\bm{\Psi}^{*})$ under $\mathcal{H}_0$, with the same $(\nu,\bm{\Psi}^{*})$ as above. A closed-form Bayes factor is no longer available since the normalizing constant of the \textit{G}-Wishart is analytically intractable for non-decomposable graphs. The marginal likelihood under each hypothesis takes the form of a ratio of \textit{G}-Wishart normalizing constants between the posterior and the prior, and the Bayes factor reduces to:
\begin{equation}
BF_{01} = 
\frac{I_{G_{-(i,j)}}\!\left(\nu + n,\, \nu\bm{\Psi}^{-1} + \mathbf{X}^\top\mathbf{X}\right)}{I_{G}\!\left(\nu + n,\, \nu\bm{\Psi}^{-1} + \mathbf{X}^\top\mathbf{X}\right)} \cdot \frac{I_{G}\!\left(\nu,\, \nu\bm{\Psi}^{-1}\right)}{I_{G_{-(i,j)}}\!\left(\nu,\, \nu\bm{\Psi}^{-1}\right)}
\label{eq:bf_normalizing_constant}
\end{equation}
where $I_{G}(\nu,\nu\bm{\Psi}^{-1})$ denotes the normalizing constant of 
the \textit{G}-Wishart $W_G(\nu,\bm{\Psi}^{*})$. Since all four normalizing constants are analytically intractable for non-decomposable graphs, this ratio is estimated via the Monte Carlo procedure of \citet{Atay2005}, which expresses it as an expectation over auxiliary variables derived from the Cholesky decomposition of the precision matrix. Technical details are given in Appendix \ref{appendix:ak_bf}, where we follow the original notation of \citet{Atay2005} with $\delta = \nu - p +1$.

\paragraph{Planning edge and sample size determination.} Since the expected Bayes factor is not available in closed form in either case, the power function is estimated by simulation following Algorithm \ref{algorithm:bfda}. The sample size recommendation is calibrated to a single planning edge selected from the distribution of prior partial correlations, so that its required sample size provides a conservative upper bound for all edges carrying a stronger signal in the prior graph. To perform a planning analysis on a representative and detectable effect, we define $\rho_{\min}$ as the empirical $q$-th quantile of $|\rho_{ij}|$ across all present edges in $G$, with default $q = 0.50$. This divides the present edges into the detectable set $\mathcal{E}^{+} = \{(i,j) \in G : |\rho_{ij}| \geq \rho_{\min}\}$, comprising the $1-q$ strongest edges, and a remainder whose partial correlations are too weak to anchor a meaningful planning recommendation. The planning edge $(i, j)$ is the weakest edge in $\mathcal{E}^{+}$, with partial correlation $\rho_{\min}$. If $\rho_{\min} < \epsilon$ for a hard minimum $\epsilon > 0$ (default $0.05$), the edge is flagged as unplannable. The threshold $\epsilon$ is substantive rather than computational: partial correlations below it are considered scientifically negligible regardless of their statistical detectability.

The optimal sample size is then determined separately under $\mathcal{H}_{0}$ and $\mathcal{H}_{1}$ by bisection on the planning edge $(i, j)$, initialized at a frequentist warm start based on the Fisher $z$-transform of $\rho_{\min}$ \citep{Fisher1915, Cohen1988}, yielding
\begin{equation}
    n^{\star}_{\mathcal{H}_{0}} = \min\bigl\{n \in \mathbb{N} :
        \beta_{0}(n,\,\rho_{\min}) \geq \pi_{0}\bigr\},
    \qquad
    n^{\star}_{\mathcal{H}_{1}} = \min\bigl\{n \in \mathbb{N} :
        \beta_{1}(n,\,\rho_{\min}) \geq \pi_{1}\bigr\},
    \label{eq:nstar_planning_edge}
\end{equation}
where $\beta_{0}(n,\rho_{\min})$ and $\beta_{1}(n,\rho_{\min})$ denote the power under $\mathcal{H}_{0}$ and $\mathcal{H}_{1}$ respectively at sample size $n$ for the planning edge, and $\pi_{0} = \pi_{1} = 0.80$ are the target power levels. The decision rule is based on a Bayes factor threshold $\gamma > 1$: the data provide sufficient evidence for $\mathcal{H}_{0}$ if $BF_{01} > \gamma$, and for $\mathcal{H}_{1}$ if $BF_{01} < 1/\gamma$, with default $\gamma = 10$ corresponding to strong evidence in the classification of \citet{Jeffreys1961}. The recommended sample size is $n^{\star} = \max(n^{\star}_{\mathcal{H}_{0}}, n^{\star}_{\mathcal{H}_{1}})$. If $n^{\star}$ exceeds a feasible maximum sample size $n_{\max}$, determined by the study's resource constraints, then the planning edge is detectable in principle but not within those constraints, and $n^{\star}$ is flagged accordingly. 

\paragraph{Monotonicity of the planned sample size.} The key property justifying this design is the monotonicity of $n^{\star}(\rho)$ in $|\rho|$: since larger partial correlations require smaller sample sizes to detect, the planned $n^{\star}$ satisfies
\begin{equation}
    n^{\star}_{\mathcal{H}_{0}}(\rho_{ij}) \leq n^{\star}_{\mathcal{H}_{0}}, \qquad
    n^{\star}_{\mathcal{H}_{1}}(\rho_{ij}) \leq n^{\star}_{\mathcal{H}_{1}},
    \label{eq:monotonicity_guarantee}
\end{equation}
for every $(i,j) \in \mathcal{E}^{+}$ with $|\rho_{ij}| \geq \rho_{\min}$, so a
study of size $n^{\star}$ is designed to achieve the target power for every edge in $\mathcal{E}^{+}$. The monotonicity guarantee in \eqref{eq:monotonicity_guarantee} is exact in
the frequentist sense but approximate in the Bayesian one, because $\beta_0$
and $\beta_1$ depend on the full prior structure through the \textit{G}-Wishart draws
and not only on $\rho_{\min}$. Although the guarantee is only approximate in the Bayesian setting, the monotonicity holds empirically: in the toy example, plotting $\log_{10}{\left(BF_{01}\right)}$ at the planned $n^{\star}$ against $|\rho|$ across the present edges shows the evidence against the null strengthening monotonically as the partial correlation grows, so the planned $n^{\star}$ reaches the target power for every stronger edge.

\paragraph{Multiplicity and absent edges.} With $p(p-1)/2$ potential edges tested simultaneously, multiplicity is a concern. The \textit{G}-Wishart prior reduces the effective number of tests by encoding conditional independence for absent edges in the elicitation stage, restricting testing to the present edges in $G$ only. Here $G$ is treated as fixed and known, reflecting the researcher's belief that those edges are truly absent before data collection. Among the present edges, the sample size $n^{\star}$ provides an implicit guarantee through the monotonicity in \eqref{eq:monotonicity_guarantee}, such that all edges in $\mathcal{E}^{+}$ achieve at least the target power at $n^{\star}$. In the simulation study of Section \ref{sec:simulation_studies}, we evaluate the planning strategy under two evidence thresholds $\gamma = 3$ and $\gamma = 10$, corresponding to moderate and strong evidence in the classification of \citet{Jeffreys1961}. Whether a larger $\gamma$ is justifiable to control the family-wise error rate across edges is left for future work. Planning for absent edges can be accommodated within the encompassing prior approach \citep{Klugkist2005}, where the Bayes factor for inclusion reduces to the reciprocal of \eqref{eq:bf01} or \eqref{eq:bf_normalizing_constant}, and the BFDA pipeline developed here applies directly.
 
\begin{figure}[t]
    \centering
    \begin{subfigure}{0.48\textwidth}
        \centering
        \def\svgwidth{\linewidth}
        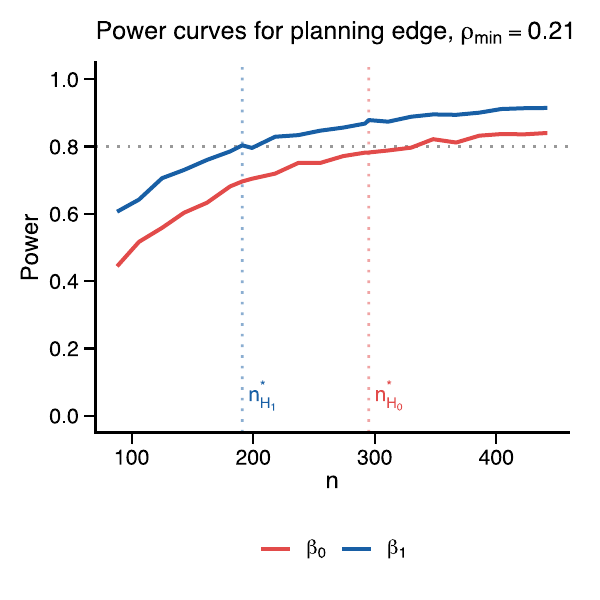
        \caption{}
        \label{subfig:bfda_power_curve}
    \end{subfigure}%
    \hfill
    \begin{subfigure}{0.48\textwidth}
        \centering
        \def\svgwidth{\linewidth}
        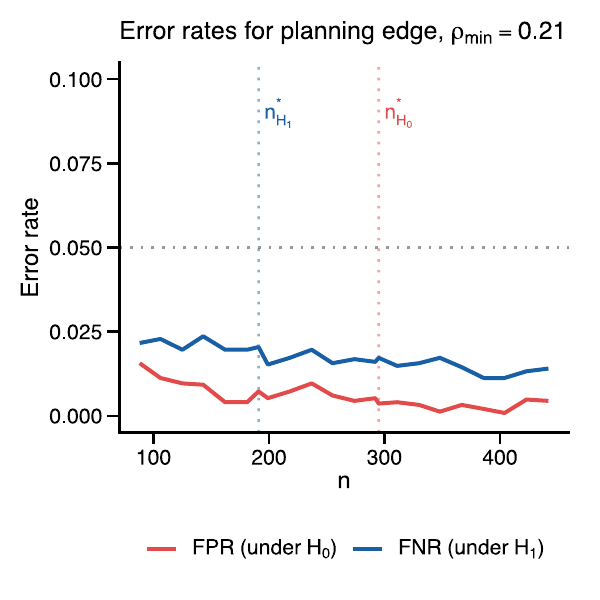
        \caption{}
        \label{subfig:bfda_error_curve}
    \end{subfigure}
    \caption{BFDA results for the planning edge $(i,j)$ with $\rho_{\min} = 0.210$, $p = 10$, and $\nu = 100$. \emph{Panel (a)}: Power curves $\beta_0$ (red) and $\beta_1$ (blue) for the planning edge as a function of $n$. Dotted lines mark the two planned sample sizes and the target power of $0.8$. \emph{Panel (b)}: FPR (red) and FNR (blue) for the planning edge as a function of $n$. Dotted lines mark the two planned sample sizes and the target error rate of $0.05$.}
    \label{fig:bfda_results}
\end{figure}

\begin{figure}[t]
    \centering
    \def\svgwidth{0.90\linewidth}
    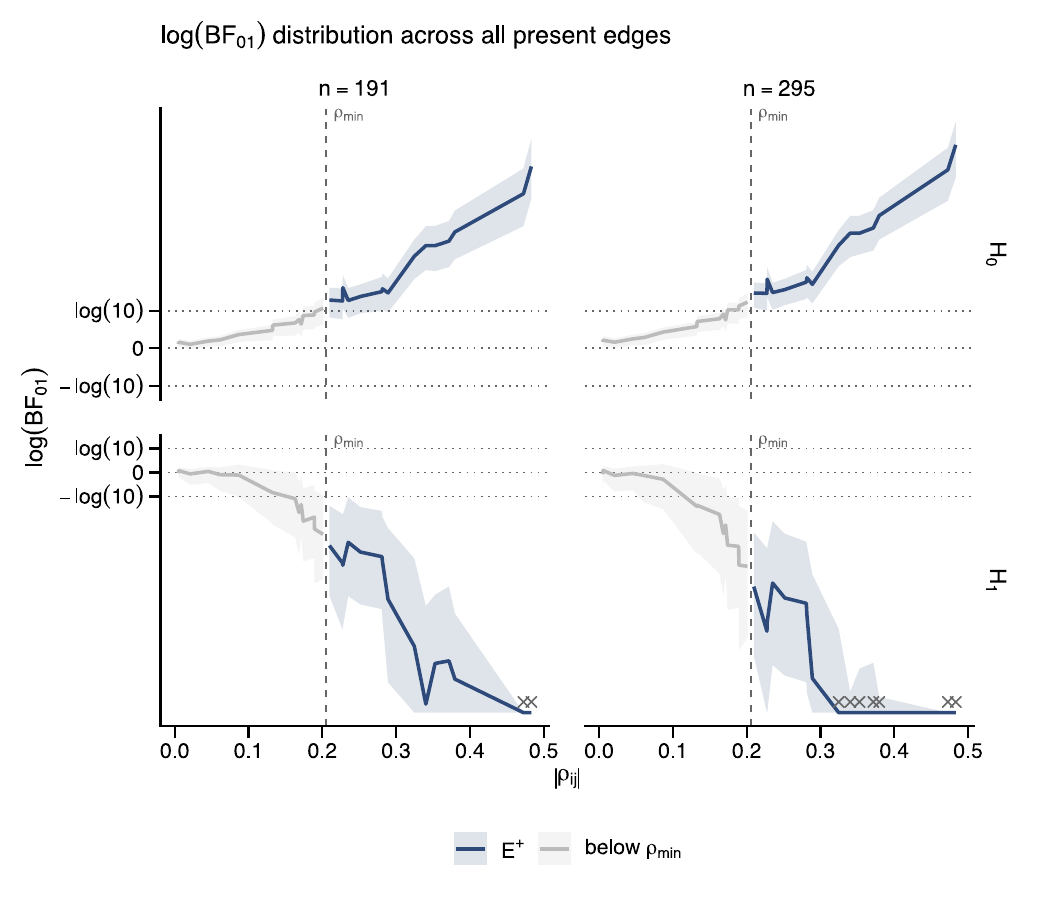
    \caption{$\log(BF_{01})$ distribution across all present edges at $n^{\star}_{\mathcal{H}_1} = 191$ (left column) and $n^{\star}_{\mathcal{H}_0} = 295$ (right column), under $\mathcal{H}_0$ (top row) and $\mathcal{H}_1$ (bottom row). Lines show the median $\log(BF_{01})$ and shaded ribbons the interquartile range across prior predictive draws. Dark blue: edges in $\mathcal{E}^+$; light gray: edges below $\rho_{\min}$. Dashed vertical line at $\rho_{\min} = 0.210$; dotted horizontal lines at $\pm\log(\gamma)$ with $\gamma = 10$ and at $0$.}
    \label{fig:bfda_bf_dist}
\end{figure}

To illustrate the BFDA strategy, we apply it to the toy example introduced earlier in this section. The empirical median of $|\rho_{ij}|$ across the 30 present edges is $\rho_{\min} = 0.21$ --- planning edge $(2,6)$ --- partitioning them into $|\mathcal{E}^{+}| = 15$ strongest edges ($|\rho_{ij}|\ge\rho_{\min}$) and 15 remainder edges. Since $\rho_{\min} = 0.21$ exceeds the threshold $\epsilon = 0.05$, we can proceed with the planning. Bisection on the planning edge yields $n^{\star}_{\mathcal{H}_{0}} = 295$ and $n^{\star}_{\mathcal{H}_{1}} = 191$, giving a recommended sample size of $n^{\star} = 295$. Figure \ref{fig:bfda_results} summarizes the results. Panel (a) shows the power curves $\beta_{0}$ and $\beta_{1}$ for the planning edge as a function of $n$, with vertical lines marking $n^{\star}_{\mathcal{H}_{0}}$ and $n^{\star}_{\mathcal{H}_{1}}$. Panel (b) shows the corresponding FPR and FNR error rates. At $n^{\star}_{\mathcal{H}_{0}} = 295$, all edges in $\mathcal{E}^{+}$ achieve the target power under $\mathcal{H}_1$ and 80\% achieve target power under $\mathcal{H}_0$, confirming that the monotonicity guarantee holds approximately across edges stronger than $\rho_{\min}$. At $n^{\star}_{\mathcal{H}_{1}} = 191$, the target power under $\mathcal{H}_1$ is achieved by $93\%$ of edges in $\mathcal{E}^{+}$, but only 47\% reach the target under $\mathcal{H}_0$, reflecting that $n^{\star}_{\mathcal{H}_{0}} > n^{\star}_{\mathcal{H}_{1}}$ in this example. Figure \ref{fig:bfda_bf_dist} illustrates the role of $\rho_{\min}$ as a signal threshold. Under $\mathcal{H}_0$, edges in $\mathcal{E}^+$ accumulate $\log(BF_{01})$ well above $\log(\gamma)$, with evidence increasing monotonically with $|\rho_{ij}|$ and strengthening from $n = 191$ to $n = 295$, while edges below $\rho_{\min}$ yield weak and inconclusive evidence regardless of sample size. This is a direct consequence of the BFDA: since $n^{\star}$ is optimized for $\rho_{\min}$, edges with $|\rho_{ij}| < \rho_{\min}$ are underpowered at the planned sample size by construction, and achieving target power for those edges would require a larger study. Under $\mathcal{H}_1$, the pattern reverses: $\log(BF_{01})$ falls well below $-\log(\gamma)$ for edges in $\mathcal{E}^+$, confirming strong evidence for conditional dependence, while edges below $\rho_{\min}$ again remain inconclusive. The $\times$ symbols indicate edges where the Monte Carlo estimate of the log Bayes factor collapsed numerically to values lower than $-10\log\left(10\right)\approx -23$. 

The DPIR and BFDA recommendations for the toy example --- $n^{\star}_{\text{global}} = 117$, $\max_{(i,j) \notin \text{diag}} n^{\star}(\theta_{ij}) = 148$, and $n^{\star} = 295$ --- illustrate how the two strategies target complementary aspects of the design problem and should be interpreted jointly. The DPIR recommends $n^{\star}_{\text{global}}$ as the minimum sample size for the data to dominate the prior on average across all precision parameters, and $\max_{(i,j) \notin \text{diag}} n^{\star}(\theta_{ij})$ as the more conservative guarantee that the data dominate the prior for every individual off-diagonal precision parameter; the researcher adopts either $n^{\star}_{\text{global}}$ or $\max_{(i,j) \notin \text{diag}} n^{\star}(\theta_{ij})$ as $n^{\star}_{\text{DPIR}}$, depending on whether average of per-parameter prior domination is required. The BFDA recommends $n^{\star} = \max(n^{\star}_{\mathcal{H}_0}, n^{\star}_{\mathcal{H}_1})$ as the sample size for edge-specific testing, calibrated to the weakest edge in the detectable set $\mathcal{E}^+$ and conservative for all stronger edges according to the monotonicity of $n^{\star}(\rho)$. A researcher concerned with prior domination uses the DPIR; one focused on edge detection uses the BFDA. When both guarantees are desired, the recommended sample size is $n^{\star} = \max(n^{\star}_{\text{DPIR}}, n^{\star}_{\text{BFDA}})$. In all cases, the prior ESS provides a useful reference: a planned $n^{\star}$ substantially smaller than the ESS suggests that the prior may dominate the posterior. In the toy example, $n^{\star} = 295$ exceeds both $\text{ESS}_{\text{VR}} = 125.5$ and $\text{ESS}_{\text{PR}} = 81.4$, confirming that the planned study carries sufficient information to dominate the prior and that the posterior will be driven primarily by the data rather than the prior elicitation. How the ratios $n^{\star}_{\text{DPIR}}/\text{ESS}$ and $n^{\star}_{\text{BFDA}}/\text{ESS}$ vary across prior specifications and network sizes is examined systematically in the simulation study of Section \ref{sec:simulation_studies}.

\section{Simulation studies}
\label{sec:simulation_studies}

Two simulation studies are conducted to evaluate the ESS methods and their use for sample size planning in GGMs. The first simulation study
examines how $ESS_{\text{VR}}$ and $ESS_{\text{PR}}$ behave under the Wishart and \textit{G}-Wishart priors, with particular attention to the density-dependent behavior and the Jensen gap between the two measures. The second study evaluates the DPIR and BFDA planning procedures under correctly specified and misspecified priors, assessing the planned sample sizes in relation to the prior ESS and their sensitivity to misspecified planning edges, across different prior study sizes, graph topologies and network sizes.

A common feature of both studies is the procedure used to generate prior precision matrices that reflect realistic elicitation scenarios. Rather than drawing directly from a Wishart or \textit{G}-Wishart distribution with arbitrarily chosen parameters, we simulate precision matrices representative of what a researcher might elicit from a prior study of size $\nu$. For each combination of $\nu$ and $p$, a well-conditioned precision matrix $\bm{\Psi}_0$ is first constructed via a rescaled orthonormal basis, where marginal variances $\sigma_i$ are drawn independently from $U(0.1,5)$. This range reflects the heterogeneity in variable scales typical of a prior study on the original measurement scale, inducing partial correlations of varying magnitude without imposing a specific distributional assumption on the dependence structure. A dataset of $\nu$ observations is then drawn from a multivariate normal with covariance $\bm{\Psi}_0^{-1}$, and the sample-based precision matrix $\bm{\Psi}$ estimated from these observations represents the prior precision, with entries on a scale proportional to a study of size $\nu$. In the structured case, where the precision matrix must respect the zero pattern of a graph $G$, an additional projection step \citep{Hastie2009} is applied before and after data generation to enforce the graph structure. The generative procedure shared across both simulation studies is formalized in 
Algorithm \ref{algorithm:random_precision}.

\subsection{Simulation Study 1: The Prior ESS}
\label{subsec:simulation_study_prior_ess}

The first simulation study investigates the behavior of the prior ESS across a range of conditions. The design includes the two prior distributions (Wishart, \textit{G}-Wishart), two ESS methods (VR, PR), four prior sample sizes ($\nu \in \{25, 50, 100, 200\}$), and four network sizes ($p \in \{10, 20, 30, 40\}$). For the \textit{G}-Wishart case, three graph generation models are considered: the Erdős-Rényi model \citep[random graph model;][]{Renyi1959}, the Barabási-Albert model \citep[scale-free model;][]{Albert2002}, and the Watts-Strogatz model  \citep[small-world graph model;][]{Watts1998}. For each condition, we consider $R = 1,000$ replications.

In each replication, a prior precision matrix $\bm{\Psi}$ is generated using Algorithm \ref{algorithm:random_precision}. In the Wishart case, $G$ is set to a complete graph. In the \textit{G}-Wishart case, the adjacency matrix $G$ is generated from one of the three graph models. The random graph model generates edges independently with a common probability drawn from $\text{Uniform}(0,1)$. The scale-free model draws the number of edges attached per node from $\text{Uniform}\{1, \lfloor p/2 \rfloor\}$ and generates a scale-free graph via preferential attachment, producing hub-dominated networks with a power-law degree distribution. The small-world graph model draws the number of node-wise neighbors from $\text{Uniform}\{1, \lfloor(p-1)/2\rfloor\}$ and iteratively increases the rewiring probability from $0.2$ until the small-worldness index $S^{\Delta} > 1$ \citep{Humphries2008}. The Supplementary Material reports descriptive summaries of the simulated precision matrices, including partial correlations, marginal variances, and log condition number $\log\kappa(\bm{K})$, confirming that all matrices are well-conditioned ($\log\kappa(\bm{K}) < 8$) under all conditions.

For each generated $\bm{\Psi}$, both the global ESS and the parameterwise ESS are computed using the VR and PR methods, with the prior hyperparameters elicited as $\bm{\Psi}^{*} = \bm{\Psi} / \nu$ for the Wishart and $\bm{\Psi}^{*}_G = \bm{\Psi}_G / \nu$ for the \textit{G}-Wishart, following the elicitation procedure described in Appendices \ref{subsection:wishart_prior} and \ref{subsection:gwishart_prior}. The network density, defined as the proportion of present edges over the total number of possible edges $\binom{p}{2}$, is recorded for each replication. As shown in Appendix \ref{appendix:prior_ess_formulas}, the prior ESS under the MTM, PT, and ELIR methods is a linear function of $\nu$ and $p$, and therefore does not vary with network density, precision matrix, or graph structure. These methods are therefore omitted from the analysis, which focuses on the density- and structure-dependent behavior of $ESS_{\text{VR}}$ and $ESS_{\text{PR}}$.

\begin{table}[t]
\centering
\begin{tabular}[t]{lccccc}
\toprule
Method & $\nu$ & $p = 10$ & $p = 20$ & $p = 30$ & $p = 40$\\
\midrule
\addlinespace[0.3em]
\multicolumn{6}{l}{$ESS_{\text{VR}} / \nu$}\\
\hspace{1em} & 25 & 1.046 & 1.043 & - & -\\
\hspace{1em} & 50 & 1.023 & 1.022 & 1.021 & 1.021\\
\hspace{1em} & 100 & 1.012 & 1.011 & 1.011 & 1.010\\
\hspace{1em} & 200 & 1.006 & 1.005 & 1.005 & 1.005\\
\addlinespace[0.3em]
\multicolumn{6}{l}{$ESS_{\text{PR}} / \text{factor}$}\\
\hspace{1em} & 25 & 0.920 & 0.744 & - & -\\
\hspace{1em} & 50 & 0.970 & 0.963 & 0.945 & 0.887\\
\hspace{1em} & 100 & 0.987 & 0.986 & 0.985 & 0.982\\
\hspace{1em} & 200 & 0.994 & 0.994 & 0.994 & 0.993\\
\bottomrule
\end{tabular}
\caption{\label{tab:appendix_ess_ratios} Mean of $\text{ESS}_{VR}/\nu$ and $\text{ESS}_{PR}/\text{factor}$ across replications under the Wishart prior. Dashes indicate conditions where $\nu \leq p + 3$.}
\end{table}

\paragraph{Wishart prior.} Table \ref{tab:appendix_ess_ratios} reports the mean of $ESS_{\text{VR}}/\nu$ and $ESS_{\text{PR}}/\text{factor}$, with $\text{factor} = (\nu - p)(\nu - p - 1)(\nu - p - 3)/\nu(\nu - p - 2)$, across replications under the Wishart prior. For VR, the ratio $ESS_{\text{VR}}/\nu$ converges monotonically toward 1 from above as $\nu$ increases, showing no sensitivity to $p$ and with little variability across replications, confirming the asymptotic result $ESS_{\text{VR}} \sim \nu$ derived in Appendix \ref{appendix:prior_ess_formulas}. For PR, the ratio $ESS_{\text{PR}}/\text{factor}$ is always strictly below 1 and also converges toward 1 as $\nu$ increases, but the convergence is slower and more sensitive to $p$: at small $\nu$ relative to $p$ the ratio drops substantially, reaching $0.744$ for $\nu = 25$ and $p = 20$, confirming that the factor characterizes the joint dependence of $ESS_{\text{PR}}$ on $\nu$ and $p$.

\begin{figure}[t]
    \centering
    \includegraphics[width=0.9\linewidth]{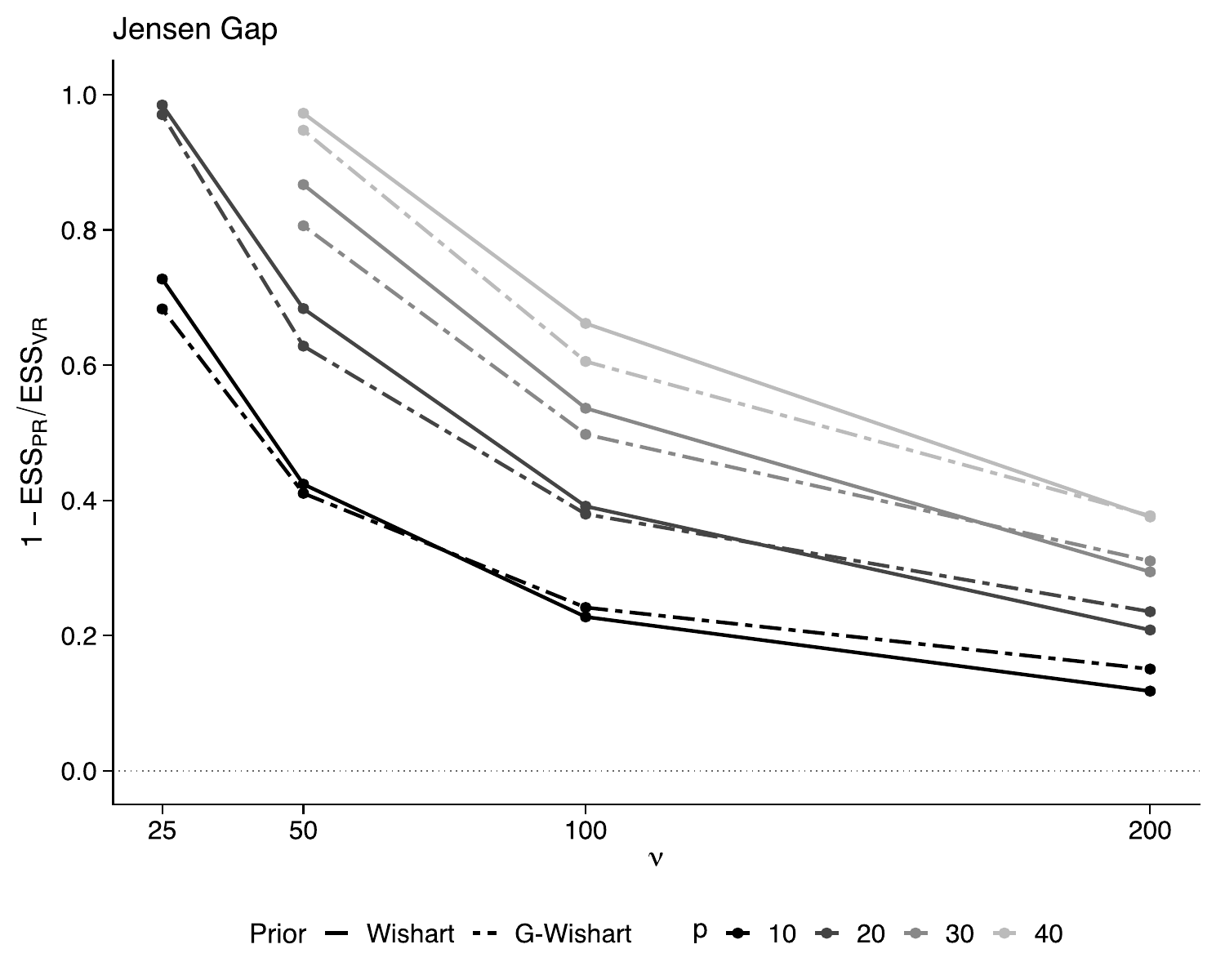}
    \caption{Jensen gap $J = 1 - ESS_{\text{PR}}/ESS_{\text{VR}}$ as a function of $\nu$, for the Wishart (solid) and \textit{G}-Wishart (dashed) priors, colored by $p$. The gap decreases with $\nu$ and increases with $p$ for both priors.}
    \label{fig:jensen_gap}
\end{figure}

\begin{figure}[t]
    \centering
    \includegraphics[width=0.9\linewidth]{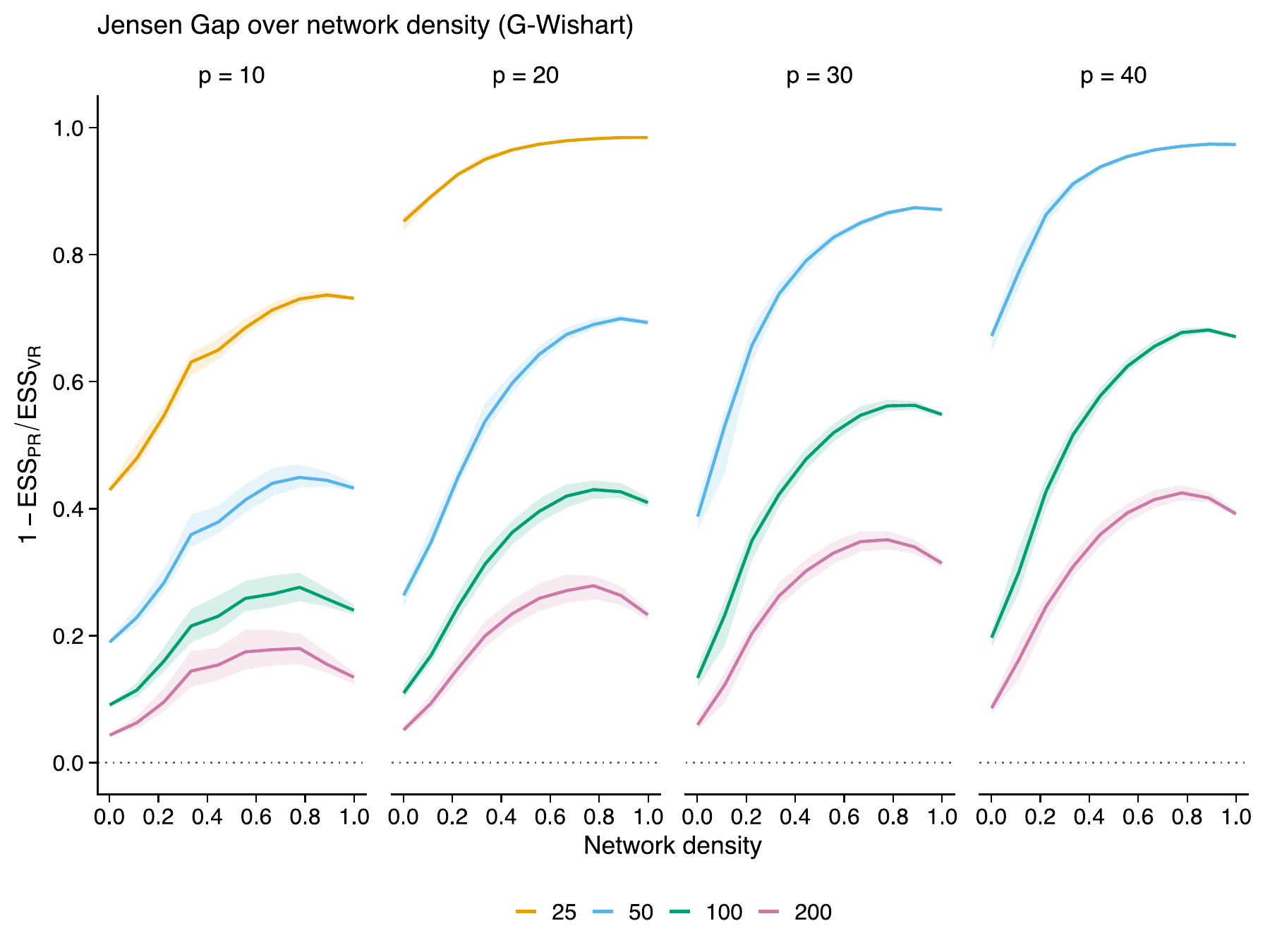}
    \caption{Jensen gap $J = 1 - ESS_{\text{PR}}/ESS_{\text{VR}}$ plotted against network density across network sizes $p$ (columns) and prior study sizes $\nu$ (color); the dotted line marks the Jensen lower bound at $0$.}
    \label{fig:jensen_gap_density}
\end{figure}

\paragraph{Jensen gap.} Figure \ref{fig:jensen_gap} shows the Jensen gap $J = 1 - ESS_{\text{PR}}/ESS_{\text{VR}}$ as a function of $\nu$, for the Wishart (solid) and \textit{G}-Wishart (dashed) priors, colored by $p$. The gap increases with $p$ for fixed $\nu$, reflecting that larger networks amplify the parameter dependencies in $\bm{\Theta}$. By contrast, the gap decreases monotonically with $\nu$ for both priors and all $p$, confirming that a larger prior study size reduces the variability of $\mathcal{I}(\mathbf{x}_1;\bm{\Theta})$ under the prior, though the two priors order differently: the \textit{G}-Wishart gap is smaller than the Wishart one for small $\nu$, but the relationship reverses at larger $\nu$, with the crossover occurring earlier for smaller $p$. Even at $\nu = 200$, however, the gap remains substantial for large $p$, indicating that the VR and PR measures can differ considerably under realistic elicitation conditions. For the \textit{G}-Wishart prior, Figure \ref{fig:jensen_gap_density} shows the gap against network density: at fixed $p$ and $\nu$, $J$ rises as the graph fills in, but the increase is non-monotone, peaking between intermediate and high density before decreasing at the highest densities. Since $J$ combines the two effective sample sizes, $ESS_\text{VR}$ which approaches $\nu$ from above and $ESS_\text{PR}$ which approaches $\nu$ from below, this non-monotonicity originates in how the two measures respond to connectivity. To investigate this, we next examine $ESS_{\text{VR}}/\nu$ and $ESS_{\text{PR}}/\nu$ separately as functions of network density.

\begin{figure}[t]
    \centering
    \includegraphics[width=0.9\linewidth]{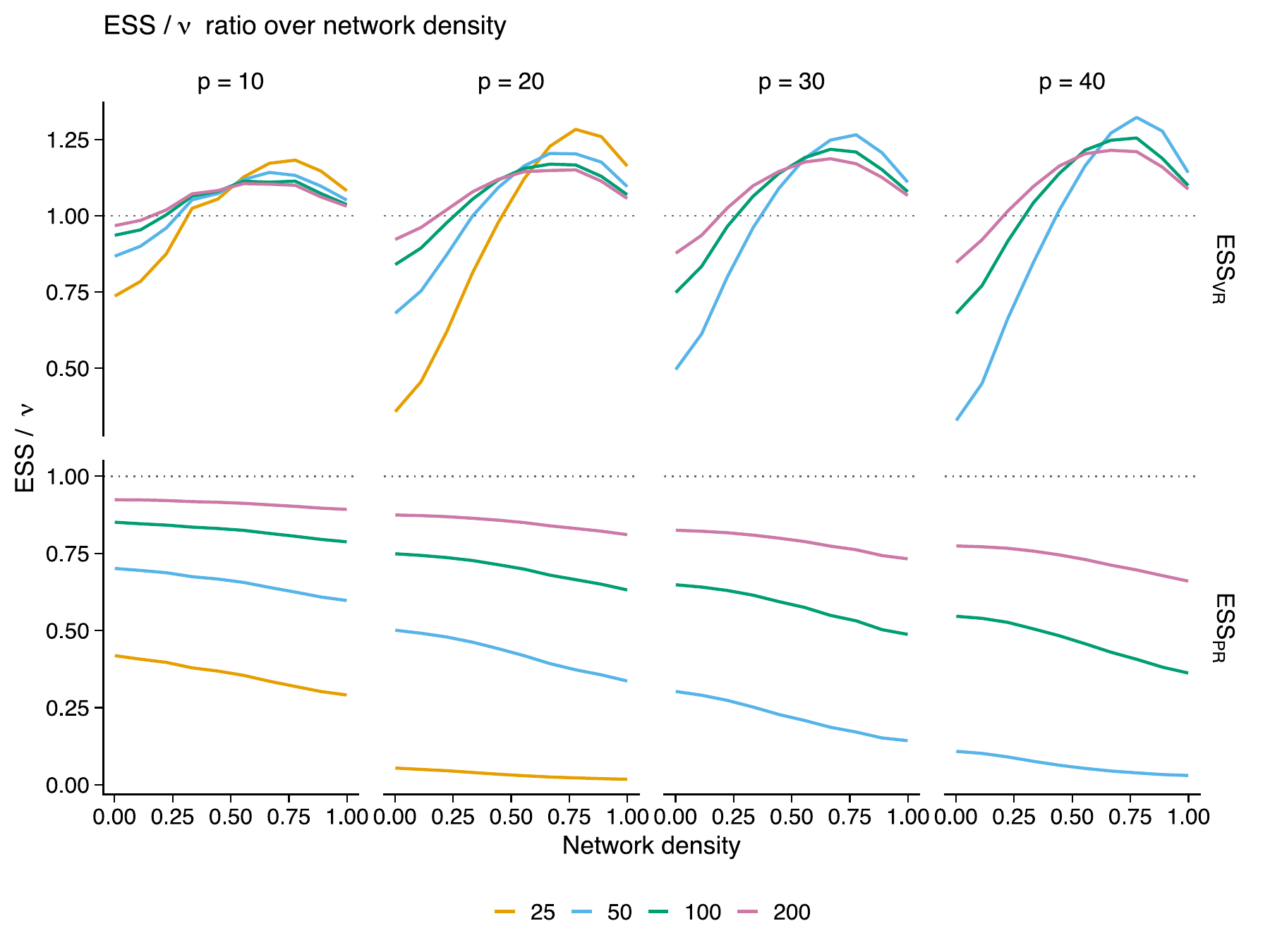}
    \caption{$\text{ESS}/\nu$ ratio plotted against network density across network sizes $p$ (columns) and prior study sizes $\nu$ (color), for the $ESS_{\text{VR}}$ (top) and $ESS_{\text{PR}}$ (bottom) estimators. The dotted horizontal line marks $\text{ESS}/\nu = 1$, where the prior contributes as many effective observations as the prior study size.}
    \label{fig:ESS_to_nu_ratio_vs_density}
\end{figure}

\paragraph{G-Wishart prior: density-dependent behavior.} Figure \ref{fig:ESS_to_nu_ratio_vs_density} shows $ESS_{\text{VR}}/\nu$ and $ESS_{\text{PR}}/\nu$ as a function of network density, across $p$ (columns) and $\nu$ (color). The two methods exhibit a contrasting behavior. $ESS_{\text{PR}}/\nu$ decreases monotonically with density for all $(p, \nu)$ combinations, reflecting that denser networks have more parameter dependencies, which reduce the effective number of prior observations relative to $\nu$. The decrease is more pronounced at small $\nu$ relative to $p$, with $ESS_{\text{PR}}/\nu$ approaching zero for $\nu = 25$ and large $p$. $ESS_{\text{VR}}/\nu$ exhibits a  different non-monotone pattern: it rises above 1 for sparse networks, peaks around density $0.5 - 0.8$, and then decreases for denser networks. The crossover point where $ESS_{\text{VR}}/\nu$ first exceeds 1 shifts toward sparser networks as $\nu$ increases, so that for large $\nu$ the prior informativeness amplifies even at low density. This density-dependent behavior is not influenced by graph topology:  random, scale-free, and small-world graphs produce indistinguishable patterns at the same density level. In the following paragraph we examine the drivers of the relationship between $ESS/\nu$ and density. Results stratified by graph structure and the parameterwise counterparts are available in the Supplementary Material. 

\begin{figure}[t]
    \centering
    \includegraphics[width=0.9\linewidth]{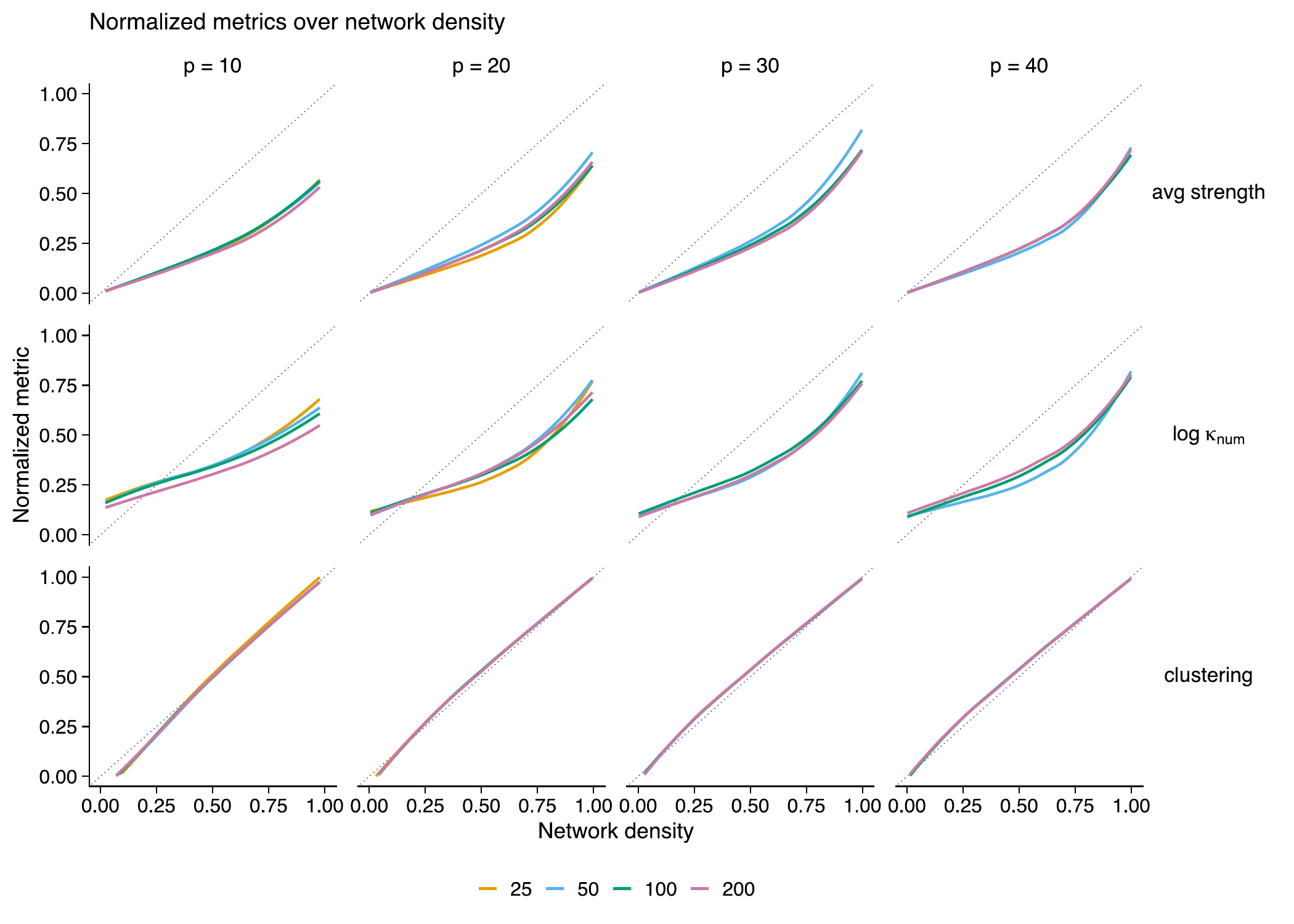}
    \caption{Network metrics normalized to $[0,1]$ within each $(p, \nu)$ condition, plotted against network density across network sizes $p$ (columns) and prior study sizes $\nu$ (color). The dotted bisector marks perfect linear tracking of density.}
    \label{fig:metric_vs_density}
\end{figure}

\paragraph{Drivers of the density-dependent behavior.} To investigate the drivers of $ESS_{\text{VR}}/\nu$ and $ESS_{\text{PR}}/\nu$, we consider three network metrics. The \emph{average strength}, $\text{avg strength} = p^{-1}\sum_{i=1}^p \sum_{j \in \mathcal{N}(i)} |\rho_{ij}|$, combines degree and partial correlation magnitude and is geometry-dependent. The \emph{log condition number} $\log\kappa_{\text{num}}$ measures the ill-conditioning of the numerator matrix of each ESS measure: $\mathbb{E}_{\bm{\Theta}}[\mathcal{I}^{-1}(\mathbf{x}_1;\bm{\Theta})]$ for $ESS_{\text{VR}}$ and $\mathbb{V}^{-1}(\bm{\Theta})$ for $ESS_{\text{PR}}$, and is also geometry-dependent. The log condition numbers of the denominators and the log condition number of the numerator matrix of $ESS_{\text{PR}}$, are omitted as they show an identical relationship with density; only $\log\kappa_{\text{num}}$ for $ESS_{\text{VR}}$ is shown. The \emph{clustering coefficient}, defined as the ratio of closed triangles to connected triples in $G$, is topology-dependent and invariant to $\bm{\Psi}$. Figure \ref{fig:metric_vs_density} shows three network metrics: average strength, $\log\kappa_{\text{num}}$, and clustering coefficient, that are normalized to $[0,1]$ within each $(p, \nu)$ condition and plotted against network density. Clustering tracks density almost perfectly across all $p$ and $\nu$ values, lying close to the bisector (dotted line) regardless of prior study size, confirming that it is essentially a deterministic function of density and a topological driver of $ESS/\nu$. The geometry-dependent metrics deviate substantially from the bisector, yet their relationship with $ESS/\nu$ mirrors that of density, including the non-monotone bump of $ESS_{\text{VR}}/\nu$ at intermediate densities (see the Supplementary Material), indicating that partial correlation strength and matrix conditioning contribute information beyond graph topology alone. We therefore recommend reporting both measures: $ESS_{\text{PR}}$ as a structural baseline robust to uncertainty in the elicited $\bm{\Psi}$, and $ESS_{\text{VR}}$ as a geometry-sensitive indicator. The gap between the two measures quantifies the degree to which the elicited $\bm{\Psi}$ amplifies prior informativeness beyond what graph topology alone implies.

\subsection{Simulation Study 2: Sample Size Planning}
\label{subsec:simulation_study_sample_size_planning}

In this simulation study, we evaluate the two complementary planning strategies described in Section \ref{sec:sample_size_planning} across two designs. The first validates the planning strategy across a range of prior study sizes, network dimensions, and graph structures, examining how the planned sample size behaves under a correctly specified prior. The second examines the sensitivity of the planned sample size to misspecification of the prior graph structure. Both designs use the same prior precision matrices generated in the simulation study of Section \ref{subsec:simulation_study_prior_ess}, so that the planning recommendations can be directly related to the corresponding prior ESS measures. For each condition, we use $R = 500$ replications, corresponding to the first five of the ten batches of $100$ replications. This reduction is motivated by computational cost, as the planning strategies require iterative computation across a grid of candidate sample sizes.

\begin{figure}[t]
    \centering
    \includegraphics[width=0.9\linewidth]{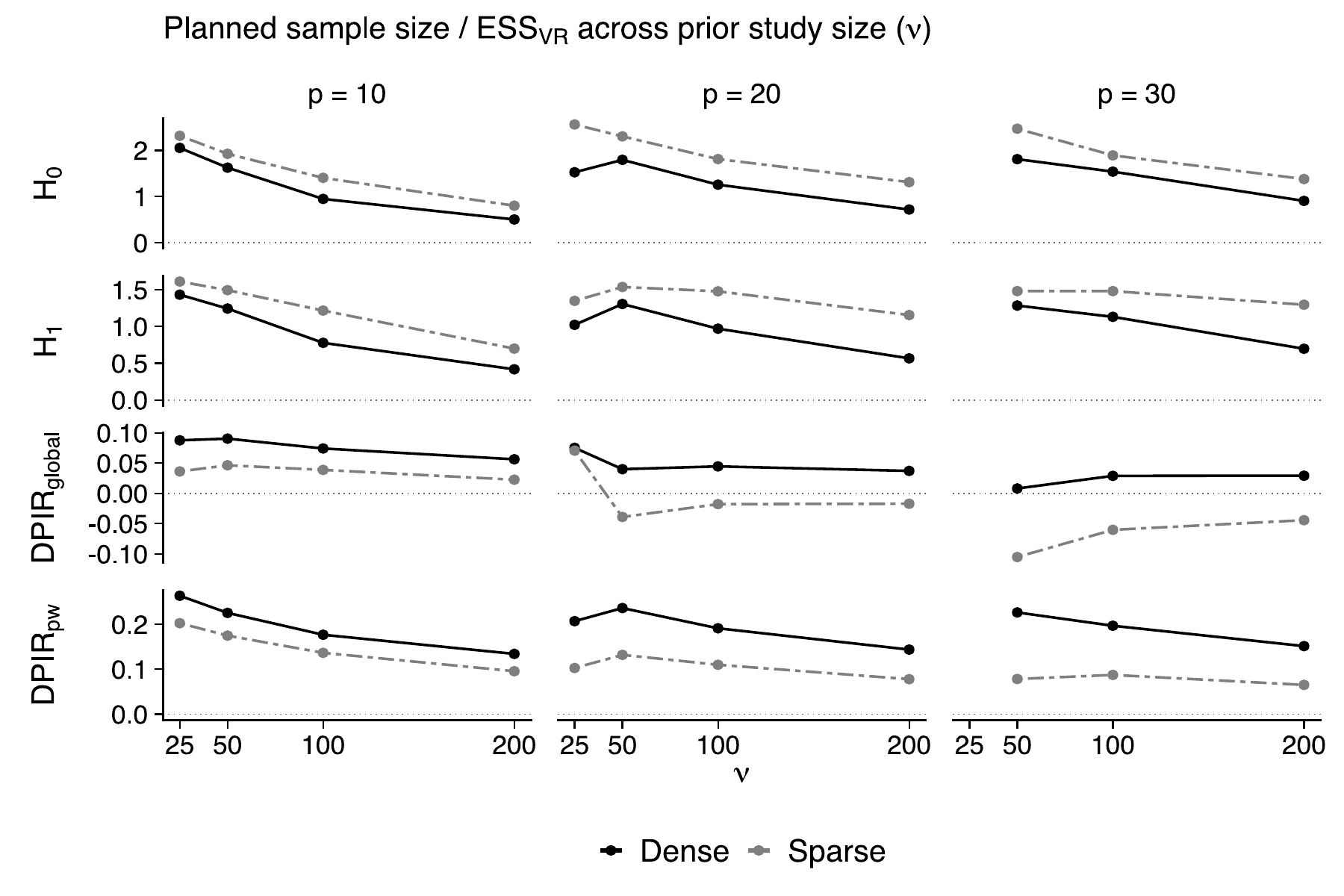}
    \caption{$\log_{10}(n^{\star}/ESS_\text{VR})$ plotted against 
    prior study size $\nu$ across network sizes $p$ (columns) and sample size 
    planning methods (rows), for dense (solid black) and sparse (dashed gray) priors. 
    The dotted horizontal line marks $\log_{10} = 0$, where $n^{\star} = ESS_\text{VR}$.}
    \label{fig:nstar_ess_vr_log_ratio_thr10}
\end{figure}

\paragraph{Validation of the Planning Strategy.} The planning strategy is validated across all graph structures (complete, random, small-world, and scale-free), crossing $p \in \{10, 20, 30\}$ and $\nu \in \{25, 50, 100, 200\}$, yielding 12 conditions per structure. The case $p = 40$ is excluded given its substantially higher computational cost compared to the smaller network sizes. For each replicate, both the DPIR and BFDA are applied following Section \ref{sec:sample_size_planning}, with the BFDA evaluated under two evidence thresholds $\gamma \in \{3, 10\}$. If $\rho_{\min} < \epsilon = 0.05$ the replicate is flagged as unplannable. All BFDA results reported in the main text use an evidence threshold of $\gamma = 10$. Results for $\gamma = 3$ are qualitatively similar and are provided in the Supplementary Material.

Figure \ref{fig:nstar_ess_vr_log_ratio_thr10} plots $\log_{10}(n^{\star}/\text{ESS}_{\text{VR}})$ against $\nu$ for each planning method and network size $p$, providing a direct comparison between planned sample sizes and prior informativeness. Two patterns hold for both planning methods: the ratio decreases monotonically with $\nu$, reflecting that more informative priors yield smaller planned sample sizes relative to $\text{ESS}_{\text{VR}}$, and the dense and sparse conditions remain clearly separated throughout. However, the direction of this separation reverses between methods: for the BFDA, sparse priors require larger samples relative to their $\text{ESS}_{\text{VR}}$ than dense priors, whereas for both DPIR criteria the ordering is inverted. The two methods also differ by orders of magnitude, highlighting a fundamental difference in how they relate to prior informativeness. For the BFDA, planned sample sizes exceed $\text{ESS}_{\text{VR}}$ by factors ranging from $3$--$24$ at $\nu = 200$ to $11$--$364$ at $\nu = 25$, with $\log_{10}$ ratios remaining well above zero even at large $\nu$, driven by the requirements of the Bayes factor threshold. For the DPIR methods, ratios are of a much smaller magnitude and are not uniformly positive: the planned sample size $n^{\star}_{\text{global}}$ exceeds
$\text{ESS}_{\text{VR}}$ by a factor of $1.02$--$1.23$ for dense priors, whereas for sparse priors it ranges from $0.78$ to $1.18$, falling below $\text{ESS}_{\text{VR}}$ at $p = 20$ and $p = 30$. The parameterwise $\max_{(i,j) \notin \text{diag}} n^{\star}(\theta_{ij})$ is consistently more demanding, exceeding $\text{ESS}_{\text{VR}}$ by a factor of $1.36$--$1.84$ for dense priors and $1.16$--$1.36$ for sparse priors.

\begin{table}[t]
\centering
\small
\begin{tabular}{llcccccccccccc}
\toprule
 & & \multicolumn{4}{c}{$p = 10$} & \multicolumn{4}{c}{$p = 20$} & \multicolumn{4}{c}{$p = 30$} \\
\cmidrule(lr){3-6} \cmidrule(lr){7-10} \cmidrule(lr){11-14}
 & $\nu$ & $25$ & $50$ & $100$ & $200$ & $25$ & $50$ & $100$ & $200$ & $25$ & $50$ & $100$ & $200$ \\
\midrule
\multirow{2}{*}{Dense} & $\mathcal{H}_0$ & 0.00 & 0.00 & 0.08 & 0.44 & 0.00 & 0.00 & 0.00 & 0.28 & -- & 0.00 & 0.00 & 0.16 \\
 & $\mathcal{H}_1$ & 0.23 & 0.32 & 0.46 & 0.61 & 0.29 & 0.24 & 0.37 & 0.53 & -- & 0.17 & 0.28 & 0.46 \\
\midrule
\multirow{2}{*}{Sparse} & $\mathcal{H}_0$ & 0.00 & 0.00 & 0.00 & 0.13 & 0.40 & 0.17 & 0.06 & 0.05 & -- & 0.36 & 0.24 & 0.19 \\
 & $\mathcal{H}_1$ & 0.21 & 0.24 & 0.32 & 0.47 & 0.50 & 0.28 & 0.25 & 0.32 & -- & 0.42 & 0.33 & 0.32 \\
\bottomrule
\end{tabular}
\caption{\label{tab:power_dpir_global_thr10} Median BFDA power under $\mathcal{H}_0$ and $\mathcal{H}_1$ evaluated at the global DPIR recommended sample size $n^{\star}_{\text{global}}$, across prior study sizes $\nu$ (columns, nested within $p$). Bayes factor threshold $= 10$.}
\end{table}

Table \ref{tab:power_dpir_global_thr10} cross-evaluates the DPIR planning strategy by reporting median BFDA power under $\mathcal{H}_0$ and $\mathcal{H}_1$ at the DPIR-recommended sample size $n^{\star}_{\text{global}}$. Under $\mathcal{H}_1$, power is low throughout, indicating that $n^{\star}_{\text{global}}$ is generally insufficient to correctly include truly present edges. For dense priors, power increases with $\nu$ across all network sizes (e.g., from $0.23$ to $0.61$ at $p = 10$); for sparse priors, power also increases with $\nu$ at $p = 10$ (from $0.21$ to $0.47$) but declines at $p = 20$ (from $0.50$ to $0.32$) and $p = 30$ (from $0.42$ to $0.32$). Under $\mathcal{H}_0$, dense priors yield power near zero for all prior sample sizes except at $\nu = 200$, where it reaches $0.44$, $0.28$, and $0.16$ for $p = 10$, $20$, and $30$ respectively; sparse priors show the opposite trend at larger networks, with power already elevated at small $\nu$ ($0.40$ at $\nu = 25$ for $p = 20$, $0.36$ at $\nu = 50$ for $p = 30$) and declining monotonically as $\nu$ grows. Together, these results confirm that the DPIR and BFDA target fundamentally different inferential objectives and should be treated as complementary rather than interchangeable planning tools. Corresponding tables for the pairwise DPIR variant $\max_{(i,j) \notin \text{diag}} n^{\star}(\theta_{ij})$ are provided in the Supplementary Material.

\paragraph{Sensitivity to Prior Misspecification.} The sensitivity of the planning strategy to misspecification of the prior graph structure is examined for $p \in \{10, 20\}$ under the random graph structure (excluding $p \in \{30, 40\}$ for computational reasons). The motivation is that prior studies are often underpowered, so edges with small partial correlations may be incorrectly included or excluded in the elicited prior. Uncertain edges are defined as structural zeros and weak present edges whose frequentist warm-start sample size exceeds $\nu$, reflecting edges whose presence or absence in the prior graph is ambiguous. For each replicate, $S = 20$ misspecification scenarios are generated by independently flipping each uncertain edge with probability $0.5$, producing a misspecified graph $G_{\text{mis}}$, constrained to be neither empty nor complete. A new precision matrix $K_{\text{mis}}$ is then drawn from the G-Wishart distribution with graph $G_{\text{mis}}$, scale $K/\nu$ and $\nu$ degrees of freedom, and the full planning strategy is re-applied under the misspecified prior. The sensitivity of the planning recommendations is then assessed by comparing the planned sample size under the misspecified prior to the reference, examining how the probability of inflated planned sample size varies over the induced shift in the planning edge partial correlation.

\begin{figure}[t]
    \centering
    \includegraphics[width=0.9\linewidth]{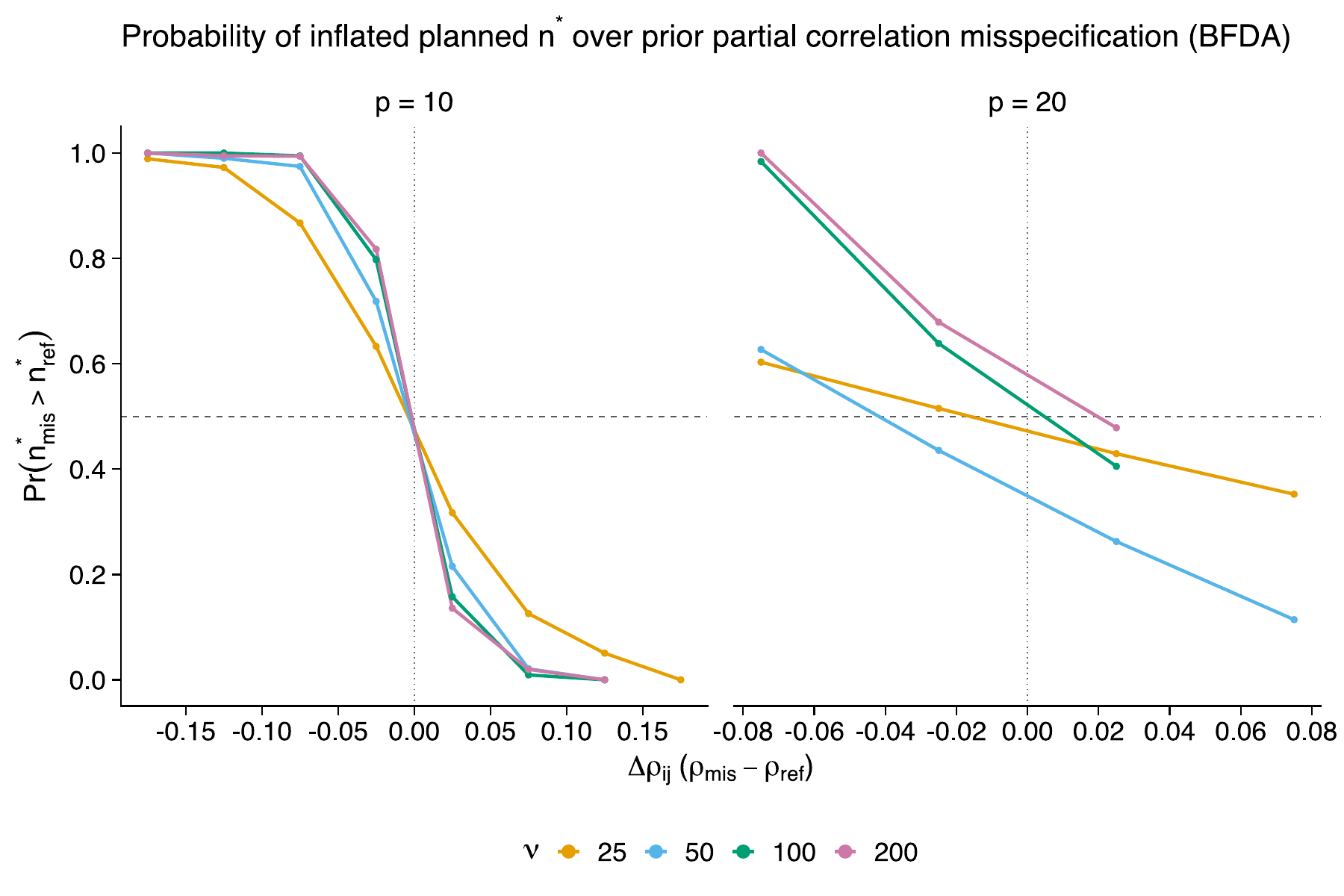}
    \caption{$\Pr(n^{\star}_{\text{mis}} > n^{\star}_{\text{ref}})$ for the BFDA planning strategy as a 
    function of the planning edge shift $\Delta\rho_{ij} = \rho_{\text{mis}} - 
    \rho_{\text{ref}}$, across prior study sizes $\nu$ (color), network sizes 
    $p$ (columns). The dashed line marks $\Pr = 0.5$ 
    and the dotted line marks $\Delta\rho = 0$. Bins of width $0.05$; outliers 
    removed per condition using the $1.5 \times \text{IQR}$ rule.}
    \label{fig:misspec_exceed_bfda}
\end{figure} 

Figure \ref{fig:misspec_exceed_bfda} plots $\Pr(n^{\star}_{\text{mis}} > n^{\star}_{\text{ref}})$ as a function of the planning edge shift $\Delta\rho_{ij}$, where $n^{\star}$ denotes the BFDA-planned sample size defined as $\max(n^{\star}_{\mathcal{H}_0}, n^{\star}_{\mathcal{H}_1})$ under both the misspecified and reference prior. Since larger partial correlations require smaller sample sizes, the direction of the effect is expected: underestimation ($\Delta\rho_{ij} < 0$) inflates the planned sample size relative to the reference, while overestimation ($\Delta\rho_{ij} > 0$) deflates it. What the figure shows is the velocity of this transition. For $p = 10$, $\Pr(n^{\star}_{\text{mis}} > n^{\star}_{\text{ref}})$ transitions sharply from $1$ to $0$ around $\Delta\rho_{ij} = 0$, with even small deviations from the reference edge strength being sufficient to inflate or deflate the planned sample size. For $p = 20$, the transition is more gradual and the observed range of $\Delta\rho_{ij}$ is narrower, reflecting the smaller magnitude of uncertain edge shifts in larger networks; most lines also lie above $0.5$ in the bins closest to zero, suggesting an upward bias in the planned sample size under misspecification. The effect of $\nu$ is more evident at $p = 20$, where larger prior study sizes tend to produce higher inflation probabilities across the observed range of $\Delta\rho_{ij}$. 

\begin{figure}[t]
    \centering
    \includegraphics[width=0.9\linewidth]{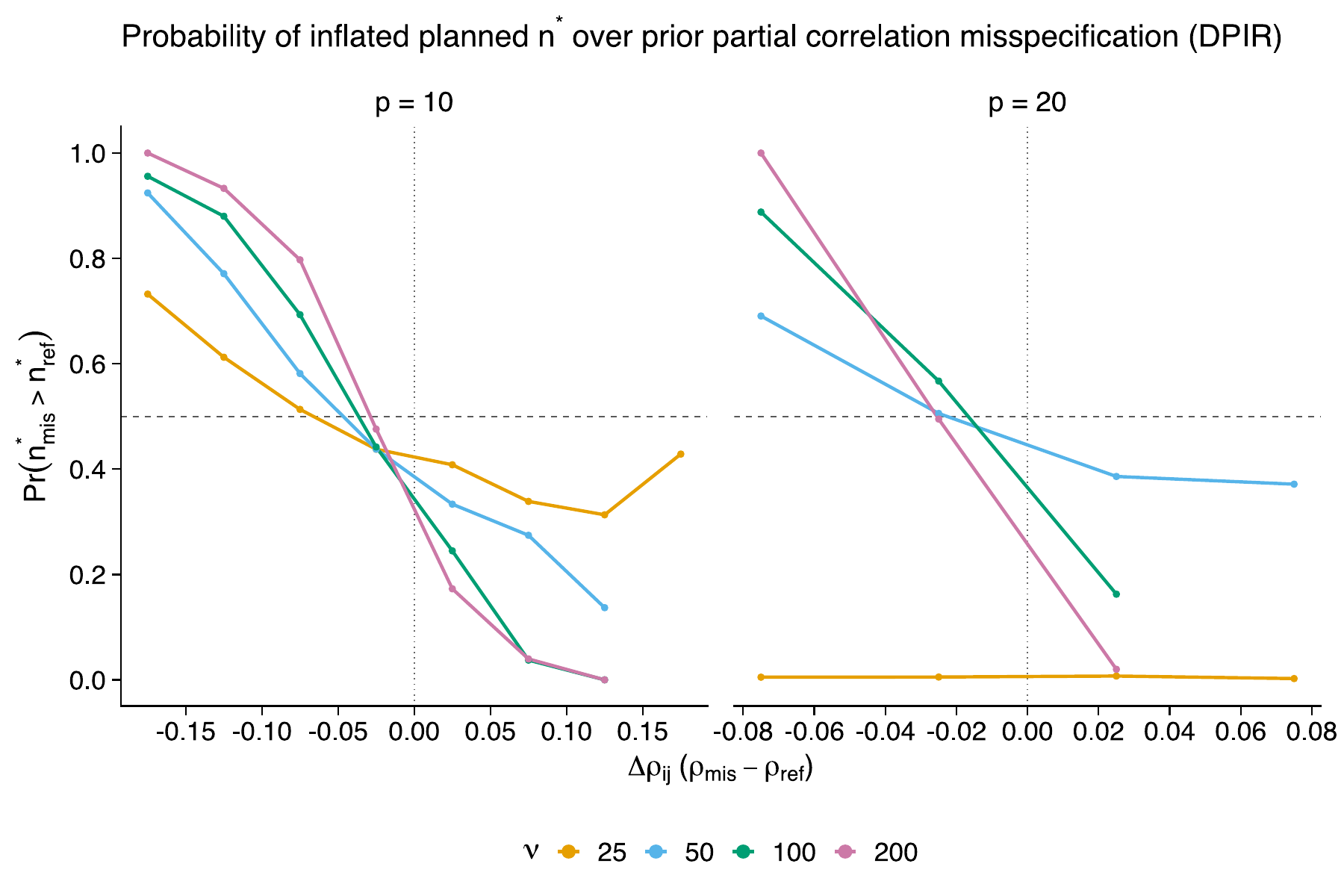}
    \caption{$\Pr(n^{\star}_{\text{mis}} > n^{\star}_{\text{ref}})$ for the DPIR planning strategy, where $n^{\star} = \max(n^{\star}_{\text{global}}, \max_{(i,j) \notin \text{diag}} n^{\star}(\theta_{ij}))$, as a function of the planning edge shift $\Delta\rho_{ij} = \rho_{\text{mis}} - \rho_{\text{ref}}$, across prior study sizes $\nu$ (color) and network sizes $p$ (columns). The dashed line marks $\Pr = 0.5$ and the dotted line marks $\Delta\rho = 0$. Bins of width $0.05$; outliers removed per condition using the $1.5 \times \text{IQR}$ rule.}
    \label{fig:misspec_exceed_dpir}
\end{figure} 

Figure \ref{fig:misspec_exceed_dpir} shows the corresponding analysis for the DPIR planning strategy, where $n^{\star}$ is defined as $\max(n^{\star}_{\text{global}}, \max_{(i,j) \notin \text{diag}} n^{\star}(\theta_{ij}))$. As for the BFDA, the probability of inflation decreases with $\Delta\rho_{ij}$ and approaches $1$ at large negative perturbations, but the sensitivity to misspecification depends strongly on the prior study size. For $\nu \geq 50$ the transition around $\Delta\rho_{ij} = 0$ is steep, and at $\nu = 200$ it is close to a step function, with the probability moving from $1$ to near $0$ over the full range of $\Delta\rho_{ij}$. For $\nu = 25$, by contrast, the response is less sharp: the probability remains between $\approx0.3$ and $\approx 0.7$ across the full range at $p = 10$ and stays close to $0$ throughout at $p = 20$, indicating that with a weakly informative prior the planned sample size via DPIR is largely unaffected by individual edge perturbations. 

Altogether, misspecification shifts the planned sample size predictably rather than erratically: underestimated edges tend to produce conservative over-planning, while the risk of under-planning increases with overstated edges.

\section{Discussion}
\label{sec:discussion}

The aim of this paper was to express the informativeness of the Wishart and \textit{G}-Wishart priors for a Gaussian graphical model in observation-equivalent units that can support sample size planning. To this end, we formalized a pre-data effective sample size (ESS) for the multiparametric GGM setting, summarized globally via a determinant ratio and parameterwise via a Cholesky decomposition. We adapted five ESS estimators (VR, PR, MTM, PT, and ELIR) to the precision matrix under both priors. To our knowledge, this is the first time the informativeness of a matrix-variate prior has been made interpretable on the scale of observations, extending a line of work developed almost entirely for univariate conjugate priors \citep{Clarke1996, Morita2008, Neuenschwander2020}.

The five estimators separate into two families whose contrast is the central finding. MTM, PT, and ELIR reduce to simple functions of the degrees of freedom $\nu$ and the dimension $p$ and are invariant to both the scale matrix $\bm{\Psi}$ and the graph $G$, so they cannot reflect how the elicited dependence structure shapes prior information. VR and PR, instead, depend on the full structure of $\bm{\Psi}$ and approach $\nu$ from above and below ($\mathrm{ESS}_{\mathrm{VR}} > \nu > \mathrm{ESS}_{\mathrm{PR}}$). We formalized this discrepancy as a Jensen gap, finding that it is substantial under realistic elicitation and driven by network density rather than graph structure family. We therefore recommend reporting PR as a conservative structural baseline alongside VR as a geometry-sensitive indicator, with the gap between them quantifying how far the elicited precision structure amplifies prior information beyond what topology alone implies. Building on these measures, we introduced two complementary planning procedures: the Data-to-Prior Information Ratio (DPIR), which identifies the sample size at which the data carry more information than the prior and only moderately exceeds the prior ESS, and our GGM extension of Bayes Factor Design Analysis (BFDA), which targets the sample size needed for conclusive evidence about a given edge and exceeds the prior ESS by one to two orders of magnitude. The gap between them is itself informative, since at the sample size, where the data begin to dominate the prior, evidence for individual edges is typically still weak. Given that edge selection and structure learning depend on adequate sample sizes to dominate an informative prior, we recommend reporting the prior ESS together with both planning targets and treating a planned sample below the prior ESS as a signal that the posterior will be shaped more by the elicitation than by the data.

Several limitations should be acknowledged. First, throughout the planning stage we treated the network structure $G$ as fixed and known, reflecting the researcher's prior belief about which edges are absent. Sample size planning that accounts for uncertainty in the structure itself is left open. Second, although the recommended sample size achieves the target power for most edges at or above the strength of the planning edge, we did not formally control the error rate across the many edges tested simultaneously: for the BFDA we relied on fixed evidence thresholds, and whether these should be raised to control the family-wise error rate is left for future work. Third, under the \textit{G}-Wishart prior the expected Fisher information has no closed form and can only be obtained by Monte Carlo; this cost grows with network size and constrained the network sizes we could include in the simulation studies. Finally, every sample size recommendation is conditional on the elicited prior, and because prior studies are often small and underpowered, the elicited structure may itself be misspecified. Our misspecification analysis (Section~\ref{subsec:simulation_study_sample_size_planning}) shows that this dependence is directional rather than arbitrary: underestimating an edge's strength only inflates the planned size, resulting in a conservative planning, whereas overstating it deflates the size and can underpower the study. This evidence is confined to the random structural perturbations we tested. Potential misspecifications regarding the prior scale, the chosen variable set, or systematic biases fall outside the scope of this paper, leaving their impact currently unexamined.

The present work suggests several directions for future research. Because the ESS construction depends only on prior and single-observation information matrices, it extends to other network models, such as ordinal Markov random fields~\citep{Marsman2025} and the graphical vector autoregressive framework~\citep[GVAR;][]{Wild2010,Epskamp2018}. More broadly, the same perspective applies to any informative but intractable prior: natural next steps include generalizing the ESS to matrix-variate mixture priors~\citep{Mulder2018}, incorporating graph uncertainty into planning, and pairing these pre-data tools with post-data ESS diagnostics \citep{Reimherr2021, Wiesenfarth2019, Jones2022}. Computational efficiency could be improved by exploiting graph decomposability: for decomposable graphs, a clique–separator factorization results in closed-form \textit{G}-Wishart computations and removes the need for simulation, while for non-decomposable graphs, breaking the graph into smaller subgraphs would confine simulation to the components that genuinely require it. The full workflow presented in this paper is implemented in the \texttt{designbgm} package~\citep{designbgm}, giving applied researchers an accessible way to quantify what an informative prior is worth before any data are collected, and to plan the sample size of a prospective study. 

\section*{Code and Data Availability}
All code required to reproduce the simulations, figures, and supplementary material in this paper is openly available at \href{https://doi.org/10.5281/zenodo.20767408}{10.5281/zenodo.20767408}. The repository includes the simulation scripts, plotting code, and supplementary materials, along with a record of the \textsf{R} and package versions used. The analyses rely on the \texttt{designbgm} package, available at \url{https://github.com/Bayesian-Graphical-Modelling-Lab/designbgm}.

\section*{Supplementary Material}
Supplementary Material associated with this article can be found at \href{https://doi.org/10.5281/zenodo.20767408}{10.5281/zenodo.20767408}.


\section*{Competing interests}
The authors declare none.

\begin{appendices}

\section{Gradient, Hessian, and Fisher information of the GGM}
\label{appendix:ggm_derivatives}

The log-likelihood of the GGM is proportional to
\[
\log f(\mathbf{X};\bm{\Theta}) \propto \frac{n}{2}\bigg(\log{|\bm{\Theta}|} - \text{tr}(S\bm{\Theta})\bigg)
\]
where $S = \frac{1}{n}\sum_{k=1}^{n}\mathbf{x}_k\mathbf{x}_k^{\top}$ is the sample covariance matrix. The trace term follows from the cyclic property $\text{tr}(AB) = \text{tr}(BA)$:
\[
\sum_{k=1}^{n}\mathbf{x}_k^{\top}\bm{\Theta}\mathbf{x}_k = \text{tr}\left( \sum_{k=1}^{n}\mathbf{x}_k^{\top}\bm{\Theta}\mathbf{x}_k\right) = \text{tr}\left(\sum_{k=1}^{n}\mathbf{x}_k\mathbf{x}_k^{\top}\bm{\Theta}\right) = n\text{tr}\left(S\bm{\Theta}\right).
\]
The score function and the Hessian are
\[
U(\bm{\Theta}) = \frac{\partial\log{f}(\mathbf{X};\bm{\Theta})}{\partial\bm{\Theta}}  = \frac{n}{2}\bm{\Theta}^{-1}-\frac{n}{2}S\]
\[\mathbf{H}(\bm{\Theta}) = \frac{\partial^{2}\log{f}(\mathbf{X};\bm{\Theta})}{\partial\bm{\Theta}^{2}} = - \frac{n}{2}\left(\bm{\Theta}^{-1}\otimes\bm{\Theta}^{-1}\right)
\]
where $\otimes$ denotes the Kronecker product.

The observed Fisher information for a sample of size $n$ is the negative Hessian
\[
I(\mathbf{X};\bm{\Theta}) = -\mathbf{H}(\bm{\Theta}) = \frac{n}{2}\left(\bm{\Theta}^{-1}\otimes\bm{\Theta}^{-1}\right)
\]
In GGMs, the observed and expected Fisher information coincide since $I(\mathbf{X};\bm{\Theta})$ does not depend on the data. The expected information for a single observation is therefore
\[
\mathcal{I}(\mathbf{x}_1;\bm{\Theta}) = \mathbb{E}_{\mathbf{x}_1}\left[I(\mathbf{x}_1;\bm{\Theta})\right] = \frac{1}{2}\left(\bm{\Theta}^{-1}\otimes\bm{\Theta}^{-1}\right)
\]
and for a sample of size $n$: $\mathcal{I}(\mathbf{X};\bm{\Theta}) = n \ \cdot\ \mathcal{I}(\mathbf{x}_1;\bm{\Theta})$.

The Kronecker product generates $p^2 \times p^2$ matrix, but $\bm{\Theta}$ has only $p(p+1)/2$ unique entries due to symmetry. Appendix \ref{appendix:matrix_operators} describes the three matrix operators used to transform between the full $p^2 \times p^2$ and the reduced $\frac{p(p+1)}{2} \times \frac{p(p+1)}{2}$, including the commutation matrix.

\section[Prior distributions for the precision matrix]{Prior distributions for the precision matrix $\bm{\Theta}$}
\label{appendix:priors}

This appendix provides technical details on the Wishart and \textit{G}-Wishart prior distributions for the precision matrix $\bm{\Theta}$. For each  prior we derive the density, the posterior distribution given the GGM likelihood, and the observed Fisher information of both the prior and the posterior. These quantities are used in Sections \ref{subsec:ess_methods} and \ref{sec:sample_size_planning} to define the ESS methods and sample size planning procedures.

\subsection{The Wishart Prior}
\label{subsection:wishart_prior}

The Wishart distribution $\bm{\Theta} \sim W(\nu,\bm{\Psi})$ has density

\[\pi(\bm{\Theta}) = C(\nu,\bm{\Psi})^{-1}|\bm{\Theta}|^{(\nu-p-1)/2}\exp{\left( - \frac{1}{2}\text{tr}(\bm{\Psi}^{-1}\bm{\Theta})\right)}\]

where $\nu > p - 1$ is the degrees of freedom, $\bm{\Psi}$ is the $p \times p$ symmetric positive definite matrix encoding the prior beliefs on $\bm{\Theta}$, and $C(\nu,\bm{\Psi})$ is the normalizing constant
\[
C(\nu,\bm{\Psi}) = 2^{\nu p/2}|\bm{\Psi}|^{\nu/2}\pi^{p(p-1)/4}\prod_{i=1}^{p}\Gamma\left(\frac{\nu-i+1}{2}\right)
\]
obtained by integrating over the cone $\mathcal{P}^{+} = \left\lbrace \bm{\Theta} : \lambda_i\left(\bm{\Theta}\right) > 0, i = 1,\ldots,p \right\rbrace$ of $p \times p$ symmetric positive definite matrices, where $\lambda_i$ indicates the $i\text{-th}$ eigenvalue of $\bm{\Theta}$.

\paragraph{Posterior distribution.} Combined with the GGM likelihood $f(\mathbf{X};\bm{\Theta})$, the posterior remains a Wishart distribution, $\bm{\Theta}\mid\mathbf{X}\sim W\left(n + \nu,\hat{\bm{\Psi}}\right)$, where $\hat{\bm{\Psi}} = \left(\bm{\Psi}^{-1} + nS\right)^{-1}$,
\[
\pi(\bm{\Theta}\mid\mathbf{X}) = C\left(n+\nu,\left(\bm{\Psi}^{-1} + nS\right)^{-1}\right)^{-1}|\bm{\Theta}|^{(n+\nu-p-1)/2}\exp{\left( - \frac{1}{2}\text{tr}\left[\left(\bm{\Psi}^{-1} + nS\right)\bm{\Theta}
\right]\right)}.
\]
\paragraph{Observed Fisher information.} The observed Fisher information of the prior is the negative second derivative of the log-density,
\[
I(\bm{\Theta}) = -\frac{\partial^2\log{\pi(\bm{\Theta})}}{\partial\bm{\Theta}^{2}} = \frac{\nu-p-1}{2} \left(\bm{\Theta}^{-1} \otimes \bm{\Theta}^{-1}\right)
\]
and for the posterior
\[
I(\bm{\Theta}\mid\mathbf{X}) = -\frac{\partial^2\log{\pi(\bm{\Theta}\mid\mathbf{X})}}{\partial\bm{\Theta}^{2}}  = \frac{n+\nu-p-1}{2} \left(\bm{\Theta}^{-1} \otimes \bm{\Theta}^{-1}\right)
\]
Both are $p^2 \times p^2$ matrices. Using the duplication matrix $\mathbf{D}$, they reduce to $p(p+1)/2 \times p(p+1)/2$, the half vectorization form. See Appendix \ref{appendix:matrix_operators} for details. Since the observed information of the Wishart does not depend on the data $\mathbf{X}$, the observed and expected Fisher information coincide.

\paragraph{Mean and variance.} The prior mean and variance are
\begin{align*}
\mathbb{E}[\bm{\Theta}] &= \nu\bm{\Psi} \quad{}\quad{} \\
\mathbb{V}(\bm{\Theta}) &= \nu\left(\mathbf{I} + \mathbf{K}\right)\left(\bm{\Psi}\otimes\bm{\Psi}\right)
\end{align*}
and the posterior mean and variance are
\begin{align*}
\mathbb{E}[\bm{\Theta}\mid\mathbf{X}] &= (n+\nu)\hat{\bm{\Psi}} \\
\quad{}\quad{} 
\mathbb{V}(\bm{\Theta}\mid\mathbf{X}) &= (n+\nu)\left(\mathbf{I} + \mathbf{K}\right)\left(\hat{\bm{\Psi}}\otimes\hat{\bm{\Psi}}\right).
\end{align*} 
The variance matrices are $p^2 \times p^2$. Using the elimination matrix $\mathbf{E}$, they reduce to $p(p+1)/2 \times p(p+1)/2$, by removing redundant entries due to the symmetry of $\bm{\Theta}$. Here $\mathbf{I}$ denotes the $p^2 \times p^2$ identity matrix and $\mathbf{K}$ the commutation matrix (see Appendix \ref{appendix:matrix_operators} for their definitions).

\paragraph{Prior elicitation.} When prior information is available from a study of size $\nu$, the hyperparameters $(\nu, \bm{\Psi}^{*})$ are elicited by setting $\bm{\Psi}^{*} = \bm{\Psi} / \nu$, where $\bm{\Psi} = S^{-1}$ is the sample precision matrix from the prior study. Under this specification, $\mathbb{E}[\bm{\Theta}] = \nu \bm{\Psi}^{*} = \bm{\Psi}$, so that $\nu$ can be directly interpreted as a prior sample size, controlling the concentration of the prior around $\bm{\Psi}$.

\subsection{The \textit{G}-Wishart Prior}
\label{subsection:gwishart_prior}

The \textit{G}-Wishart distribution extends the Wishart to precision matrices constrained to the graph structure $G$. The support is restricted to $\mathcal{P}^{+}_G$, the cone of $p \times p$ symmetric positive definite matrices whose zero entries correspond to absent edges of $G$. Throughout this work, we parametrize the \textit{G}-Wishart following the reparametrization $\nu = \delta + p -1$ of \citet{Roverato2000}, where $\delta > 0$ is the shape parameter in the original formulation, so that $\nu > p-1$. Under this convention, the \textit{G}-Wishart and Wishart share the same interpretation of $\nu$, and the exponent $(\delta-2)/2$ in the original density becomes $(\nu - p - 1) / 2$.

The \textit{G}-Wishart prior $\bm{\Theta}\sim W_G(\nu,\bm{\Psi})$ has density
\[
\pi(\bm{\Theta}) = C_G(\nu,\bm{\Psi})^{-1}|\bm{\Theta}|^{(\nu-p-1)/2}\exp{\left( -\frac{1}{2}\text{tr}\left(\bm{\Psi}^{-1}\bm{\Theta}\right)\right)}
\] 
where $\nu$ is the degrees of freedom, $\bm{\Psi}$ is a $p \times p$ symmetric positive definite matrix encoding prior beliefs about $\bm{\Theta}$, and $C_G(\nu,\bm{\Psi})$ is the normalizing constant,
\[
C_G(\nu,\bm{\Psi}) = \int_{\mathcal{P}^{+}_G}|K|^{(\nu-p-1)/2}\exp{\left( -\frac{1}{2}\text{tr}\left(\bm{\Psi}^{-1}K\right)\right)}\text{d}K
\]
Unlike the normalizing constant of a Wishart, $C_G(\nu,\bm{\Psi})$ does not have a closed form in general, it depends on the graph structure $G$ and must be approximated numerically.

\paragraph{Posterior distribution.} Combined with the GGM likelihood $f(\mathbf{X};\bm{\Theta})$, the posterior remains \textit{G}-Wishart $\bm{\Theta}\mid\mathbf{X} \sim W_G(n+\nu,\hat{\bm{\Psi}})$, where $\hat{\bm{\Psi}} = (\bm{\Psi}^{-1}+nS)^{-1}$,
\[
\pi(\bm{\Theta}\mid\mathbf{X}) = C_G\left(n+\nu,(\bm{\Psi}^{-1}+nS)^{-1}\right)^{-1}|\bm{\Theta}|^{(n+\nu-p - 1)/2}\exp{\left( -\frac{1}{2}\text{tr}\left((\bm{\Psi}^{-1}+nS)\bm{\Theta}\right)\right)}.
\] 
\paragraph{Observed Fisher information.} We use $\bm{\Theta}_G$ to denote a precision matrix following a \textit{G}-Wishart distribution, constraining $\bm{\Theta}_{G}$ to $\mathcal{P}^{+}_G$. The observed Fisher information of the prior is
\[
I(\bm{\Theta}_G) = -\frac{\partial^2\log{\pi(\bm{\Theta}_G)}}{\partial\bm{\Theta}_G^{2}} = \frac{\nu-p -1}{2} \left(\bm{\Theta}_G^{-1} \otimes \bm{\Theta}_G^{-1}\right)
\]
and for the posterior,
\[
I(\bm{\Theta}_G\mid \mathbf{X}) = -\frac{\partial^2\log{\pi(\bm{\Theta}_G|\mathbf{X})}}{\partial\bm{\Theta}_G^{2}} = \frac{n+\nu-p - 1}{2} \left(\bm{\Theta}_G^{-1} \otimes \bm{\Theta}_G^{-1}\right)
\]
Both are $p^2 \times p^2$ matrices. Using the duplication matrix $\mathbf{D}$, they reduce to $p(p+1)/2 \times p(p+1)/2$ (see Appendix \ref{appendix:matrix_operators} for details). Since the observed information of the \textit{G}-Wishart does not depend on the data $\mathbf{X}$, the observed and expected Fisher information coincide.

Under the reparametrization $\nu = \delta + p - 1$, the observed Fisher information of the \textit{G}-Wishart prior takes the same form as that of the Wishart prior (Appendix \ref{subsection:wishart_prior}), with exponent $(\nu - p - 1)/2$ in both cases.

\paragraph{Mean and variance.} The mean and variance of the \textit{G}-Wishart do not have closed form expressions in general \citep{Roverato2002}. In this work, they are estimated via Monte Carlo simulation by drawing samples from $W_G(\nu,\bm{\Psi})$ using the algorithm of \citet[Section~2.4]{Lenkoski2013}, and computing the sample mean and variance of $\bm{\Theta}$.

\paragraph{Prior elicitation.} When prior information is available from a study of size $\nu$, the hyperparameters $(\nu, \bm{\Psi}_{G}^{*})$ are elicited by setting $\bm{\Psi}_{G}^{*} = \bm{\Psi}_G / \nu$, where $\bm{\Psi}_{G} \in \mathcal{P}^{+}_G$ is the estimate (maximum likelihood or posterior mode) of the precision matrix under $G$ obtained from the prior study. Under this specification the prior mean satisfies $\mathbb{E}[\bm{\Theta}] \approx \bm{\Psi}_G$, with the approximation arising from the intractability of the normalizing constant for non-decomposable graphs. The centering bias vanishes as $\nu \to \infty$ and as $G$ approaches the complete graph. When $G$ is complete, $\bm{\Psi}_G$ reduces to the sample precision matrix $\bm{\Psi} = S^{-1}$ and the elicitation recovers the Wishart case exactly (Appendix \ref{subsection:wishart_prior}).

\section{Matrix operators}
\label{appendix:matrix_operators}

The precision matrix $\bm{\Theta}$ is a $p \times p$ symmetric matrix with $p^2$ entries, of which only $d = p(p+1)/2$ are unique. Many quantities in this work, including Fisher information matrices and variance matrices, are naturally expressed in the full $p^2 \times p^2$ space but can be reduced to the $d \times d$ space by exploiting this symmetry. We describe here the three matrix operators used for this reduction, and the commutation matrix used to reorganize elements within the full vectorization. See \citet{Magnus1979} for further properties.

The \textbf{vec} operator stacks the columns of a $p\times p$ matrix $A$ into a $p^2 \times 1$ vector, while the \textbf{vech} operator stacks only the lower triangular entries into a $d \times 1$ vector.

\subsection{The Duplication Matrix $\mathbf{D}$}
\label{appendix:duplication_matrix}

The duplication matrix $\mathbf{D}$ of dimension $p^2 \times d$ maps the half vectorization of a symmetric matrix to its full vectorization,
\[
\mathbf{D}\text{vech}(A) = \text{vec}(A)
\]
In this paper, $\mathbf{D}$ is used to reduce information matrices from  $p^2 \times p^2$ to $d \times d$. Specifically, for an information matrix $\mathcal{I}$ of dimension $p^2 \times p^2$, the reduced form is 
\[
\mathcal{I}_{d\times d} = \mathbf{D}^{\top}\mathcal{I}\mathbf{D}.
\]

\paragraph{Example: Fisher Information for a $2\times2$ Precision Matrix.}  

The following example illustrates why the duplication matrix $\mathbf{D}$ is the 
correct operator for dimension reduction of Fisher information matrices. We consider a $2\times 2$ precision matrix and show that the analytical Fisher information computed via direct differentiation coincides with the result obtained using $\mathbf{D}$.

Given a $2 \times 2$ precision matrix
\[
\bm{\Theta} = 
\begin{pmatrix} 
\theta_{11} & \theta_{12} \\ 
\theta_{12} & \theta_{22} 
\end{pmatrix}
\]
with $d = p(p+1)/2=3$ unique entries $(\theta_{11},\theta_{12},\theta_{22})$, the expected Fisher information $\mathcal{I}(\mathbf{X};\bm{\Theta}) = \frac{n}{2}(\bm{\Theta}^{-1}\otimes\bm{\Theta}^{-1})$ 
is a $4\times 4$ matrix in the full vectorization. Using $\mathbf{D}$ of dimension $4\times 3$, the reduced $3\times 3$ form is
\[
\mathcal{I}_{d\times d} = \mathbf{D}^{\top}\left(\frac{1}{2}(\bm{\Theta}^{-1}\otimes\bm{\Theta}^{-1})\right)\mathbf{D}
\]
To verify, consider the specific matrix
\[
\bm{\Theta} = \begin{pmatrix} 2 & 0.5 \\ 0.5 & 1 \end{pmatrix}
\]
The analytical Fisher information, computed as the negative second derivative 
of the log-likelihood with respect to $(\theta_{11}, \theta_{12}, \theta_{22})$, is
\[
-\mathbf{H}(\bm{\Theta}) = \frac{1}{2\Delta^2}
\begin{pmatrix} 
\theta_{22}^2 & -2\theta_{12}\theta_{22} & \theta_{12}^2 \\ 
-2\theta_{12}\theta_{22} & 2\theta_{11}\theta_{22} + 2\theta_{12}^2 & -2\theta_{11}\theta_{12} \\ 
\theta_{12}^2 & -2\theta_{11}\theta_{12} & \theta_{11}^2 
\end{pmatrix}
\]
where $\Delta = \theta_{11}\theta_{22} - \theta_{12}^2$ is the determinant of $\bm{\Theta}$. For the specific $\bm{\Theta}$ above, this gives (rounded to three decimal values)
\[
-\mathbf{H}(\bm{\Theta}) =
\begin{pmatrix} 
0.163 & -0.163 & 0.041 \\ 
-0.163 & 0.735 & -0.327 \\ 
0.041 & -0.327 & 0.653
\end{pmatrix}
\]
The duplication matrix for $p=2$ is
\[
\mathbf{D} = 
\begin{pmatrix} 
1 & 0 & 0 \\ 
0 & 1 & 0 \\ 
0 & 1 & 0 \\
0 & 0 & 1 
\end{pmatrix}
\]
and the operation
\[
\overbrace{\begin{pmatrix} 
1 & 0 & 0 & 0 \\ 
0 & 1 & 1 & 0 \\ 
0 & 0 & 0 & 1 
\end{pmatrix}}^{\mathbf{D}^{\top}}\left(\frac{1}{2}(\bm{\Theta}^{-1}
\otimes\bm{\Theta}^{-1})\right)\overbrace{\begin{pmatrix} 
1 & 0 & 0 \\ 
0 & 1 & 0 \\ 
0 & 1 & 0 \\
0 & 0 & 1 
\end{pmatrix}}^{\mathbf{D}}
\]
gives
\[
\mathbf{D}^{\top}\left(\frac{1}{2}(\bm{\Theta}^{-1}\otimes\bm{\Theta}^{-1})\right)\mathbf{D} =\begin{pmatrix} 
0.163 & -0.163 & 0.041 \\ 
-0.163 & 0.735 & -0.327 \\ 
0.041 & -0.327 & 0.653
\end{pmatrix}
\]
The two matrices are identical, confirming that $\mathbf{D}$ provides the correct dimension reduction for Fisher information matrices.

The following R code reproduces this computation:
\begin{lstlisting}[style=rstyle]
# Define the precision matrix
Theta <- matrix(c(2, 0.5, 0.5, 1), nrow = 2)
Theta_inv <- chol2inv(chol(Theta))
# Duplication matrix for p = 2
D <- matrix(c(1,0,0,0,0,1,1,0,0,0,0,1), nrow = 4, byrow = FALSE)
# Fisher information via duplication matrix
I <- t(D) %*% (0.5 * kronecker(Theta_inv, Theta_inv)) %*% D
round(I, 3) 
\end{lstlisting}

For the inverse Fisher information, the Moore-Penrose pseudoinverse $\mathbf{D}^{+}$ must be used instead (see Section \ref{appendix:inverse_duplication_matrix}.

\subsection{The Inverse Duplication Matrix $\mathbf{D}^{+}$}
\label{appendix:inverse_duplication_matrix}

The Moore-Penrose pseudoinverse of the duplication matrix, $\mathbf{D}^{+} = (\mathbf{D}^{\top}\mathbf{D})^{-1}\mathbf{D}^{\top}$, is of dimension $d \times p^2$ and satisfies $\mathbf{D}^{+}\mathbf{D} = \mathbf{I}$. In this paper, $\mathbf{D}^{+}$ is used to reduce inverse information matrices from $p^2 \times p^2$ to $d \times d$. For an inverse information matrix $\mathcal{I}^{-1}$ of dimension $p^2 \times p^2$, the reduced form is
\[
\mathcal{I}^{-1}_{d \times d} = \mathbf{D}^{+}\mathcal{I}^{-1}{\mathbf{D}^{+}}^{\top}
\]

Continuing the example from Section \ref{appendix:duplication_matrix}, the inverse Fisher information is computed as follows
\[
\mathcal{I}^{-1}_{d\times d} = \mathbf{D}^{+}\left(2(\bm{\Theta}
\otimes\bm{\Theta})\right){\mathbf{D}^{+}}^{\top}
\]
The R code below shows that the result coincides with the direct matrix inverse of $\mathcal{I}_{d \times d}$:
\begin{lstlisting}[style=rstyle]
# Moore-Penrose pseudoinverse of the duplication matrix for p = 2
D_plus <- solve(t(D) %*% D) %*% t(D)
# Inverse Fisher information via D_plus
I_inv <- D_plus %*% (2 * kronecker(Theta, Theta)) %*% t(D_plus)
round(I_inv, 3) 
# Compare to the analytical inverse
round(solve(I),3)
\end{lstlisting}

\subsection{The Elimination Matrix $E$}

The elimination matrix $\mathbf{E}$ of dimensions $d \times p^2$ maps the full vectorization of a symmetric matrix to its half vectorization,
\[
\mathbf{E}\text{vec}(A) = \text{vech}(A),
\]
and it satisfies $\mathbf{E}\mathbf{D}=\mathbf{I}_d$, where $\mathbf{I}_d$ is the $d \times d$ matrix. In this paper, $\mathbf{E}$ is used to reduce the variance matrices from $p^2 \times p^2$ to $d \times d$. For a variance matrix $\mathbb{V}$ of dimension $p^2 \times p^2$, the reduced form is
\[
\mathbb{V}_{d\times d} = \mathbf{E} \mathbb{V} \mathbf{E}^{\top}.
\]

\subsection{The Commutation Matrix $\mathbf{K}$}

The commutation matrix $\mathbf{K}$ of dimensions $p^2 \times p^2$ permutes the elements of $\text{vec}(A)$ to give $\text{vec}(A^{\top})$,
\[
\mathbf{K}\text{vec}(A) = \text{vec}(A^{\top}).
\]
It does not operate any dimension reduction, but it rearranges elements within the full $p^2 \times p^2$ space. In this paper, $\mathbf{K}$ appears in the computation of variance and information matrices, exploiting the symmetry of $\bm{\Theta}$ through the identity $\mathbf{K}(\bm{\Theta}\otimes\bm{\Theta}) = (\bm{\Theta}\otimes\bm{\Theta})\mathbf{K}$.

\section{Prior ESS: Derivations of Analytical Formulas}
\label{appendix:prior_ess_formulas}

This appendix provides the analytical formulas for each prior ESS method under the elicitation convention $\bm{\Psi}^{*} = \bm{\Psi}/\nu$, so that $\mathbb{E}\left[\bm{\Theta}\right]=\bm{\Psi}$ exactly under the Wishart prior $\bm{\Theta}\sim W(\nu,\bm{\Psi}/\nu)$, and approximately under the \textit{G}-Wishart prior $\bm{\Theta}\sim W_{G}(\nu,\bm{\Psi}/\nu)$. Under this convention, $\bm{\Psi}$ is directly interpretable as the prior belief about the precision matrix. The determinant ratio is used throughout to obtain a global measure of the ESS. The Morita-Thall-Müller \citep{Morita2008,Morita2010}, Pennello-Thompson and ELIR \citep{Neuenschwander2020} methods have closed-form expressions under both priors. For the Variance Ratio and the Precision Ratio \citep{Neuenschwander2020} specific quantities must be estimated via Monte Carlo simulation under the \textit{G}-Wishart prior, as discussed at the end of this appendix.

\paragraph{Variance Ratio (VR).} The Variance Ratio defines the prior ESS as the ratio of the expected single-observation information to the prior variance,
\begin{align*}
    ESS_{\text{VR}} &= \left(\frac{\lvert\mathbb{E}_{\bm{\Theta}}\left[\mathcal{I}^{-1}(\mathbf{x}_1;\bm{\Theta})\right]\lvert}{\lvert\mathbb{V}\left(\bm{\Theta}\right)\lvert}\right)^{1/d} \\
    &= \left(\frac{\lvert 2\mathbf{D}^{+}\mathbb{E}_{\bm{\Theta}}[\bm{\Theta}\otimes\bm{\Theta}]{\mathbf{D}^{+}}^{\top} \rvert}{\lvert\nu\mathbf{E}\left(\mathbf{I} + \mathbf{K}\right)\left(\bm{\Psi}^{*}\otimes\bm{\Psi}^{*}\right)\mathbf{E}^{\top}\rvert}\right)^{1/d}
\end{align*}
where $\mathbb{E}_{\bm{\Theta}}[\bm{\Theta}\otimes\bm{\Theta}] = \nu^{2} \ (\bm{\Psi}^{*} \otimes \bm{\Psi}^{*}) + \nu \ \text{vec}(\bm{\Psi}^{*})\text{vec}(\bm{\Psi}^{*})^{\top} + \nu \ \mathbf{K} (\bm{\Psi}^{*}\otimes\bm{\Psi}^{*})$ \citep[Section 5, Eq. 9]{Hagedorn2022}.

Substituting $\bm{\Psi}^{*}=\bm{\Psi}/\nu$:
\[
\mathbb{E}_{\bm{\Theta}}[\bm{\Theta}\otimes\bm{\Theta}] = (\bm{\Psi} \otimes \bm{\Psi}) + (1/\nu)\text{vec}(\bm{\Psi})\text{vec}(\bm{\Psi})^{\top} + (1/\nu)\mathbf{K} (\bm{\Psi} \otimes\bm{\Psi})
\]
and the denominator becomes $(1/\nu)\mathbf{E}\left(\mathbf{I} + \mathbf{K}\right)\left(\bm{\Psi}\otimes\bm{\Psi}\right)\mathbf{E}^{\top}$, giving:
\[
    ESS_{\text{VR}} = \left(\frac{\lvert 2\mathbf{D}^{+}\left[(\bm{\Psi} \otimes \bm{\Psi}) + (1/\nu)\text{vec}(\bm{\Psi})\text{vec}(\bm{\Psi})^{\top} + (1/\nu)\mathbf{K} (\bm{\Psi} \otimes\bm{\Psi})\right]{\mathbf{D}^{+}}^{\top} \rvert}{\lvert  (1/\nu)\mathbf{E}\left(\mathbf{I} + \mathbf{K}\right)\left(\bm{\Psi}\otimes\bm{\Psi}\right)\mathbf{E}^{\top} \rvert}\right)^{1/d} \\
\]
To simplify the leading term in the numerator, we use the identity $\mathbf{D}^{+} = (1/2)\mathbf{E}(\mathbf{I} + \mathbf{K})$ \citep[][Lemma 3.6 (iv)]{Magnus1979}, where $\mathbf{E}$ corresponds to their elimination matrix $\mathbf{L}$ and $\mathbf{N} = (1/2)(\mathbf{I} + \mathbf{K})$, and consider Lemma 2.1 (ii) and (iii) in \citet{Magnus1979}. Therefore we have:
\begin{align*}
& 2\mathbf{D}^{+}(\bm{\Psi} \otimes \bm{\Psi}){\mathbf{D}^{+}}^{\top} \\
&= 2 (1/4)\mathbf{E}(\mathbf{I} + \mathbf{K})(\bm{\Psi} \otimes \bm{\Psi})(\mathbf{I} + \mathbf{K})^{\top} \mathbf{E}^{\top} \quad{}\text{(substituting } \mathbf{D}^{+}\text{)}\\
&= 2 (1/4)\mathbf{E}(\mathbf{I} + \mathbf{K})(\bm{\Psi} \otimes \bm{\Psi})(\mathbf{I} + \mathbf{K}) \mathbf{E}^{\top} \quad{}\text{(symmetry of } \mathbf{I}+\mathbf{K}\text{)}\\
&= 2 (1/4)\mathbf{E}(\mathbf{I} + \mathbf{K})^{2}(\bm{\Psi} \otimes \bm{\Psi})\mathbf{E}^{\top}  \quad{}\text{(} \mathbf{K}(\bm{\Psi} \otimes \bm{\Psi}) = (\bm{\Psi} \otimes \bm{\Psi})\mathbf{K} \text{ for symmetric } \mathbf{\Psi} \text{, Lemma 2.1 (ii))}\\
&= \mathbf{E}(\mathbf{I} + \mathbf{K})(\bm{\Psi} \otimes \bm{\Psi})\mathbf{E}^{\top}\quad{}\text{(} (\mathbf{I} + \mathbf{K})^{2} = 2(\mathbf{I} + \mathbf{K}) \text{, Lemma 2.1 (iii))}
\end{align*}
Substituting and separating the leading term $\mathbf{A} = \mathbf{E}(\mathbf{I} + \mathbf{K})(\bm{\Psi} \otimes \bm{\Psi})\mathbf{E}^{\top}$ from the remainder $\mathbf{B} = 2\mathbf{D}^{+}\left[\text{vec}(\bm{\Psi})\text{vec}(\bm{\Psi})^{\top} + \mathbf{K} (\bm{\Psi} \otimes\bm{\Psi})\right]{\mathbf{D}^{+}}^{\top} $:
\[
ESS_{\text{VR}} = \left( \frac{\lvert \mathbf{A} + \mathbf{B}/\nu\rvert}{(1/\nu)^{d}\lvert\mathbf{A}\rvert} \right)^{1/d} = \nu \cdot \left( \frac{\lvert \mathbf{A} + \mathbf{B}/\nu\rvert}{\lvert\mathbf{A}\rvert} \right)^{1/d}
\]
As $\nu \rightarrow \infty$, by continuity of the determinant $\left(\frac{\lvert A + B/\nu \rvert}{\lvert A\rvert} \right)^{1/d}\rightarrow 1$ and $ESS_{\text{VR}}\sim \nu$ growing linearly with the prior sample size $\nu$.

\paragraph{Precision Ratio (PR).} The Precision Ratio compares prior precision directly to expected information
\begin{align*}
  ESS_{\text{PR}} &= \left(\frac{\lvert\mathbb{V}^{-1}\left(\bm{\Theta}\right)\rvert}{\lvert\mathbb{E}_{\bm{\Theta}}\left[\mathcal{I}(\mathbf{x}_1;\bm{\Theta})\right]\rvert}\right)^{1/d} \\
  &= \left(\frac{\lvert\mathbb{V}\left(\bm{\Theta}\right)\rvert^{-1}}{\lvert\mathbb{E}_{\bm{\Theta}}\left[\mathcal{I}(\mathbf{x}_1;\bm{\Theta})\right]\rvert}\right)^{1/d} \\
  &= \left(\frac{\lvert\nu\mathbf{E}\left(\mathbf{I} + \mathbf{K}\right)\left(\bm{\Psi}^{*}\otimes\bm{\Psi}^{*}\right)\mathbf{E}^{\top}\rvert^{-1}}{\lvert \frac{1}{2}\mathbf{D}^{\top}\mathbb{E}_{\bm{\Theta}}[\bm{\Theta}^{-1}\otimes\bm{\Theta}^{-1}]\mathbf{D} \rvert}\right)^{1/d}
\end{align*}
where $\mathbb{E}_{\bm{\Theta}}[\bm{\Theta}^{-1}\otimes\bm{\Theta}^{-1}] = c_1 \ ({\bm{\Psi}^{*}}^{-1}  \otimes  {\bm{\Psi}^{*}}^{-1}) + c_2 \ \text{vec}({\bm{\Psi}^{*}}^{-1})\text{vec}({\bm{\Psi}^{*}}^{-1})^{\top} + c_2 \ \mathbf{K} ({\bm{\Psi}^{*}}^{-1} \otimes  {\bm{\Psi}^{*}}^{-1})$, with $c_2 = \left[(\nu-p)(\nu-p-1)(\nu-p-3)\right]^{-1}$ and $c_1=c_2(\nu-p-2)$ \citep{VonRosen1997}.

Substituting ${\bm{\Psi}^{*}}^{-1} = \nu\bm{\Psi}^{-1}$:
\[
\mathbb{E}_{\bm{\Theta}}[\bm{\Theta}^{-1}\otimes\bm{\Theta}^{-1}] = c_1 \nu^2 \ (\bm{\Psi}^{-1}  \otimes  \bm{\Psi}^{-1}) + c_2 \nu^2 \ \text{vec}(\bm{\Psi}^{-1})\text{vec}(\bm{\Psi}^{-1})^{\top} + c_2 \nu^2 \ \mathbf{K} (\bm{\Psi}^{-1} \otimes  \bm{\Psi}^{-1})
\]
and the numerator becomes $(1/\nu)\mathbf{E}\left(\mathbf{I} + \mathbf{K}\right)\left(\bm{\Psi}\otimes\bm{\Psi}\right)\mathbf{E}^{\top}$, so:
\[
ESS_{\text{PR}} = \left( \frac{\lvert (1/\nu)\mathbf{E}\left(\mathbf{I} + \mathbf{K}\right)\left(\bm{\Psi}\otimes\bm{\Psi}\right)\mathbf{E}^{\top} \rvert ^{-1}}{ \lvert \frac{1}{2}\mathbf{D}^{\top}\left[c_1 \nu^2 \ (\bm{\Psi}^{-1}  \otimes  \bm{\Psi}^{-1}) + c_2 \nu^2 \ \text{vec}(\bm{\Psi}^{-1})\text{vec}(\bm{\Psi}^{-1})^{\top} + c_2 \nu^2 \ \mathbf{K} (\bm{\Psi}^{-1} \otimes  {\bm{\Psi}^{*}}^{-1})\right]\mathbf{D} \rvert} \right)^{1/d}
\]
Let 
\begin{align*}
\mathbf{A} &= \mathbf{E}(\mathbf{I} + \mathbf{K})(\bm{\Psi} \otimes \bm{\Psi})\mathbf{E}^{\top} \\
\mathbf{B} &= \mathbf{D}^{\top}(\bm{\Psi}^{-1}  \otimes  \bm{\Psi}^{-1})\mathbf{D} \\
\mathbf{C} &= \mathbf{D}^{\top}\left(\ \text{vec}(\bm{\Psi}^{-1})\text{vec}(\bm{\Psi}^{-1})^{\top} +  \ \mathbf{K} (\bm{\Psi}^{-1} \otimes  \bm{\Psi}^{-1})\right)\mathbf{D}.
\end{align*}
Then:
\begin{align*}
ESS_{\text{PR}} &= \left( \frac{\lvert(1/\nu)\mathbf{A}\rvert^{-1}}{\lvert \frac{1}{2}c_2 \nu^{2} \left[(\nu - p - 2) \mathbf{B} + \mathbf{C} \right]\rvert} \right)^{1/d} \\
&= \left( \frac{\nu^{d}\lvert\mathbf{A}\rvert^{-1}}{ \left[\frac{1}{2}c_2 (\nu - p - 2)\nu^{2}\right]^{d}\lvert   \mathbf{B} + \mathbf{C}/(\nu - p - 2) \rvert} \right)^{1/d} \\
&= \frac{2(\nu - p)(\nu - p - 1)(\nu - p - 3)}{\nu (\nu - p - 2)}\cdot \left( \frac{\lvert\mathbf{A}\rvert^{-1}}{ \lvert   \mathbf{B} + \mathbf{C}/(\nu - p - 2) \rvert} \right)^{1/d} \\
&= \frac{(\nu - p)(\nu - p - 1)(\nu - p - 3)}{\nu (\nu - p - 2)}\cdot \left( \frac{\lvert\mathbf{B}\rvert}{ \lvert   \mathbf{B} + \mathbf{C}/(\nu - p - 2) \rvert} \right)^{1/d} 
\end{align*}
where $\mathbf{A} = \mathbf{E}(\mathbf{I} + \mathbf{K})(\bm{\Psi} \otimes \bm{\Psi})\mathbf{E}^{\top} = 2\mathbf{D}^{+}(\bm{\Psi} \otimes \bm{\Psi}){\mathbf{D}^{+}}^{\top}$ from the explanation of the VR method, and $\mathbf{A}^{-1} = \left(2\mathbf{D}^{+}(\bm{\Psi} \otimes \bm{\Psi}){\mathbf{D}^{+}}^{\top}\right)^{-1} = \frac{1}{2}\mathbf{D}^{\top}(\bm{\Psi}^{-1}  \otimes  \bm{\Psi}^{-1})\mathbf{D}= \frac{1}{2}\mathbf{B}$.

As $\nu \rightarrow \infty$, by continuity of the determinant $\left( \frac{\lvert\mathbf{B}\rvert}{ \lvert   \mathbf{B} + \mathbf{C}/(\nu - p - 2) \rvert} \right)^{1/d} \rightarrow 1$ and $ESS_{\text{PR}}$ is approximated by the factor $\frac{(\nu - p)(\nu - p - 1)(\nu - p - 3)}{\nu (\nu - p - 2)}$, which grows linearly with $\nu$ and converges to $\nu$ as $\nu \rightarrow \infty$ with $p$ fixed.

\paragraph{Morita-Thall-Müller (MTM).} The Morita-Thall-Müller method identifies the prior ESS as the sample size $m$ that minimizes the discrepancy between the prior and a locally-matched reference prior at the prior mean $\overline{\bm{\Theta}} = \bm{\Psi}$. Let $\pi_0 \sim W(\nu,\bm{\Theta}_{0})$ be the informative prior and $\pi_{\varepsilon}\sim(\nu_{\text{min}} + \varepsilon,\bm{\Theta}_{\varepsilon})$ the \textit{$\varepsilon\text{-information}$ prior}. Here, $\nu_{\text{min}}=p$ ensures the existence of the second moment, and $\varepsilon>0$ is arbitrarily small. The prior ESS is
\begin{align*}
    ESS_{\text{MTM}} &= \arg\min_{m} \Delta(m,\overline{\bm{\Theta}},\pi,\pi_{\varepsilon}) \\ 
    &=\arg\min_{m} \left\lvert  \text{tr}(I(\nu,\overline{\bm{\Theta}})) -\ \text{tr}(I(m+\nu_{\text{min}} + \varepsilon,\overline{\bm{\Theta}})) +\ \text{tr}(I(\nu_{min},\overline{\bm{\Theta}})) -\ \text{tr}(I(0,\overline{\bm{\Theta}}))   \right\rvert \\
    &= \arg\min_{m} \left\lvert  \left[\frac{\nu-p-1}{2} - \frac{m+\nu_{\text{min}}+\varepsilon-p-1}{2} + \frac{\nu_{\text{min}}-p-1}{2} - \frac{-p-1}{2}\right]g(\overline{\bm{\Theta}})  \right\rvert \\
    &= \arg\min_{m} \left\lvert  \left(\frac{\nu-m-\varepsilon}{2}\right)g(\overline{\bm{\Theta}})  \right\rvert
\end{align*}
where $g(\overline{\bm{\Theta}}) = \text{tr}(\mathbf{D}^{\top}(\overline{\bm{\Theta}}^{-1}\otimes\overline{\bm{\Theta}}^{-1})\mathbf{D})$. Setting $\varepsilon=1$ gives
\[
m^{\star} = \nu - 1
\]
Under the \textit{G}-Wishart, $\nu_{\text{min}}=p$ is the minimum admissible degrees of freedom under the reparametrization $\nu = \delta + p - 1$ of \citep{Roverato2002}. Substituting into the MTM discrepancy, the terms involving $\nu_{\text{min}}$ cancel and the minimization reduces to the same expression as in the Wishart case, giving $ m^{\star} = \nu - 1$.

\paragraph{Pennello-Thompson (PT).} The Pennello-Thompson method compares prior to single-observation Fisher information at the prior mode $\tilde{\bm{\Theta}}$,
\begin{align*}
ESS_{\text{PT}} &= \left(\frac{\lvert I(\tilde{\bm{\Theta}})\rvert}{\lvert\mathcal{I}(\mathbf{x}_1;\tilde{\bm{\Theta}})\rvert}\right)^{1/d} \\
&= \left(\frac{\lvert \frac{\nu-p-1}{2} \mathbf{D}^{\top}(\tilde{\bm{\Theta}}^{-1}\otimes\tilde{\bm{\Theta}}^{-1})\mathbf{D}\rvert}{\lvert\frac{1}{2}\mathbf{D}^{\top}(\tilde{\bm{\Theta}}^{-1}\otimes\tilde{\bm{\Theta}}^{-1})\mathbf{D}\rvert}\right)^{1/d} \\
&= \left(\frac{(\nu-p-1)^{d}\cancel{\lvert \frac{1}{2} \mathbf{D}^{\top}(\tilde{\bm{\Theta}}^{-1}\otimes\tilde{\bm{\Theta}}^{-1})\mathbf{D}\rvert}}{\cancel{\lvert\frac{1}{2}\mathbf{D}^{\top}(\tilde{\bm{\Theta}}^{-1}\otimes\tilde{\bm{\Theta}}^{-1})\mathbf{D}\rvert}}\right)^{1/d} = \nu-p-1
\end{align*}
Under the \textit{G}-Wishart, the observed Fisher information takes the same form as the Wishart with exponent $(\nu - p - 1)/2$ (Appendix \ref{subsection:gwishart_prior}), so the same cancelation applies and $ESS_{\text{PT}} = \nu - p - 1$.

\paragraph{ELIR.} The ELIR method averages the information ratio over the prior distribution,
\begin{align*}
    ESS_{\text{ELIR}} &= \left(\mathbb{E}_{\bm{\Theta}}\left[\frac{\lvert I(\bm{\Theta})\rvert}{\lvert\mathcal{I}(\mathbf{x}_1;\bm{\Theta})\rvert}\right]\right)^{1/d} \\
    &=  \left(\mathbb{E}_{\bm{\Theta}}\left[\frac{\lvert \frac{\nu-p-1}{2} \mathbf{D}^{\top}(\bm{\Theta}^{-1}\otimes\bm{\Theta}^{-1})\mathbf{D}\rvert}{\lvert\frac{1}{2}\mathbf{D}^{\top}(\bm{\Theta}^{-1}\otimes\bm{\Theta}^{-1})\mathbf{D}\rvert}\right]\right)^{1/d} \\
    &= \left((\nu-p-1)^{d \ }\cancel{\mathbb{E}_{\bm{\Theta}}\left[\frac{\lvert \frac{1}{2} \mathbf{D}^{\top}(\bm{\Theta}^{-1}\otimes\bm{\Theta}^{-1})\mathbf{D}\rvert}{\lvert\frac{1}{2}\mathbf{D}^{\top}(\bm{\Theta}^{-1}\otimes\bm{\Theta}^{-1})\mathbf{D}\rvert}\right]}\right)^{1/d} = \nu - p -1
\end{align*}
Under the \textit{G}-Wishart, the same argument as in PT applies and $ESS_{\text{ELIR}} = \nu - p -1$.

Under the \textit{G}-Wishart prior, the variance matrix $\mathbb{V}(\bm{\Theta})$, the expected Fisher information $\mathbb{E}_{\bm{\Theta}}[\mathcal{I}(\mathbf{x}_1;\bm{\Theta})]$, and their posterior counterparts do not have closed-form expressions, as the normalizing constant of the \textit{G}-Wishart is analytically intractable. For the VR and PR methods, these quantities are estimated via Monte Carlo simulation from $W_G(\nu, \bm{\Psi}/\nu)$, constrained to the free parameters of $G$.

\section{Jensen Gap}
\label{appendix:jensen_gap}

The Jensen gap formalizes the difference between $ESS_{\text{VR}}$ and $ESS_{\text{PR}}$ arising from the convexity of matrix inversion on the positive definite cone. Since $\mathcal{I}(\mathbf{x}_1;\bm{\Theta}) = \frac{1}{2}\mathbf{D}^\top(\bm{\Theta}^{-1}\otimes\bm{\Theta}^{-1})\mathbf{D}$ is a nonlinear function of $\bm{\Theta}$, the expectation of its inverse differs from the inverse of its expectation. By Jensen's inequality for the convex function $f(Y)=Y^{-1}$,
\[
\mathbb{E}_{\bm{\Theta}}\left[\mathcal{I}(\mathbf{x}_1;\bm{\Theta})\right]^{-1} \le \mathbb{E}_{\bm{\Theta}}\left[\mathcal{I}^{-1}(\mathbf{x}_1;\bm{\Theta})\right]
\]
where the inequality holds in the Löwner sense \citep{Bhatia1997}, that is,  $\mathbb{E}_{\bm{\Theta}}\left[\mathcal{I}^{-1}(\mathbf{x}_1;\bm{\Theta})\right] - \mathbb{E}_{\bm{\Theta}}\left[\mathcal{I}(\mathbf{x}_1;\bm{\Theta})\right]^{-1}$ is positive semidefinite. Since both matrices are positive definite under $\nu > p + 3$, as expectations of positive definite matrices with finite moments, the Löwner ordering implies the determinant ordering $\left|\mathbb{E}_{\bm{\Theta}}\left[\mathcal{I}(\mathbf{x}_1;\bm{\Theta})\right]^{-1}\right| \le \left|\mathbb{E}_{\bm{\Theta}}\left[\mathcal{I}^{-1}(\mathbf{x}_1;\bm{\Theta})\right]\right|$.
Taking determinants and applying the $1/d$ root, with $d$ denoting the number of free parameters of $\bm{\Theta}$
\[
\left|\mathbb{E}_{\bm{\Theta}}\left[\mathcal{I}(\mathbf{x}_1;\bm{\Theta})\right]\right|^{-1/d} \leq \left|\mathbb{E}_{\bm{\Theta}}\left[\mathcal{I}^{-1}(\mathbf{x}_1;\bm{\Theta})\right]\right|^{1/d},
\]
the Jensen gap is then defined as
\[
J = \left|\mathbb{E}_{\bm{\Theta}}\left[\mathcal{I}^{-1}(\mathbf{x}_1;\bm{\Theta})\right]\right|^{1/d} - \left|\mathbb{E}_{\bm{\Theta}}\left[\mathcal{I}(\mathbf{x}_1;\bm{\Theta})\right]\right|^{-1/d} \ge 0.
\]
 
Considering the definitions of VR and PR,
\[
ESS_{\text{VR}} = \left(\frac{\left|\mathbb{E}_{\bm{\Theta}}\left[\mathcal{I}^{-1}(\mathbf{x}_1;\bm{\Theta})\right]\right|}{\left|\mathbb{V}(\bm{\Theta})\right|}\right)^{1/d}, \qquad
ESS_{\text{PR}} = \left(\frac{\left|\mathbb{V}^{-1}(\bm{\Theta})\right|}{\left|\mathbb{E}_{\bm{\Theta}}\left[\mathcal{I}(\mathbf{x}_1;\bm{\Theta})\right]\right|}\right)^{1/d}
\]
and noting that $|\mathbb{V}^{-1}(\bm{\Theta})| = |\mathbb{V}(\bm{\Theta})|^{-1}  $, their ratio is:
\begin{align*}
 \frac{ESS_{\text{PR}}}{ESS_{\text{VR}}} &= \left(\frac{\left|\mathbb{E}_{\bm{\Theta}}\left[\mathcal{I}(\mathbf{x}_1;\bm{\Theta})\right]\right|}{\left|\mathbb{E}_{\bm{\Theta}}\left[\mathcal{I}^{-1}(\mathbf{x}_1;\bm{\Theta})\right]\right|}\right)^{1/d} \\
 &= \left(\frac{\cancel{|\mathbb{V}^{-1}(\bm{\Theta})|}}{\left|\mathbb{E}_{\bm{\Theta}}\left[\mathcal{I}(\mathbf{x}_1;\bm{\Theta})\right]\right|}\cdot\frac{\cancel{|\mathbb{V}(\bm{\Theta})|}}{\left|\mathbb{E}_{\bm{\Theta}}\left[\mathcal{I}^{-1}(\mathbf{x}_1;\bm{\Theta})\right]\right|}\right)^{1/d} \\
 &= \left(\frac{\left|\mathbb{E}_{\bm{\Theta}}\left[\mathcal{I}(\mathbf{x}_1;\bm{\Theta})\right]\right|^{-1}}{\left|\mathbb{E}_{\bm{\Theta}}\left[\mathcal{I}^{-1}(\mathbf{x}_1;\bm{\Theta})\right]\right|}\right)^{1/d}.
\end{align*}
Rewriting the above inequality as
\[
J=1 - \frac{\left|\mathbb{E}_{\bm{\Theta}}\left[\mathcal{I}(\mathbf{x}_1;\bm{\Theta})\right]\right|^{-1/d}}{\left|\mathbb{E}_{\bm{\Theta}}\left[\mathcal{I}^{-1}(\mathbf{x}_1;\bm{\Theta})\right]\right|^{1/d}} = 1- \frac{ESS_{\text{PR}}}{ESS_{\text{VR}}} \ge0,
\]

confirms that $ESS_{\text{PR}} \leq ESS_{\text{VR}}$ for all finite $\nu$ and $p$. The empirical counterpart, computable from the simulation output is therefore
\[
J = 1- \frac{ESS_{\text{PR}}}{ESS_{\text{VR}}}
\]
This quantity is non-negative, and decreases toward zero as $\nu \rightarrow \infty$ for fixed $p$, and increases with $p$ for fixed $\nu$, reflecting that a larger prior study size reduces the variability of $\mathcal{I}(\mathbf{x}_1;\bm{\Theta})$ under the prior, while larger networks ($p$ increasing) amplify the parameter dependencies in $\bm{\Theta}$, widening the gap.

\section{Predictive Consistency in the Multiparametric GGM Setting}
\label{appendix:pcc_proofs}

\citet{Neuenschwander2020} introduced the predictive consistency criterion as a coherence requirement for a prior ESS measure: a method is predictively consistent when, in expectation, observing $n$ samples raises the posterior ESS by exactly $n$ relative to the prior ESS,
\begin{equation}
\label{eq:pcc}
\mathbb{E}_{\bm{\Theta}} \bigg[\mathbb{E}_{\mathbf{X}_{n}\mid\bm{\Theta}}
\big[ESS_{\bm{\Theta}\mid\mathbf{X}}\big]\bigg] = ESS_{\bm{\Theta}} + n,
\end{equation}
where $ESS_{\bm{\Theta}\mid\mathbf{X}}$ and $ESS_{\bm{\Theta}}$ are the posterior and prior ESS of Section~\ref{sec:ess}. We extend the criterion to the multiparametric precision matrix $\bm{\Theta}$ and verify it for the five methods of Section~\ref{subsec:ess_methods} under both the Wishart and \textit{G}-Wishart priors, with the elicitation convention $\bm{\Psi}^{*}=\bm{\Psi}/\nu$. Both priors are conjugate, so the posterior degrees of freedom equal $\nu+n$ (Appendix~\ref{appendix:priors}).

\subsection*{Exact consistency: MTM, PT, and ELIR}

The prior ESS under MTM, PT, and ELIR is linear in $\nu$ and invariant to $\bm{\Psi}$ and the graph $G$ (Section~\ref{subsec:ess_methods}). This linearity alone yields exact consistency.

\begin{lemma}[Linear degrees-of-freedom rule implies consistency]
\label{lem:linear_pcc}
Suppose a method assigns prior ESS $ESS_{\bm{\Theta}}=\nu-c$ for a constant $c$ independent of $\nu$ and $n$, and that applying the method to the posterior gives the same rule at the posterior degrees of freedom,
$ESS_{\bm{\Theta}\mid\mathbf{X}}=(\nu+n)-c$. Then $ESS_{\bm{\Theta}\mid\mathbf{X}}$ is non-random and \eqref{eq:pcc} holds exactly.
\end{lemma}

\begin{proof}
Since the posterior degrees of freedom equal $\nu+n$ irrespective of the realized data, so $ESS_{\bm{\Theta}\mid\mathbf{X}}=(\nu+n)-c$ is constant in both $\mathbf{X}$ and $\bm{\Theta}$. The two expectations in \eqref{eq:pcc} become trivial and $(\nu+n)-c=(\nu-c)+n=ESS_{\bm{\Theta}}+n$.
\end{proof}

\begin{proposition}
\label{prop:exact_pcc}
Under both the Wishart and \textit{G}-Wishart priors, MTM, PT, and ELIR satisfy \eqref{eq:pcc} exactly.
\end{proposition}

\begin{proof}
Each method has the linear form of Lemma~\ref{lem:linear_pcc} with the same value under both priors (Section~\ref{subsec:ess_methods}): $c=1$ for MTM and $c=p+1$ for PT and ELIR. Since $p$ is unchanged under updating, the posterior ESS is the same rule evaluated at $\nu+n$.
\end{proof}

\begin{table}[t]
\centering
\footnotesize
\begin{tabular}[t]{lrrrr}
\toprule
\multicolumn{1}{c}{ } & \multicolumn{4}{c}{$\nu$} \\
\cmidrule(l{3pt}r{3pt}){2-5}
$n$ & 25 & 50 & 100 & 200\\
\midrule
\multicolumn{5}{c}{\textit{Wishart}}\\
\multicolumn{5}{l}{\textbf{VR}}\\
\hspace{1em}25 & 1.0004  & 1.0002  & 1.0001  & 1.0000 \\
\hspace{1em}50 & 1.0003  & 1.0001  & 1.0001  & 1.0000 \\
\hspace{1em}100 & 1.0002  & 1.0001  & 1.0000  & 1.0000 \\
\hspace{1em}200 & 1.0001  & 1.0001  & 1.0000  & 1.0000 \\
\addlinespace[0.3em]
\multicolumn{5}{l}{\textbf{PR}}\\
\hspace{1em}25 & 0.6305  & 0.8767  & 0.9630  & 0.9897 \\
\hspace{1em}50 & 0.7536  & 0.9075  & 0.9691  & 0.9907 \\
\hspace{1em}100 & 0.8521  & 0.9383  & 0.9769  & 0.9923 \\
\hspace{1em}200 & 0.9178  & 0.9630  & 0.9846  & 0.9942 \\
\midrule
\multicolumn{5}{c}{\textit{G-Wishart (random graph structure)}}\\
\multicolumn{5}{l}{\textbf{VR}}\\
\hspace{1em}25 & 2.5253 (0.7252) & 1.3879 (0.2407) & 1.2500 (0.1170) & 1.1555 (0.0803)\\
\hspace{1em}50 & 2.2791 (0.6965) & 1.3953 (0.2379) & 1.2210 (0.1300) & 1.1558 (0.0837)\\
\hspace{1em}100 & 2.2072 (0.6689) & 1.4011 (0.2173) & 1.2293 (0.1055) & 1.1507 (0.0868)\\
\hspace{1em}200 & 2.3238 (0.6319) & 1.3825 (0.1985) & 1.2325 (0.1158) & 1.1559 (0.0914)\\
\addlinespace[0.3em]
\multicolumn{5}{l}{\textbf{PR}}\\
\hspace{1em}25 & 0.8347 (0.1075) & 0.9830 (0.0384) & 1.0106 (0.0137) & 1.0103 (0.0127)\\
\hspace{1em}50 & 0.8926 (0.0783) & 0.9892 (0.0302) & 1.0097 (0.0096) & 1.0108 (0.0076)\\
\hspace{1em}100 & 0.9581 (0.0406) & 0.9995 (0.0168) & 1.0093 (0.0076) & 1.0102 (0.0052)\\
\hspace{1em}200 & 0.9818 (0.0196) & 1.0005 (0.0112) & 1.0089 (0.0046) & 1.0097 (0.0042)\\
\bottomrule
\end{tabular}
\caption{Mean (sd) of $\delta_{VR}$ and $\delta_{PR}$ by sample size $n$ (rows) and prior degrees of freedom $\nu$ (columns), for the Wishart and $G$-Wishart priors and $p=20$. The standard deviation for the Wishart results is not reported, as it was $<10^{-11}$ for all combinations.}
\label{tab:pcc_check_VR_PR}
\end{table}

\subsection*{Consistency of VR and PR: simulation check}

VR and PR lack the linear degrees-of-freedom structure of Lemma~\ref{lem:linear_pcc}: each prior ESS has the form $ESS_{\bm{\Theta}}=g(\nu,p)\,h(\nu,\bm{\Psi}_G)$, where the correction $h(\nu,\bm{\Psi}_G)$ depends on $\bm{\Psi}$, $G$, and $p$ but not on $n$, and tends to $1$ only as $\nu\to\infty$. Predictive consistency \eqref{eq:pcc} therefore cannot reduce to an algebraic identity in $\nu$ and $n$, and under the \textit{G}-Wishart prior the relevant moments have no closed form. We assess it by Monte Carlo simulation through the normalized increment $(ESS_{\bm{\Theta}\mid\mathbf{X}}-ESS_{\bm{\Theta}})/n$, whose value equals $1$ whenever predictive consistency holds exactly (Lemma~\ref{lem:linear_pcc}).

We fix $p=20$ and consider a grid of prior and prospective study sizes $\nu\in\{25,50,100,200\}$ and $n\in\{25,50,100,200\}$, for both the Wishart and the \textit{G}-Wishart prior. Under the \textit{G}-Wishart prior, each prior study is assigned an Erd\H{o}s--R\'enyi graph $G^{(s)}(p,\pi)$ with edge probability $\pi\sim\mathrm{Uniform}(0,1)$, so that the studies span graph densities across $[0,1]$. For each $(\nu,n)$ cell we generate $n_{\text{studies}} = 100$ prior studies, each represented by a precision matrix $\bm{\Psi}^{(s)}$ from which the prior $W_{G^{(s)}}\big(\nu,\bm{\Psi}^{(s)}/\nu\big)$ is elicited as in Section~\ref{appendix:priors} and its prior ESS, $ESS^{(s)}_{\bm{\Theta}}$, computed.

For each prior study $s$ we then simulate $B=100$ datasets. For $b=1,\dots,B$: (i)~draw $\bm{\Theta}^{(s,b)}\sim W_{G^{(s)}} \big(\nu,\bm{\Psi}^{(s)}/\nu\big)$; (ii) draw $\mathbf{X}^{(s,b)}_n$ as $n$ i.i.d. observations from $\mathcal{N} \big(\mathbf{0},(\bm{\Theta}^{(s,b)})^{-1}\big)$; (iii) form the posterior $W_{G^{(s)}} \left(\nu+n,\,\left({\bm{\Psi}^{(s)*}}^{-1}+nS^{(s,b)}\right)^{-1}\right)$ (Appendix~\ref{appendix:priors}), with $\bm{\Psi}^{(s)*}=\bm{\Psi}^{(s)}/\nu$ and $S^{(s,b)}$ the sample covariance, and compute its ESS, $ESS^{(s,b)}_{\bm{\Theta}\mid\mathbf{X}}$, with an inner Monte Carlo step for the VR/PR moments under the \textit{G}-Wishart prior. We record the normalized increment $\delta^{(s,b)}=(ESS^{(s,b)}_{\bm{\Theta}\mid\mathbf{X}}-ESS^{(s)}_{\bm{\Theta}})/n$ and average it over the $B$ datasets, $\bar{\delta}^{(s)}=\sum_{b}\delta^{(s,b)}/B$. Across the $n_{\text{studies}}$ prior studies we report the mean $\overline{\delta}=\sum_{s}\bar{\delta}^{(s)}/n_{\text{studies}}$ and its standard deviation, summarizing how close the increment is to $1$ and how much it varies with the elicited prior.

Table \ref{tab:pcc_check_VR_PR} shows that across both priors and both methods the mean increment approaches $1$ --- VR from above, PR from below --- and its standard deviation shrinks as $\nu$ increases, confirming that VR and PR are predictively consistent as the prior information increases. The residual departure from $1$ at small $\nu$, and its variability across studies, reflects the dependence of $h(\nu,\bm{\Psi}_G)$ on the elicited precision and graph.

\clearpage
\section{Precision Matrix for the Illustrative Example}
\label{appendix:illustrative_precision_matrix}
The following precision matrix $\bm{\Psi}$ and the corresponding partial correlation matrix $\bm{P}$ were used for the illustrative  example in Section \ref{sec:sample_size_planning}, with $p = 10$ and  $\nu = 100$. The partial correlation corresponding to the planning edge used in the example is highlighted in bold.   

{
\footnotesize
\begin{equation*}
\bm{\Psi} = \begin{pmatrix}
 1.910 & \cdot & 0.345 & 0.290 & 0.146 & -0.860 & -0.323 & -0.254 & \cdot & -0.364 \\
 \cdot & 0.577 & \cdot & -0.003 & \cdot & -0.205 & -0.239 & 0.114 & -0.015 & \cdot \\
 0.345 & \cdot & 0.788 & \cdot & -0.167 & -0.403 & 0.004 & 0.171 & -0.140 & -0.420 \\
 0.290 & -0.003 & \cdot & 0.693 & -0.247 & 0.049 & -0.139 & \cdot & \cdot & \cdot \\
 0.146 & \cdot & -0.167 & -0.247 & 0.636 & -0.297 & \cdot & -0.204 & 0.141 & \cdot \\
-0.860 & -0.205 & -0.403 & 0.049 & -0.297 & 1.658 & 0.066 & 0.329 & \cdot & \cdot \\
-0.323 & -0.239 & 0.004 & -0.139 & \cdot & 0.066 & 0.687 & 0.081 & \cdot & 0.374 \\
-0.254 & 0.114 & 0.171 & \cdot & -0.204 & 0.329 & 0.081 & 1.265 & -0.183 & -0.738 \\
 \cdot & -0.015 & -0.140 & \cdot & 0.141 & \cdot & \cdot & -0.183 & 0.874 & \cdot \\
-0.364 & \cdot & -0.420 & \cdot & \cdot & \cdot & 0.374 & -0.738 & \cdot & 1.929
\end{pmatrix}
\end{equation*}
}

{
\footnotesize
\begin{equation*}
\bm{P} = \begin{pmatrix}
 \cdot & \cdot & -0.281 & -0.252 & -0.133 & 0.483 & 0.282 & 0.163 & \cdot & 0.189 \\
 \cdot & \cdot & \cdot & 0.005 & \cdot & \mathbf{0.210} & 0.380 & -0.133 & 0.021 & \cdot \\
-0.281 & \cdot & \cdot & \cdot & 0.235 & 0.353 & -0.005 & -0.171 & 0.168 & 0.340 \\
-0.252 & 0.005 & \cdot & \cdot & 0.372 & -0.046 & 0.201 & \cdot & \cdot & \cdot \\
-0.133 & \cdot & 0.235 & 0.372 & \cdot & 0.289 & \cdot & 0.228 & -0.189 & \cdot \\
 0.483 & \mathbf{0.210} & 0.353 & -0.046 & 0.289 & \cdot & -0.061 & -0.227 & \cdot & \cdot \\
 0.282 & 0.380 & -0.005 & 0.201 & \cdot & -0.061 & \cdot & -0.087 & \cdot & -0.325 \\
 0.163 & -0.133 & -0.171 & \cdot & 0.228 & -0.227 & -0.087 & \cdot & 0.174 & 0.473 \\
 \cdot & 0.021 & 0.168 & \cdot & -0.189 & \cdot & \cdot & 0.174 & \cdot & \cdot \\
 0.189 & \cdot & 0.340 & \cdot & \cdot & \cdot & -0.325 & 0.473 & \cdot & \cdot
\end{pmatrix}
\end{equation*}
}

\section{DPIR Algorithm}
\label{appendix:dpir_algorithm}
\begin{algorithm}[H]
\small
\SetAlgoSkip{smallskip}
\caption{DPIR: Computation of $\Pr(\Lambda_n > \xi)$ for a fixed $n$}
\label{algorithm:dpir}

\KwIn{Prior parameters $(\nu, \bm{\Psi})$, graph $G$, 
      sample size $n$, number of prior draws $H$, 
      number of data draws $J$, duplication matrix $\mathbf{D}$, $d = p(p+1)/2$\;}
\KwOut{$\Pr(\Lambda_n > \xi)\;$}

Initialize counter $c \leftarrow 0$\;

\For{$h = 1, \ldots, H$}{
    Draw $\bm{\Theta}_h \sim \pi(\bm{\Theta})$
    \hspace*{1em}\textit{where $\pi(\bm{\Theta}) = W(\nu, \bm{\Psi})$ 
    if $G$ is complete, and $W_G(\nu, \bm{\Psi})$ otherwise}\;
        
    Compute $\bm{\Sigma}_h = \bm{\Theta}_h^{-1}$\;
    
    Compute $I(\bm{\Theta}_h) = \dfrac{\nu-2}{2}
    \mathbf{D}(\bm{\Sigma}_h \otimes \bm{\Sigma}_h)\mathbf{D}^{\top}$
    restricted to the free parameters of $G$\;
    
    \For{$j = 1, \ldots, J$}{
        Draw $\mathbf{X}_j = (\mathbf{x}_1, \ldots, \mathbf{x}_n)$, 
        $\quad \mathbf{x}_k \sim \mathcal{N}_p(\mathbf{0}, \bm{\Sigma}_h)$\;
        
        Compute $S_j = \dfrac{1}{n}\mathbf{X}_j^{\top}\mathbf{X}_j$\;
    
        Compute $\mathcal{I}(\mathbf{X}_j) = \dfrac{n}{2}
        \mathbf{D}(S_j \otimes S_j)\mathbf{D}^{\top}$\;
        restricted to the free parameters of $G$\;
        
        Compute $\Lambda_n(\bm{\Theta}_h, \mathbf{X}_j) = 
        \left(\dfrac{\lvert \mathcal{I}(\mathbf{X}_j)\rvert}
        {\lvert I(\bm{\Theta}_h)\rvert}\right)^{1/d}$\;
        
        \If{$\Lambda_n(\bm{\Theta}_h, \mathbf{X}_j) > \xi$}{
            $c \leftarrow c + 1$\;
        }
    }
}
\Return $\Pr(\Lambda_n > \xi) = \dfrac{c}{H \times J}$\;
\end{algorithm}

\section{BFDA Algorithm}
\label{appendix:bfda_algorithm}
\begin{algorithm}[H]
\footnotesize
\SetAlgoSkip{smallskip}
\caption{BFDA: Computation of performance measures for a fixed $n$ and edge $(i,j)$}
\label{algorithm:bfda}
\KwIn{Prior $\pi(\bm{\Theta} \mid \mathcal{H})$, graph $G$, edge $(i,j)$ such that $(i,j)\in E$,
      sample size $n$, number of prior draws $H$,
      number of replications $J$,
      Bayes factor threshold $\gamma = 10$\;}
\KwOut{$\beta_0$, $\beta_1$, FPR, FNR\;}
\BlankLine
Set $G_0 \leftarrow G$ with $G_0[i,j] = G_0[j,i] = 0$\;
\For{$\mathcal{H} \in \{\mathcal{H}_0, \mathcal{H}_1\}$}{
    \lIf{$\mathcal{H} = \mathcal{H}_0$}{set $G_{\mathcal{H}} \leftarrow G_0$}
    \lElse{set $G_{\mathcal{H}} \leftarrow G$}
    \For{$h = 1, \ldots, H$}{
        Draw $\bm{\Theta}_h \sim \pi(\bm{\Theta} \mid \mathcal{H})$\;
        Compute $\bm{\Sigma}_h = \bm{\Theta}_h^{-1}$\;
        \For{$j = 1, \ldots, J$}{
            Draw $\mathbf{X}_j=\left(\mathbf{x}_1,\ldots,\mathbf{x}_n\right)$, 
            $\quad \mathbf{x}_k \sim \mathcal{N}_p(\mathbf{0}, \bm{\Sigma}_h)$\;
            Compute $BF_{01}^{(h,j)}$ from $\mathbf{X}_j$ and $\pi(\bm{\Theta} \mid \mathcal{H})$\;
        }
    }
    \If{$\mathcal{H} = \mathcal{H}_0$}{
        Compute $\beta_0 = \dfrac{1}{H\times J}\displaystyle\sum_{h,j}
        \mathbf{1}(BF_{01}^{(h,j)} > \gamma)$\;
        Compute $\text{FNR} = \dfrac{1}{H\times J}\displaystyle\sum_{h,j}
        \mathbf{1}(BF_{01}^{(h,j)} < 1/\gamma)$\;
    }
    \Else{
        Compute $\beta_1 = \dfrac{1}{H\times J}\displaystyle\sum_{h,j}
        \mathbf{1}(BF_{01}^{(h,j)} < 1/\gamma)$\;
        Compute $\text{FPR} = \dfrac{1}{H\times J}\displaystyle\sum_{h,j}
        \mathbf{1}(BF_{01}^{(h,j)} > \gamma)$\;
    }
}
\BlankLine
\Return $\beta_0,\ \beta_1,\ \text{FPR},\ \text{FNR}$\;
\end{algorithm}

\section{Algorithm for Generating Random Prior Precision Matrix}
\label{appendix:generate_random_prior_precision_matrix}
\begin{algorithm}[H]
\caption{Algorithm for generating random prior precision matrix}
\SetAlgoSkip{smallskip}
\SetCommentSty{textit}
\label{algorithm:random_precision}
\KwIn{$p$ (number of variables), $\nu$ (prior sample size), 
      $G$ (adjacency matrix; empty or fully connected for the Wishart case)}
\KwOut{$\bm{\Psi}$ (estimated prior precision matrix)}

\tcp{Step 1: Define a well-conditioned base precision matrix}
Draw a $p \times p$ random normal matrix, $\mathbf{M} = \left(\mathbf{m}_1, \ldots, 
\mathbf{m}_p\right)$, $\quad \mathbf{m}_k \sim \mathcal{N}_p(\mathbf{0}, \mathbf{I}_p)$\;
Compute QR decomposition, $\mathbf{M} = \mathbf{Q}\mathbf{R}$\;
Draw $\sigma_i \sim U(0.1, 5)$, $i = 1, \ldots, p$\;
$\bm{\Psi}_{0} \leftarrow \mathbf{Q}^{\top} \, \text{diag}(1/\sigma_1, \ldots, 1/\sigma_p) \, \mathbf{Q}$\;

\BlankLine

\tcp{Step 2: Simulate $\nu$ observations and estimate precision matrix}
\eIf{$G$ is empty \textbf{or} fully connected}
    {
        \tcp{Wishart case: dense precision matrix}
        Simulate $\nu$ observations, $\mathbf{X}=\left(\mathbf{x}_1,\ldots,\mathbf{x}_\nu\right)$, $\quad \mathbf{x}_k \sim \mathcal{N}_p(\mathbf{0}, \bm{\Psi}_{0}^{-1})$\; 
        $\hat{\Sigma} \leftarrow \mathbf{X}\mathbf{X}^{\top} / \nu$\;
        $\bm{\Psi} \leftarrow \hat{\Sigma}^{-1}$\;
    }
    {
        \tcp{\textit{G}-Wishart case: structured precision matrix}
        $\bm{\Psi}_{0} \leftarrow \textsc{Projection}(\bm{\Psi}_{0}, G)$ \;
        Simulate $\nu$ observations, $\mathbf{X}=\left(\mathbf{x}_1,\ldots,\mathbf{x}_\nu\right)$, $\quad \mathbf{x}_k \sim \mathcal{N}_p(\mathbf{0}, \bm{\Psi}_{0}^{-1})$\; 
        $\hat{\Sigma} \leftarrow \mathbf{X}\mathbf{X}^{\top} / \nu$\; 
        $\bm{\Psi} \leftarrow \text{\textsc{Projection}}\;(\hat{\Sigma}^{-1}, G)$\;
    }
\KwRet{$\bm{\Psi}$}
\end{algorithm}

\section{Monte Carlo Estimation of the \textit{G}-Wishart Bayes Factor}
\label{appendix:ak_bf}

Throughout this appendix, we follow the original notation of \citet{Atay2005} and use the parametrization $\delta = \nu - p +1$. The Bayes factor in \eqref{eq:bf_normalizing_constant} involves four \textit{G}-Wishart normalizing constants. Following \citet{Atay2005}, let $\mathcal{V}$ denote the set of free pairs in the upper triangle of the precision matrix, defined as
\begin{equation}
\mathcal{V} = \{(i,j) \ , \ i \leq j \ : \ i = j\ \text{or}\ (i,j) \in E\},
\label{eq:ak_V}
\end{equation}
that is, the diagonal entries together with the upper-triangular pairs corresponding to present edges in $G$. The complement $\overline{\mathcal{V}} = \{(i,j) : i < j,\ (i,j) \notin E\}$ denotes the set of absent edges, whose corresponding Cholesky entries are non-free and determined by the free ones via the graph constraint. Each normalizing constant $I_G(\delta, \nu\bm{\Psi}^{-1})$ decomposes as
\begin{equation}
\log I_G(\delta, \nu\bm{\Psi}^{-1}) = \log C_{\delta,T} + 
\log \mathbb{E}\left[f_T(\psi_\mathcal{V})\right],
\label{eq:ak_decomposition}
\end{equation}
where $T$ is the upper-triangular Cholesky factor of $\nu\bm{\Psi}^{-1} = T^\top T$, $C_{\delta,T}$ is a closed-form constant, and $\mathbb{E}\left[f_T(\psi_\mathcal{V})\right]$ is an expectation over free variables $\psi_\mathcal{V}$ with a known product distribution that can be read off the graph $G$.

\paragraph{Log Bayes factor decomposition.} Since \eqref{eq:bf_normalizing_constant} involves differences of log normalizing constants between $G$ and $G_{-(i,j)}$, it is convenient to work directly with the differences
\begin{align}
\Delta\log C &= \log C_{\delta,T}^{G_{-(i,j)}} - \log C_{\delta,T}^{G}, 
\label{eq:delta_log_C}\\
\Delta\mathbb{E} &= \log \mathbb{E}_{G_{-(i,j)}}\left[f_T(\psi_\mathcal{V})\right] - 
\log\mathbb{E}_G\left[f_T(\psi_\mathcal{V})\right].
\label{eq:Delta_E}
\end{align}
The log Bayes factor then reduces to
\begin{equation}
\log BF_{01} = (\Delta\log C_{\text{post}} + \Delta \mathbb{E}_{\text{post}}) - 
(\Delta\log C_{\text{prior}} + \Delta \mathbb{E}_{\text{prior}}),
\label{eq:ak_logbf}
\end{equation}
where the subscripts ``prior'' and ``post'' refer to evaluation at 
$(\delta, \nu\bm{\Psi}^{-1})$ and $(\delta+n, \nu\bm{\Psi}^{-1} + 
\mathbf{X}^\top\mathbf{X})$ respectively, with Cholesky factors $T$ 
such that $T^\top T = \nu\bm{\Psi}^{-1}$ and $T^*$ such that $T^{*\top}T^* = 
\nu\bm{\Psi}^{-1} + \mathbf{X}^\top\mathbf{X}$.

\paragraph{Closed-form difference.}
The graphs $G$ and $G_{-(i,j)}$ differ only in the presence of edge $(i,j)$, with $i < j$. Removing this edge affects two quantities: $n_i$, the number of upper-triangular neighbors of node $i$, decreases by one, and $k_j$, the number of present edges in the upper triangle with column index $j$, also decreases by one. All other terms in $C_{\delta,T}$ are unchanged. The resulting difference is
\begin{equation}
\Delta\log C = -\frac{1}{2}\log(4\pi) + 
\log\Gamma\left(\frac{\delta + n_i - 1}{2}\right) - 
\log\Gamma\left(\frac{\delta + n_i}{2}\right) - 
\log T_{ii} - \log T_{jj},
\label{eq:Delta_log_C_formula}
\end{equation}
where $n_i = |\mathrm{ne}(i) \cap \{i+1,\ldots,p\}|$ is the number of upper-triangular neighbors of node $i$ in $G$, $k_j = |\mathrm{ne}(j) \cap \{1,\ldots,j-1\}|$ is the number of present edges with column index $j$, and $T_{ii}$, $T_{jj}$ are the corresponding diagonal entries of $T$.

\paragraph{Monte Carlo difference.}
The two expectations in \eqref{eq:Delta_E} require separate Monte Carlo runs because the free variable distributions differ between $G$ and $G_{-(i,j)}$. Under $G$, the free variables are drawn as
\begin{align}
\psi_{kk} & \sim \sqrt{\chi^2_{\delta+n_k}}, \quad k = 1,\ldots,p, 
\label{eq:ak_diag}\\
\psi_{kl} & \sim \mathcal{N}(0,1), \quad (k,l) \in E,\ k < l,
\label{eq:ak_offdiag}
\end{align}
including $\psi_{ij} \sim \mathcal{N}(0,1)$ since $(i,j) \in E$. Under $G_{-(i,j)}$, edge $(i,j)$ is removed, so $\psi_{ij}$ becomes a non-free parameter and $\psi_{ii} \sim \sqrt{\chi^2_{\delta+n_i-1}}$ since the upper-triangular degree of node $i$ decreases by one. The non-free elements under either graph are deterministic functions of the free elements, computed recursively via \citet[][Eq.~31--32]{Atay2005}. For $S$ Monte Carlo draws, the estimate of each expectation is
\begin{equation}
\mathbb{E}\left[f_T(\psi_\mathcal{V})\right] \approx \frac{1}{S}\sum_{s=1}^{S} 
\exp\left(-\frac{1}{2}\sum_{(k,l)\in\overline{\mathcal{V}}} \left(\psi_{kl}^{(s)}\right)^{2}
\right),
\label{eq:ak_mc}
\end{equation}
where $\overline{\mathcal{V}}$ denotes the set of non-free pairs corresponding 
to absent edges under the respective graph. The four Monte Carlo estimates in \eqref{eq:ak_logbf} are computed independently, each using $S$ draws from \eqref{eq:ak_diag}--\eqref{eq:ak_offdiag} with degrees of freedom $\delta$ or $\delta+n$ and Cholesky factor $T$ or $T^*$, under graph $G$ or $G_{-(i,j)}$.

\end{appendices}

\bibliography{bibliography}

\end{document}